\RequirePackage{afterpackage}
\AfterPackage{amsthm}{
  \RequirePackage{hyperref,cleveref}
}

\documentclass[a4paper,USenglish]{lipics-v2021}

\usepackage{mathtools,amsthm,amsmath}
\usepackage{pdfpages}
\usepackage{forloop}

\usepackage[square,sort,comma,numbers]{natbib}
\usepackage{thmtools}
\usepackage[utf8]{inputenc}
\usepackage{url}

\usepackage{multirow}
\usepackage{threeparttable}
\usepackage{siunitx}
\usepackage{booktabs}

\usepackage{tabularx}
\usepackage{cleveref}
\usepackage{mdframed}
\usepackage{bbm}
\usepackage{bm}
\usepackage{microtype}
\usepackage{xcolor}
\usepackage{makecell}
\usepackage{mathtools}
\usepackage{float}
\usepackage[lined,ruled,linesnumbered,noend]{algorithm2e}

\usepackage[inline]{enumitem}

\usepackage{subfiles}
\usepackage{standalone}
\usepackage[utf8]{inputenc}

\usepackage{url}
\usepackage{empheq}
\usepackage[noend]{algpseudocode}
\usepackage{graphicx}
\usepackage{caption}
\usepackage{subcaption}
\usepackage{todonotes}

\usepackage{url}
\usepackage[noend]{algpseudocode}
\algrenewcommand\algorithmicloop{\textbf{repeat}}

\usepackage{tcolorbox}
\tcbuselibrary{skins,breakable}
\tcbset{enhanced jigsaw}

\usepackage{subfiles}
\usepackage{subcaption}

\usepackage{comment}
\usepackage{natbib}
\usepackage{bibentry}
\nobibliography*

\newcounter{ct}
\newcommand{\includepdfpages}[1]{
  \forloop{ct}{1}{\value{ct} < 50}{%
    \ifodd\value{ct}%
      \includepdf[pages=\thect, offset=-55 0, pagecommand={\thispagestyle{empty}}]{#1}
    \else
      \includepdf[pages=\thect, offset=75 0, pagecommand={\thispagestyle{empty}}]{#1}
    \fi
  }
}

\algrenewcommand\algorithmicdo{\textbf{}}
\algrenewcommand\algorithmicend{\textbf{}}
\algnewcommand\algorithmicforeach{\textbf{for each}}
\algdef{S}[FOR]{ForEach}[1]{\algorithmicforeach\ #1\ \algorithmicdo}

\newcommand{\etal}{\textit{et al.}\@}

\usepackage{xspace}

\makeatletter
\renewcommand{\paragraph}{%
  \@startsection{paragraph}{4}{\z@}%
    {0ex \@plus1ex \@minus .2ex}%
    {-1em}%
    {\normalfont\normalsize\bfseries}}
\makeatother

\let\originalleft\left
\let\originalright\right
\renewcommand{\left}{\mathopen{}\mathclose\bgroup\originalleft}
\renewcommand{\right}{\aftergroup\egroup\originalright}

\def\Pr{\mathbf{Pr}}

\newcommand{\defn}[1]{\emph{\textbf{#1}}}
\newcommand{\id}[1]        {\ifmmode\mathit{#1}\else\textit{#1}\fi}

\DeclareMathOperator*{\E}{\mathbb{E}}

\newcommand{\dist}{\id{dist}}

\newcommand{\Ot}{\ensuremath{\widetilde{O}}}

\newcommand{\paren}[1]{\ensuremath{\left(#1\right)}\xspace}

\renewcommand{\hat}{\widehat}
\renewcommand{\tilde}{\widetilde}

\newcommand{\TODO}[1]{\typeout{TODO: \the\inputlineno: #1}\textbf{{\color{red}[[[ #1 ]]]}}}
\newcommand{\tOh}[1]{\tilde{O}\left(#1\right)}
\newcommand{\tOm}[1]{\tilde{\Omega}\left(#1\right)}
\newcommand{\DM}{Distributed Minor-Aggregation Model\xspace}
\newcommand{\ingball}{\operatorname{Ball_G^{in}}}
\newcommand{\outgball}{\operatorname{Ball_G^{out}}}

\newcommand{\sgball}{\operatorname{Ball_G^{*}}}

\newcommand{\Oraclem}{\Oracle_{S}}

\newenvironment{tbox}{\begin{tcolorbox}[
		enlarge top by=5pt,
		enlarge bottom by=5pt,
		 breakable,
		 boxsep=2pt,
                  left=5pt,
                  right=7pt,
                  top=10pt,
                  arc=0pt,
                  boxrule=1pt,toprule=1pt,
                  colback=white
                  ]%
	}
{\end{tcolorbox}}

\newcommand{\congest}{$\mathsf{CONGEST}$\xspace}

\newcommand{\scaledown}{\textsc{ScaleDown}\xspace}
\newcommand{\spmain}{\textsc{SPMain}\xspace}
\newcommand{\Oracle}{\mathcal{O}^{NN-SSSP}}

\newcommand{\SccTop}{\textsc{SCC+Topsort}\xspace}
\newcommand{\FixDag}{\textsc{FixDAGEdges}\xspace}
\newcommand{\EstDist}{\textsc{EstDist}\xspace}
\newcommand{\ElimNeg}{\textsc{ElimNeg}\xspace}
\newcommand{\LDD}{\textsc{LowDiameterDecomposition}\xspace}

\newcommand{\Vout}{\ensuremath{V_{out}}\xspace}
\newcommand{\Vin}{\ensuremath{V_{in}}\xspace}
\newcommand{\Vheavy}{\ensuremath{V_{heavy}}\xspace}
\newcommand{\Ain}{\ensuremath{A_{in}}\xspace}
\newcommand{\Aout}{\ensuremath{A_{out}}\xspace}
\newcommand{\outcut}{\ensuremath{\delta^{+}}\xspace}
\newcommand{\incut}{\ensuremath{\delta^{-}}\xspace}
\newcommand{\Eref}{\ensuremath{E^{rem}}\xspace}
\newcommand{\Erem}{\ensuremath{E^{rem}}\xspace}
\newcommand{\Eneg}{\ensuremath{E^{neg}}\xspace}
\newcommand{\EremI}{\ensuremath{E^{rem}_1}\xspace}
\newcommand{\EremII}{\ensuremath{E^{rem}_2}\xspace}
\newcommand{\Eremin}{\ensuremath{E^{rem}_{in}}\xspace}
\newcommand{\Eremout}{\ensuremath{E^{rem}_{out}}\xspace}
\newcommand{\Erems}{\ensuremath{E^{rem}_{*}}\xspace}

\theoremstyle{definition}

\SetKwInput{KwData}{Input}
\SetKwInput{KwResult}{Output}
\SetKwInput{KwOracle}{Oracle}

\title{Parallel, Distributed, and Quantum Exact Single-Source Shortest Paths with Negative Edge Weights}
\author{Vikrant Ashvinkumar}{Rutgers University, USA}{}{}{}
\author{Aaron Bernstein}{Rutgers University, USA}{}{}{}
\author{Nairen Cao}{Department of Computer Science, Boston College, USA}{}{}{}
\author{Christoph Grunau}{ETH Zürich, Switzerland}{}{}{}
\author{Bernhard Haeupler}{ETH Zürich, Switzerland}{}{}{}
\author{Yonggang Jiang}{Max Planck Institute for Informatics, Saarland Informatics Campus}{}{}{}
\author{Danupon Nanongkai}{Max Planck Institute for Informatics, Saarland Informatics Campus}{}{}{}
\author{Hsin-Hao Su}{Department of Computer Science, Boston College, USA}{}{}{}

\relatedversion{}\relatedversiondetails{Full Version}{https://arxiv.org/abs/2303.00811}
\authorrunning{V. Ashvinkumar et al.}
\Copyright{Vikrant Ashvinkumar, Aaron Bernstein, Nairen Cao, Christoph Grunau, Bernhard Haeupler, Yonggang Jiang, Danupon Nanongkai, and Hsin-Hao Su}


\EventEditors{Timothy Chan, Johannes Fischer, John Iacono, and Grzegorz Herman}
\EventNoEds{4}
\EventLongTitle{32nd Annual European Symposium on Algorithms (ESA 2024)}
\EventShortTitle{ESA 2024}
\EventAcronym{ESA}
\EventYear{2024}
\EventDate{September 2--4, 2024}
\EventLocation{Royal Holloway, London, United Kingdom}
\EventLogo{}
\SeriesVolume{308}
\ArticleNo{25}


\begin{document}
\pagenumbering{roman}
\nolinenumbers
\maketitle

\begin{abstract}
This paper presents parallel, distributed, and quantum algorithms for single-source shortest paths when edges can have negative integer weights (negative-weight SSSP).
We show a framework that reduces negative-weight SSSP in all these settings to $n^{o(1)}$ calls to any SSSP algorithm that works on inputs with non-negative integer edge weights (non-negative-weight SSSP) with a virtual source.
More specifically, for a directed graph with $m$ edges, $n$ vertices, undirected hop-diameter $D$, and polynomially bounded integer edge weights, we show randomized algorithms for negative-weight SSSP with
\begin{itemize}
    \item $W_{SSSP}(m,n)n^{o(1)}$ work and $S_{SSSP}(m,n)n^{o(1)}$ span, given access to a non-negative-weight SSSP algorithm with $W_{SSSP}(m,n)$ work and $S_{SSSP}(m,n)$ span in the parallel model, and
    \item $T_{SSSP}(n,D)n^{o(1)}$ rounds, given access to a non-negative-weight SSSP algorithm that takes $T_{SSSP}(n,D)$ rounds in CONGEST, and
    \item $Q_{SSSP}(m,n)n^{o(1)}$ quantum edge queries, given access to a non-negative-weight SSSP algorithm that takes $Q_{SSSP}(m,n)$ queries in the quantum edge query model.
\end{itemize}
This work builds off the recent result of Bernstein, Nanongkai, Wulff-Nilsen~\cite{bernstein2022negative}, which gives a near-linear time algorithm for negative-weight SSSP in the sequential setting.

Using current state-of-the-art non-negative-weight SSSP algorithms yields randomized algorithms for negative-weight SSSP with
\begin{itemize}
    \item $m^{1+o(1)}$ work and $n^{1/2+o(1)}$ span in the parallel model, and
    \item $(n^{2/5}D^{2/5} + \sqrt{n} + D)n^{o(1)}$ rounds in \congest, and
    \item $m^{1/2}n^{1/2+o(1)}$ quantum queries to the adjacency list or $n^{1.5+o(1)}$ quantum queries to the adjacency matrix.
\end{itemize}
Up to a $n^{o(1)}$ factor, the parallel and distributed results match the current best upper bounds for reachability~\cite{jambulapati2019parallel,distributedhopsets}.
Consequently, any improvement to negative-weight SSSP in these models beyond the $n^{o(1)}$ factor necessitates an improvement to the current best bounds for reachability.
The quantum result matches the lower bound up to an $n^{o(1)}$ factor~\cite{BerzinaDFLS04}.

Our main technical contribution is an efficient reduction from computing a low-diameter decomposition (LDD) of directed graphs to computations of non-negative-weight SSSP with a virtual source.
Efficiently computing an LDD has heretofore only been known for undirected graphs in both the parallel and distributed models, and been rather unstudied in quantum models.
The directed LDD is a crucial step of the sequential algorithm in \cite{bernstein2022negative}, and we think that its applications to other problems in parallel and distributed models are far from being exhausted.

Other ingredients of our results include altering the recursion structure of the scaling algorithm in \cite{bernstein2022negative} to surmount difficulties that arise in these models, and also an efficient reduction from computing strongly connected components to computations of SSSP with a virtual source in CONGEST. 
The latter result answers a question posed in~\cite{bernstein2019distributed} in the negative.

\end{abstract}

\tableofcontents

\newpage
\pagenumbering{arabic}
\setcounter{page}{1}
\section{Introduction}

Single-source shortest paths (SSSP) is one of the most fundamental problems in graph algorithms.
Given a directed graph $G=(V,E)$, an integer weight function $w: E \rightarrow \mathbb{Z}$, and a source vertex $s \in V$, we want
 to compute the distance from $s$ to $v$
 for all $v \in V$.

Efficient solutions to this problem are typically better understood in the regime where edge weights are non-negative, which we denote with non-negative-weight SSSP.
For example, Dijkstra's algorithm, from the 50s, requires this assumption and runs in near-linear time.
The algorithms for single-source shortest paths with negative integer weights (denoted negative-weight SSSP), on the other hand, have until very recently been significantly slower.
From the 50s, the classic Bellman-Ford algorithm gives an $O(mn)$ time algorithm,\footnote{Here and throughout, we use $n$ to denote the number of vertices, $m$ to denote the number of edges of $G$.} which either computes distances from $s$ to $v$ or reports a negative-weight cycle.
A series of improvements since then (\cite{Goldberg,CMSV,BLN20,AMV20}) culminated in two recent breakthroughs: the algorithm of Chen, Kyng, Liu, Peng, Probst Gutenberg, and Sachdeva~(\cite{cfl2022}) solving transshipment and min-cost flow in time $m^{1+o(1)}$, thus implying the same runtime for negative-weight SSSP, and a parallel and independent result of Bernstein, Nanongkai, Wulff-Nilsen~(\cite{bernstein2022negative}) giving a $\tOh{m}$ time\footnote{Here and throughout, we use the soft-O notation $\tilde{O}$ to suppress polylogarithmic (in $n$) factors. Throughout the paper, we assume the maximum weight edge (in absolute value) of $G$ is polynomially bounded.} algorithm for negative-weight SSSP that uses relatively simpler techniques.
Follow-up work by Bringmann, Cassis, and Fischer significantly reduces the number of log factors in the $\tilde{O}(m)$ runtime~(\cite{bringmann2023negative}).
In this paper, we take the exploration of negative-weight SSSP to parallel, distributed, and quantum models of computation.
Should there be analogous results there?

In parallel models, there has been much recent progress for the non-negative-weight SSSP problem.
Rozhoň, Haeupler, Martinsson, Grunau and Zuzic~(\cite{rozhon2022undirected}) and Cao and Fineman~(\cite{exactcf2023}) showed that SSSP with polynomially bounded non-negative integer edge weights can be solved with $\tilde{O}(m)$ work and $n^{1/2+o(1)}$ depth in the parallel model.
By contrast, the known bounds for negative-weight SSSP are significantly weaker: the classic Bellman-Ford algorithm solves negative-weight SSSP with $O(mn)$ work and $O(n)$ depth, and recently, Cao, Fineman and Russell~(\cite{cfrnagetivesssp}) improved this to $\tilde{O}(m\sqrt{n})$ work and $n^{5/4+o(1)}$ depth.

Similarly, in distributed models, Rozhoň \etal~(\cite{rozhon2022undirected}) and Cao and Fineman~(\cite{exactcf2023}) show algorithms for SSSP with non-negative integer edge weights that take $\tilde{O}((n^{2/5+o(1)}D^{2/5}+\sqrt{n} + D)$ rounds\footnote{Here and throughout, we use $D$ to denote the undirected hop-diameter of $G$.}.
On the negative-weight SSSP front, the Bellman-Ford algorithm takes $O(n)$ rounds.
The current state-of-the-art by Forster, Goranci, Liu, Peng, Sun and Ye~(\cite{distributednegativessp2021}), which uses Laplacian solvers, gives an $\tilde{O}(m^{3/7+o(1)}(n^{1/2}D^{1/4}+D))$ round algorithm for negative-weight SSSP.

In the quantum edge query model, Durr, Heiligman, Høyer, and Mhalla~\cite{DurrHHM06} show an algorithm for SSSP with non-negative edge weights in $O(n^{1.5})$ queries to the adjacency matrix or $O(m^{1/2}n^{1/2})$ queries to the adjacency list, which are both tight.
We are not aware of any quantum edge query algorithm solving negative-weight SSSP better than the trivial $O(n^2)$ or $O(m)$ algorithm.

There is a substantial gap between the best known upper bounds for non-negative-weight SSSP and negative-weight SSSP in these models and, in fact, the number of landmark algorithms for negative-weight SSSP has been comparatively few.
This begets the following question: Can we close the gap, and get parallel, distributed, and quantum algorithms for negative-weight SSSP that are nearly as efficient as the best non-negative-weight SSSP algorithms?
This paper gives an answer in the affirmative. 

\paragraph*{Main Results.}
The main results of this paper are as follows.

\begin{restatable}[Parallel SSSP reduction with negative edge-weight]{theorem}{mp}
\label{thm:mainparallelreduction}
Assuming there is a parallel algorithm answering (non-negative integer weight) SSSP on directed graphs in $W(m, n)$ work and $S(m, n)$ span, then there
 exists a randomized algorithm that solves negative-weight SSSP on directed graphs $G$ with polynomially bounded integer edge-weights with $O(W(m, n)(\log n)^{O(\sqrt{\log n})})$ work and $\tilde{O}(S(m, n)2^{\sqrt{\log n}})$ span with high probability.
\end{restatable}

Using state-of-the-art results for non-negative-weight SSSP (\cite{rozhon2022undirected} and \cite{exactcf2023}) with \Cref{thm:mainparallelreduction} immediately gives a randomized parallel algorithm that solves negative-weight SSSP on directed graphs with $m^{1+o(1)}$ work and $n^{1/2+o(1)}$ span, with high probability. 

\begin{restatable}[Distributed SSSP reduction with negative edge-weight]{theorem}{md}
\label{thm:maindistributedreduction}
In the \congest model, assuming there is an algorithm answering (non-negative integer weight) SSSP on directed graphs in $T(n, D)$ rounds, then there
 exists a randomized algorithm that solves negative-weight SSSP on directed graphs $G$ with polynomially bounded integer edge-weights and undirected hop-diameter $D$ in $O((T(n,D) + \sqrt{n} + D)(\log n)^{O(\sqrt{\log n})})$ rounds with high probability.
\end{restatable}

Using state-of-the-art results for non-negative-weight SSSP (\cite{rozhon2022undirected} and \cite{exactcf2023}) with \Cref{thm:maindistributedreduction} immediately gives a distributed randomized algorithm that solves negative-weight SSSP on directed graphs with $O((n^{2/5+o(1)}D^{2/5}+\sqrt{n} + D)n^{o(1)})$ rounds of communication in the \congest model with high probability. For general graphs there is a lower bound of $T(n,D) = \Omega(\sqrt{n} + D)$~(\cite{peleg1999near}), so the factor of $\sqrt{n} + D$ in our runtime does not impact the efficiency of our reduction.

\begin{restatable}[Quantum SSSP reduction with negative edge-weight]{theorem}{mq}
\label{thm:mainquantum}
In the quantum edge query model, assuming there is an algorithm answering (non-negative integer weight) SSSP on directed graph in $Q(m,n)$ queries, then there
 exists a randomized algorithm that solves negative-weight SSSP on directed graphs $G$ with polynomially bounded integer edge-weights in $O(Q(m,n)(\log n)^{O(\sqrt{\log n})}))$ queries.
\end{restatable}
Using the state-of-the-art results for non-negative-weight SSSP~\cite{DurrHHM06} with~\Cref{thm:mainquantum} immediately gives a quantum edge query algorithm that solves negative-weight SSSP on directed graphs with $n^{1.5+o(1)}$ queries to the adjacency matrix, or $m^{1/2}n^{1/2+o(1)}$ queries to the adjacency list. The upper bound is optimal  up to an $n^{o(1)}$ factor by the $\Omega(n^{1.5})$ and $\Omega(\sqrt{mn})$ lower bound result~\cite{DurrHHM06}.

We note that all of our results take the form of a general reduction from negative-weight SSSP to non-negative-weight SSSP, so any further advance in non-negative-weight SSSP immediately translates to improved bounds for negative-weight SSSP.
Modulo $n^{o(1)}$ factors, the complexity of parallel and distributed algorithms for non-negative-weight SSSP match that of directed reachability \cite{jambulapati2019parallel,distributedhopsets}; any improvements to negative-weight SSSP beyond the $n^{o(1)}$ factor would thus first require improvements to directed reachability in these models.

Our reductions follow the high-level framework of the recent sequential $\tilde{O}(m)$-time algorithm of Bernstein, Nanongkai, and Wulff-Nilsen \cite{bernstein2022negative}. At the heart of their framework is the use of directed low-diameter decompositions (on graphs with non-negative edge weights), and one of our key technical contributions is to give algorithms for computing such a directed decomposition in parallel, distributed, and quantum models. We next give an overview of low-diameter decompositions.

\subsection{Our Further Contributions}
On top of algorithms for negative-weight SSSP, we provide two algorithms that we believe are of independent interest. The most significant one is an efficient implementation of directed low-diameter decomposition in parallel, distributed, and quantum models. We also show an algorithm for computing strongly connected components and their topological ordering in the \congest model.
Like \Cref{thm:mainparallelreduction,thm:maindistributedreduction}, these results are presented as reductions to non-negative-weight SSSP.
An advantage of this approach is that our results scale with non-negative-weight SSSP; if there is any progress in the upper bounds to non-negative-weight SSSP, progress to the bounds here immediately follow.

\subsection*{Directed Low Diameter Decomposition}

\paragraph*{Previous Work} Low-Diameter Decomposition (LDD) has long been used to design efficient algorithms for \textit{undirected} graphs in several models of computation \cite{Awerbuch89,LinialS93,bartal1996probabilistic,ElkinNe16,GKM17,rozhon2020polylogarithmic,REGH2022,rozhon2022undirected,bernstein2022negative,Ghaffari23}. A few recent papers developed a generalization of LDD that also applies to \emph{directed} graphs\cite{chechik2016decremental,bernstein2019decremental,bernstein2020near}. Bernstein \etal \cite{bernstein2022negative} use directed LDD as one of the key subroutines in their sequential algorithm for negative-weight SSSP, and they present a sequential algorithm for computing directed LDD in near-linear time.

In undirected graphs, it is also known how to compute LDD efficiently in other models of computation, including parallel and distributed models; in fact, the well-known algorithm of Miller, Peng, and Xu (MPX) reduces this problem to a single shortest-path-tree computation from a dummy source $s$ \cite{MillerPX13}. 

\paragraph*{Our Results}
One of our main technical contributions is showing that in several computation models, computing \emph{directed} LDD can similarly be reduced to a small number of shortest-path-tree computations.
This requires new techniques for overcoming obstacles that are unique to directed graphs; see Section~\ref{subsec:LDDoverview}
for an overview of these new techniques.

The input/output guarantees of directed LDD are stated below; they are the same as those in the sequential paper of \cite{bernstein2022negative}. (Note in particular that the input to LDD is a graph with \emph{non-negative} weights.) Intuitively, for a given parameter $d$, the decomposition computes a small set of ``bad" edges $\Erem$ such that (1) Every strongly connected component in $G \setminus \Erem$ has weak diameter at most $d$ and (2) Every edge of the graph is in $\Erem$ with probability at most $\tilde{O}(w(e)/d)$.

\begin{restatable}[Low-Diameter Decomposition, Algorithm~\ref{alg:lowdiamterdecomposition}]{lemma}{lddl}\label{lem:LDDcorrectness}

Let $G=(V,E,w)$ be a directed graph with a polynomially bounded weight function $w:E\to\mathbb{N}$ and let $d$ be a positive integer. There
 exists a randomized algorithm $\LDD(G,d)$ with following guarantees:
    \begin{itemize}
        \item INPUT: An $n$-node $m$-edge, graph $G = (V, E, w)$ with non-negative integer edge weight and a positive integer $d$.
        \item OUTPUT: (proved in~\Cref{subsec:LDDalg}) a set of edges $\Eref\subseteq E$
    satisfying:
        \begin{itemize}
            \item Each SCC of the subgraph $G \setminus \Eref$ has weak diameter at most $d$ in $G$, i.e. if $u,v$ are two vertices in the same SCC, then $dist_G(u,v)\le d$ and $dist_G(v,u)\le d$.
            \item For any $e\in E$, we have $\Pr[e\in \Eref]=O\left(\frac{w(e)\log^2n}{d}+\frac{1}{n^8}\right)$
        \end{itemize}
        \item RUNNING TIME: The algorithm is randomized and takes $\tilde{O}(1)$ calls to (non-negative integer weight) SSSP. 
        More specifically:
            \begin{itemize}
                \item Assuming there is a parallel algorithm answering non-negative-weight SSSP in $W(m, n)$ work and $S(m, n)$ span, then $\LDD(G,d)$ takes $\tilde{O}(W(m, n))$ work and $\tilde{O}(S(m, n))$ span with high probability.
                \item Assuming there exists a \congest algorithm answering non-negative-weight SSSP in $T(n, D)$ rounds, then $\LDD(G,d)$ takes $\tilde{O}(T(n, D) + \sqrt{n} + D)$ rounds in the \congest model with high probability, where $D$ is the undirected hop diameter.
                \item Assuming there exists a quantum edge query algorithm answering non-negative-weight SSSP using $Q(m,n)$ queries, then $\LDD(G,d)$ takes $\tilde{O}(Q(m,n))$ queries with high probability.
            \end{itemize} 
    \end{itemize}
\end{restatable}

We observe that the complexity of quantum query algorithms is typically sublinear in \(m\), yet the output size of \(E^{rem}\) may reach up to \(m\). Consequently, rather than directly producing \(E^{rem}\) as output, it is represented in an implicit format within \(\tilde{O}(n)\) bits. For further details, refer to Section~\ref{subsubsec:lddquantum}.

The concept of undirected low-diameter decomposition was first introduced in the context of parallel and distributed algorithms, and some of the most important applications and use cases are in these areas.
We are therefore optimistic that our directed parallel and distributed low-diameter decomposition algorithm can be applied to solve various problems in the distributed and parallel setting in the future, beyond the application of computing negative weight shortest paths addressed in this paper.

\subsection*{Strongly Connected Components and Their Topological Order in \congest}
Another subroutine we need in our algorithm is finding the strongly connected components of a graph. It is known that in the parallel setting this problem reduces to single-source reachability (\cite{schudy2008finding}). In this paper we show a similar reduction for the \congest setting; 
we use the same high-level framework as the parallel reduction, but this is difficult to port directly into the \congest model; we show that by going through the recently developed \DM, we are able to overcome this difficulty. See \Cref{sec:scc-top} for details.

\begin{restatable}{lemma}{scctopcongest}
\label{lem:scc-topsort-congest}
    There is a \congest algorithm that, given a directed graph $G=(V,E)$, and assuming there is an algorithm answering non-negative-weight SSSP in $T(n, D)$ rounds, outputs strongly connected components listed in a topological order.
    More specifically, it outputs a polynomially-bounded labelling $(r_v)_{v\in V}$ such that, with high probability
    \begin{enumerate}
        \item $r_u = r_v$ if and only if $u$ and $v$ are in the same strongly connected component;
        \item when the SCC that $u$ belongs to has an edge towards the SCC that $v$ belongs to, $r_u > r_v$.\footnote{As a matter of convenience, the labels correspond to a reverse topological order (i.e. something which appears earlier in a topological order has a larger label than something which appears later).}
    \end{enumerate}
    The algorithm takes $\tOh{T(n,D) + \sqrt{n} + D}$ rounds.
\end{restatable}

It is worth noting that a more careful examination gives a round complexity in terms of calls to a reachability oracle, rather than a non-negative-weight SSSP oracle (see \Cref{rem:scc-using-reach}).
Plugging in the current state-of-the-art \congest algorithm for non-negative-weight SSSP (\cite{rozhon2022undirected} and \cite{exactcf2023}) leads to a $\tOh{n^{1/2}+D+n^{2/5+o(1)}D^{2/5}}$ round algorithm (\Cref{cor:scc-topsort-congest}), answering a question posed in \cite{bernstein2019distributed} which asked if a lower bound of $\tOm{n}$ rounds applies to the problem of finding SCCs.

\subsection{Organization}
In Section~\ref{sec:prelim}, we provide the necessary terminology, notation, and basic results that will be used throughout the paper.
This section can be skipped and referred back to as needed. 
In Section~\ref{sec:overview}, we present a high-level overview of \cite{bernstein2022negative} and discuss the key challenges involved in adapting the results to other models. Section~\ref{sec:lowediameterdecomposition} presents our algorithm for low-diameter decomposition, Section~\ref{sec:scc-top} our \congest algorithm for computing SCCs and topological sort. 
Section~\ref{sec:framework} introduces the input/output guarantees of all our key subroutines, and provides a description and pseudocode of the overall algorithm that combines these subroutines.
Section~\ref{sec:fix-dag} and Section~\ref{section:estdist} discuss two other subroutines needed by our algorithm ($\FixDag$ and $\EstDist$).
Lastly, Section~\ref{sec:scaledown_analysis} and Section~\ref{sec:SPMain} give a formal analysis of the algorithm described in Section~\ref{sec:framework}.

Sections~\ref{sec:lowediameterdecomposition} and ~\ref{sec:scc-top} are described as self-contained problems. Sections~\ref{sec:framework}-\ref{sec:scaledown_analysis} are only necessary if the reader wishes to understand how we combine those subroutines to parallelize the sequential negative-weight SSSP~\cite{bernstein2022negative} algorithm.

\section{Definitions and Preliminaries}
\label{sec:prelim}

A weighted directed graph $G$ is a triple $(V, E, w)$ where $w: E \rightarrow \mathbb{Z}$ is a weight function. For a weighted directed graph $G$, the number of vertices and edges are $|V(G)| = n $ and $|E(G)| = m$, respectively. We denote the set of negative edges by $\Eneg(G) = \{e \in E \mid w(e) < 0 \}$.
For a subset $V' \subset V$, we denote the induced graph on $V'$ by $G[V']$ and the induced edges on $V'$ by $E(V')$. For an edge set $E'\subseteq E$, when we treat $E'$ as a \emph{subgraph} of $G$, we mean the graph $(\{u,v\mid (u,v)\in E'\},E',w)$. 
A \defn{path} is a sequence of vertices joined by edges; sometimes we refer to the path by the sequence of vertices and sometimes by the edges.
A \defn{strongly connected component} (SCC) is a set of vertices $S$ such that for any pair of vertices $u,v \in S$, there is a path from $u$ to $v$ contained entirely in $S$.
For a path $\Gamma=\langle v_0,v_1,\ldots,v_k\rangle$, the \defn{weight} of $\Gamma$ is given by $w(\Gamma) = \sum_{i=1}^k w(v_{i-1},v_i)$, that is, the sum of the weights of the edges on the path.   
For a pair of nodes $u,v \in V$, the \defn{shortest path distance} from $u$ to $v$ is the minimum length over all paths that start at $u$ and end at~$v$.  We use $\dist_G(u,v)$ to denote this shortest path distance with respect to the graph~$G$.  When the graph $G$ is clear in the context, we simply write $\dist(u,v)$.  If there is no $u$-to-$v$ path, then we define $\dist(u,v)=+\infty$. 
Given a directed graph $G=(V,E,w)$,
a vertex $s\in V$ and $d\in\mathbb{N}$, we define $\ingball(s,d)=\{v\mid dist_G(v,s)\le d\}$ and $\outgball(s,d)=\{v\mid dist_G(s,v)\le d\}$, the in or out balls centered at $s$ with weighted radius $d$.
For a given graph $G = (V, E, w)$ and a subset of vertices $S\subseteq V$, we define $\incut(S)=\{(u,v) \in E \mid u\not\in S,v\in S\}$ and $\outcut(S)=\{(u,v) \in E\mid u\in S,v\not\in S\}$, the in or out edge sets crossing $S$ .

 When we say that an algorithm achieves performance $O(f(n))$ with high probability, we mean the following: for a particular choice of constant $c>0$, with probability at least $1-1/n^c$ the algorithm achieves performance $O(f(n))$.

Throughout the paper, we will assume the maximum weight edge (in absolute value), $W_{in}$, is polynomially bounded in $n$ and ignore $\log W_{in}$ terms. Based on the following theorem, we only incur one additional $\log (nW_{in})$ factor. 
\begin{theorem}
\label{thm:dependecyonW}
In the parallel or distributed model, if there is an algorithm solving exact SSSP with edge weights from $\{-1, 0, 1, ..., 2n - 1, 2n \}$ with runtime $T(m, n)$ (work, span, or rounds), then there exists an algorithm solving exact SSSP with edge weight from $\{-W_{in}, -(W_{in} - 1), ..., 0, 1, ..., W_{in} - 1, W_{in} \}$ with runtime $O(T(m, n)\log \big(nW_{in}\big))$.
\end{theorem}

\begin{proof}

Let $\mathcal{A}$ be the algorithm solving exact SSSP with edge weights from $\{-1, 0, 1, ..., 2n - 1, 2n \}$. We will construct algorithm $\mathcal{B}$ solving SSSP with arbitrary integer edge weights. The algorithm $\mathcal{B}$ first uses the scaling framework of Goldberg~\cite{Goldberg} to eliminate all edges with the negative edge weights at the expense of an extra $\log W_{in}$ factor by calling $\mathcal{A}$. We divide the algorithms into $\log W_{in}$ rounds, where in each round, we only need to produce a feasible price function for the special case that all edge weights are integers with value at least $-1$. 

Notably, while we ensure that negative edges have weights no less than $-1$, positive edges may be larger than $2n$. A crucial insight is that for positive edges with weights exceeding $n$, we adjust their weights to $n$ during each round. This adjustment is sufficient for generating a feasible price function, considering that any simple path will have at most $n - 1$ edges and any negative edge is no less than $-1$. Consequently, $\mathcal{A}$ is invoked to compute the price function, given that edge weights are now constrained between $-1$ and $2n$. Since the edge weight per round does not exceed $n$, the maximum SSSP distance is capped at $n^2$, potentially increasing the edge weight by up to $O(n^2 W_{in})$ in the final price function.

Next, for graph with non-negative integer edge weight, Klein and Subramanian~\cite{ks97} already show that we can use $O(\log D)$ call to reduce the maximum edge weight to $2n$, where $D$ is the maximum edge weight. The same computation also has been shown by Forster and Nanongkai~\cite{fn18} in the \congest model. Given the maximum edge weight reaches $O(n^2 W_{in})$, an extra computational expense of $O(\log(nW_{in}))$ is incurred.
\end{proof}

\subsection{Definitions from \cite{bernstein2022negative}}

The following two definitions are taken from \cite{bernstein2022negative}.
The first definition is used for describing graphs augmented with a dummy source, and for defining a new graph where negative edge weights are raised by an additive constant.

\begin{definition}[$G_s,w_s,G^B,w^B,G^B_s,w^B_s$][Definition 2.3 of \cite{bernstein2022negative}]
\label{def:GBs}

Given any graph $G = (V,E,w)$, we let $G_s = (V \cup \{s\},E \cup \{(s,v)\colon v \in V\}, w_s)$ refer to the graph $G$ with a dummy source $s$ added, where there is an edge of weight $0$ from $s$ to $v$ for every $v \in V$ and no edges into $s$. Note that $G_s$ has a negative-weight cycle if and only if $G$ does and that $dist_{G_s}(s,v) = \min_{u \in V} dist_G(u,v)$. \\
For any integer $B$, let $G^B = (V,E,w^B)$ denote the graph obtained by adding $B$ to all negative edge weights in $G$, i.e., $w^B(e) = w(e) + B$ for all $e \in \Eneg(G)$ and $w^B(e) = w(e)$ for $e \in E \setminus \Eneg(G)$. Note that $(G^B)_s = (G_s)^B$ so we can simply write $G^B_s = (V \cup \{s\},E \cup \{(s,v) \colon v \in V\}, w^B_s)$.  
\end{definition}

The next definition introduces $P_G(v)$, which is the shortest path to $v$ (from a dummy source) with the least number of negative-weight edges and $\eta_G(v)$, which counts the number of aforementioned negative edges.

\begin{definition}[$\eta(G),\eta_G(v),P_G(v)$][Definition 2.4 of \cite{bernstein2022negative}]
\label{def:eta}

For any graph $G = (V,E,w)$, let $G_s$ and $s$ be as in Definition \ref{def:GBs}.
Define 

$\eta_G(v) := \left\{
\begin{array}{ll}
\infty & \text{if $dist_{G_s}(s,v) = -\infty$} \\
\min\{|\Eneg(G) \cap P| \colon \text{ $P$ is a shortest $sv$-path in $G_s$}\}; & \, \textrm{otherwise}. \\
\end{array}
\right. $

Let $\eta(G) = \max_{v \in V} \eta_G(v)$. When $dist_G(s,v) \neq - \infty$, let $P_G(v)$ be a shortest $sv$-path on $G_s$ such that  $|\Eneg(G) \cap P_G(v)| = \eta_G(v).$
When the context is clear, we drop the subscripts.

\end{definition}

The following definitions and lemmas about price functions are standard in the literature, and can all be found in \cite{bernstein2022negative}.  Price functions were first introduced by Johnson \cite{Johnson77} and heavily used since then; they are used to rescale the weights of edges (in our case, make them non-negative) without changing the structure of shortest paths.

\begin{definition}[Definition 2.5 of \cite{bernstein2022negative}]
\label{def:pricefunction}
Consider a graph $G = (V,E,w)$ and let $\phi$ be any function: $V \mapsto \mathbb{Z}$. Then, we define $w_\phi$ to be the weight function $w_\phi(u,v) = w(u,v) + \phi(u) - \phi(v)$ and we define $G_\phi = (V,E,w_\phi)$. We will refer to $\phi$ as a price function on $V$. Note that $(G_\phi)_\psi = G_{\phi + \psi}$.
\end{definition}
\begin{definition}[Definition 2.6 of \cite{bernstein2022negative}]
\label{def:graphequivalence}
We say that two graphs $G = (V,E,w)$ and $G' = (V,E,w')$ are \emph{equivalent} if (1) any shortest path in $G$ is also a shortest path in $G'$ and vice-versa and (2) $G$ contains a negative-weight cycle if and only if $G'$ does.
\end{definition}
\begin{lemma}[Lemma 2.7 of \cite{bernstein2022negative}]
\label{lem:prelim_eqivalence_price_function}
Consider any graph $G = (V,E,w)$ and price function $\phi$. For any pair $u,v \in V$ we have $dist_{G_\phi}(u,v) = dist_G(u,v) + \phi(u) - \phi(v)$, and for any cycle $C$ we have $w(C) = w_\phi(C)$. As a result, $G$ and $G_\phi$ are equivalent. Finally, if $G = (V,E,w)$ and $G' = (V,E,w')$ and $w' = c w$ for some positive $c$, then $G$ and $G'$ are equivalent. 

\end{lemma}
\begin{lemma}[Lemma 2.8 of \cite{bernstein2022negative}]
\label{lem:prelim_nonnegative}
Let $G = (V,E)$ be a directed graph with no negative-weight cycle and let $S$ be the dummy source in $G_s$. Let $\phi(v) = dist_{G_s}(s,v)$ for all $v \in V$. Then, all edge weights in $G_\phi$ are non-negative. 
\end{lemma}

\subsection{Models of Computation}
\paragraph*{Parallel Model. } We consider the  PRAM CRCW model, where both simultaneous reads of and simultaneous writes to the same memory cell are allowed. The time complexity of a PRAM CRCW algorithm is measured by the work and span, where the work is defined as the total number of instructions executed across all
processors and the span is the length of the critical path (i.e., the
length of the longest chain of sequential dependencies).

\paragraph*{Distributed \congest Model. } %
In the \congest model, time is divided into discrete time slots, where each slot is called a \emph{round}. Throughout the paper, we always use $n$ to denote the number of vertices in our distributed network, i.e., $|V|$. In each round, each vertex in $V$ can send an $O(\log n)$ bit message to each of its neighbors. At the end of each round, vertices can do arbitrary local computations. A \congest algorithm initially specifies the input for each vertex and, after several rounds, all vertices terminate and generate output. The time complexity of a \congest algorithm is measured by the number of rounds.

Since the inputs and outputs to the distributed network should be specified for each vertex, we must be careful when we say something is given as input or output. Here we make some assumptions. For the network $G=(V,E)$, we say a subset of vertices (or a single vertex) $V'\subseteq V$ is the input or output if every vertex is given the information about whether it is in $V'$. We say a subgraph $H$ (a subset of edges, for example, paths or circles) is the input or output if every vertex knows the edges in $H$ adjacent to it. When we say a number is an input or output, we normally mean the number is the input or output of every vertex unless otherwise specified.

In single source shortest path problem, each edge $e$ in the network $G=(V,E)$ is assigned an integer weight $w(e)$ and a direction, which defines a weighted directed graph $G'=(V,E,w)$. A source node $s$ is specified. In the \congest model, the inputs are (i) each node knows the weights and directions of all its incident edges, (ii) each node knows whether it is $s$ or not. The goal is to let every node $v$ output $dist_{G'}(s,v)$.

\paragraph*{\DM. } For ease of explaining algorithms in \congest, we allow ourselves to perform the following steps defined in the \DM (\cite{0001Z22}). 
\begin{definition}[Steps in \DM] \quad

\begin{enumerate}
    \item \textbf{Contraction step.} Each node $v$ computes a $\tOh{1}$-bits value $a_v$. Each edge $(u,v)$ is marked with $c_e\in\{\bot,\top\}$ based on $a_u,a_v$. Contracting all edges with $c_e = \top$ and self-loops removed to get the minor graph $G'=(V',E')$. We also treat each node $S\in V'$ as a vertex set $S\subseteq V$.
    \item \textbf{Consensus step.} Each node $v\in V$ computes a $\tOh{1}$-bits value $x_v$. For every $S\in V'$, each $u\in S$ gets the value $\oplus_{v\in S}x_v$ where $\oplus$ is an operator satisfying commutative and associative laws, like sum, min, max.
\end{enumerate}
\end{definition}

This allows us to aggregate the values in each contracted node in $\tOh{\sqrt{n}+D}$ rounds.
\begin{theorem}[Theorem 17 in~\cite{0001Z22}]\label{thm:minoraggregation}
    The steps of the \DM can be simulated in the \congest model in $\tOh{\sqrt{n}+D}$ rounds.
\end{theorem}

\paragraph*{Quantum Edge Query Model.} There are two ways of defining algorithms in the quantum edge query model, depending on whether the input is given by an adjacency list or adjacency matrix. When the input is given by an adjacency matrix, the algorithm needs to compute some property of a graph $G=(V,E)$ while they can only access the graph by making queries to $Q_A$ defined as follows. $A$ is the adjacency matrix of $G$. Each query to $Q_A$ is a transformation which takes inputs $|u,v,0,w\rangle$ and output $|u,v,A_{u,v},w\rangle$. The inputs to $Q_A$ can also be superposed. Similarly, when the input is given by an adjacency list, the algorithm can access the graph by a transformation which takes inputs $|u,i,0,w\rangle$ and output $|u,i,N(u,i),w\rangle$ where $N(u,i)$ is the $i-th$ neighborhood of $u$ (if not exists, return a non-vertex index like $\bot$). The complexity is measured by the number of queries made by an algorithm.

\section{High Level Overview}
\label{sec:overview}

Our results follow the framework of \cite{bernstein2022negative}, which provides a sequential algorithm for negative-weight SSSP that takes $\Ot(m)$ time with high probability.

\subsection{Overview of the Sequential Algorithm from~\cite{bernstein2022negative}}
This section provides a summary of Bernstein \etal's approach to computing exact shortest paths on a graph with integer edge weights (both positive and negative)~\cite{bernstein2022negative}.

The final goal is to compute a price function $\phi$ such that all edges in $G_{\phi}$ are non-negative; since $G_{\phi}$ and $G$ are equivalent (Lemma~\ref{lem:prelim_eqivalence_price_function}), one can then run SSSP for non-negative weights on $G_{\phi}$. 
Following the standard scaling framework, Bernstein \etal's algorithm computes such a $\phi$ over multiple scaling rounds.  
The key component of the algorithm is a procedure \scaledown that computes a price function that halves the minimum negative edge weight: given a weighted directed graph $G$ where all edge weights are at least as large as $-2B$ for some non-negative parameter $B$, \scaledown outputs a price function $\phi$ such that in $G_{\phi}$ all edge weights are at least as large as $-B$. A procedure \scaledown with these guarantees can then easily be used to solve negative-weight SSSP, so for the rest of this section we focus exclusively on the algorithm \scaledown.

Recall from \Cref{def:GBs} that we obtain $G^B_s$ from $G$ by adding $B$ to all \emph{negative} edge weights, and adding a dummy source $s$. For somewhat subtle reasons that are not relevant to this overview section, the \scaledown algorithm primarily works with $G^B_s$, and in particular it computes a price function $\phi$ such that $w_{\phi}(e) \ge 0$ in $G^B_s$, which implies that $w_{\phi}(e) \geq -B$ in $G_s$, as desired.

\paragraph*{The Algorithm $\scaledown$. }
We now give a high-level overview of the sequential algorithm for \scaledown in \cite{bernstein2022negative}. Our algorithm will follow the same general framework, but with a few key differences discussed below. See \Cref{sec:scaledown-high-lvl} for a more detailed description of our algorithm, along with pseudocode. 

In order to compute the desired price function $\phi$, the \scaledown procedure consists of four phases.
\begin{itemize}
    \item Phase 0: Run a Low Diameter Decomposition on $G^B$ with negative-weight edges rounded up to $0$ (See \Cref{lem:LDDcorrectness}).
        This gives a set $\Erem$ of removed edges, such that all SCCs of $G \setminus \Erem$ have small weak diameter. Observe that by the guarantees of directed LDD, after Phase 0 the graph will contain three types of edges: \textbf{(i)} edges within each SCC of $G \setminus \Erem$, \textbf{(ii)} edges that connect one component to another and are not in $\Erem$; one can intuitively think of these edges as being DAG-like, since they always go forward in the topological ordering of the SCCs of $G \setminus \Erem$, and \textbf{(iii)} edges from $\Erem$, where any edge $e \in E$ is in $\Erem$ with probability at most $\tilde{O}(w(e)/D)$. 
        Phase 1 of $\scaledown$ addresses the first type of edge, Phase 2 the second, and Phase 3 the third.

    \item Phase 1: Recursively call \scaledown on the edges inside each SCC. This finds a price function under which edges inside each SCC have non-negative weight, thus fixing the type 1 edges.
    \item Phase 2: Fix the edges not in $\Erem$ that connect one component to another (i.e. the DAG-like edges); that is, compute a price function that makes their weight non-negative.
    \item Phase 3: Fix the edges of $\Erem$.
\end{itemize}

\paragraph*{Implementing the Three Phases in the Sequential Model.} Recall that the low-diameter decomposition provides two guarantees: first that the weak diameter in each SCC is bounded, and second that each edge will be in $\Erem$ with probability proportional to its weight. Loosely speaking, the first guarantee ensures that in Phase 1, recursively calling \scaledown on the SCCs is making progress, because one can show that as the diameter decreases the maximum number of negative edges on any shortest path is reduced (This is technically only true in a carefully defined auxiliary graph). Bernstein \etal~show that after $O(\log n)$ recursive calls to \scaledown, the number of negative edges on any shortest path is at most $\tilde{O}(1)$. 
They then show an algorithm called $\ElimNeg$ that can efficiently compute single-source shortest paths in graphs with this property; running this algorithm from a dummy source $s$ and applying Lemma \ref{lem:prelim_nonnegative} yields the desired price function.

For Phase 2, the focus is on DAG-like edges connecting the SCCs of $G \setminus \Erem$; by Phase 1, the edges in each SCC already have non-negative weights. The algorithm simply contracts each SCC into a vertex to get an acyclic graph whose edge set consists of type-2 edges. Computing a price function for these edges turns out to be very easy because the underlying graph is a DAG.

By the time the algorithm reaches Phase 3, only edges in $\Erem$ can still be negative. %
The second guarantee of low-diameter decomposition ensures that every shortest path has few edges from $\Erem$ in expectation.
In other words, by the time the algorithm reaches Phase 3, the remaining graph has the following property: \emph{on average}, the shortest path from $s$ to any vertex $v$ contains few negative edges. The authors of ~\cite{bernstein2022negative} then show that their subroutine \ElimNeg can efficiently compute shortest distances in any graph with this property; by \Cref{lem:prelim_nonnegative}, these distances then give the desired price function.

\subsection{Adapting to Other Models: Challenges \& Solutions}
\label{subsec:challengesandsolutions}
Although this framework works well in the sequential setting, it presents additional challenges in parallel and distributed models. We summarize these obstacles below.
\paragraph*{Obstacle 1: Low Diameter Decomposition. }
Undirected LDD in parallel and distributed models (e.g. \cite{MillerPX13},\cite{BeckerEL20},\cite{RozhonEGH22}) is commonly solved via the following framework: each node grows a ball starting at a random ``delayed'' time (the distribution of the randomness is picked carefully), and the boundary of a ball stops growing once it reaches another ball. The balls create a partition of the graph, where each ball has a low diameter, and if we define $\Erem$ to be the edges between different balls, then it can be shown that any particular edge $e$ is in $\Erem$ with small probability.
For example, the simplest instantiation of the above approach is the well-known MPX algorithm of Miller, Peng, Xu~\cite{MillerPX13}. In this algorithm, every vertex picks a random delay $\delta(v)$, and then vertex $x$ is assigned to the ball $B_y$ of the vertex $y$ that minimizes $\min_{v \in V} \dist(x,v) + \delta(v)$ (some of the $B_y$ may end up empty.) One can easily compute this minimum for every $x$ by computing a single shortest path tree from a dummy source with an edge of weight $\delta(v)$ to every vertex $v$.

{\em Natural approaches:}
A natural way to extend the above algorithm to the directed setting is as follows: each vertex in parallel grows an outgoing ball (which means the ball growing only uses the edges going out of this ball), and the boundary stops growing once it reaches another out-ball; the edges pointing out of every ball are then included in $\Erem$. (The process should be repeated with incoming balls, but we leave this out to keep the discussions simple.)
However, we can no longer argue that each edge is included in $\Erem$ with a small probability. See~\Cref{fig:star} for an example. This example contains a star with a middle vertex denoted by $s'$, and $n-1$ other vertices which have edges pointing to $s'$. If $s'$ is the first vertex to start growing a ball (because it ends up with the lower random delay), then the LDD algorithm will create a ball $\{s'\}$, so when every vertex on the boundary of the star later grows its own ball, the algorithm will include all the edges of the graph into $\Erem$. On the other hand, if some other vertex $s$ in the boundary of the star starts growing a ball before $s'$, this will result in ball $\{s, s'\}$, and when other boundary vertices start growing their own balls, all edges other than $(s,s')$ will be added to $\Erem$.

\begin{figure}[t!]
    \centering
    \begin{minipage}{.7\columnwidth}
    \centering
        \includegraphics[width=0.5\columnwidth]{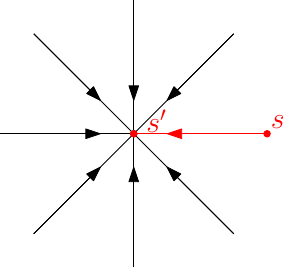}
        \captionof{figure}{One vertex in the boundary of the star(denoted as $s$) starts growing an outgoing ball and includes $s'$ into the ball. Every other vertex in the boundary of the star will start growing balls later and include all edges other than $(s,s')$ into $\Erem$.}\label{fig:star}
    \end{minipage}%
\end{figure}

The above example is rather naive
because the input graph is not strongly connected; so, we can just return $\Erem =\emptyset$ (then, all the strongly connected components (SCCs) already have low diameters). 
But the construction can be extended to the more sophisticated example in \Cref{fig:bottlenecks}.
The example graph in \Cref{fig:bottlenecks} contains a clique of $n/2$ vertices and a directed cycle of $n/2$ vertices. An edge $(s,s')$ is pointing from the clique to the directed cycle.  We will show that in this graph, there is some particular edge that is included in $\Erem$ with a constant probability, which is too high.

The right way to compute a directed LDD on this graph is to set $\Erem$ to be a single random edge on the cycle, but this is not what the parallel ball-growing approach would do. To see this, consider two cases. The first case is that the vertex $s'$ ends up in the ball $B_v$ of some vertex $v$ in the clique (because $v$ gets low random delay). The vertex $s''$ cannot be in this same ball $B_v$ because the resulting out-diameter of $B_v$ would be too large, so $s''$ ends up in a different ball and the edge $(s'',s')$ is necessarily added to $\Erem$. In the second case, the vertex $s'$ ends up in the ball $B_v$ of some vertex $v$ on the cycle; in this case $(s,s')$ is added to $\Erem$. So no matter what, at least one of $(s,s')$ or $(s'',s')$ is added to $\Erem$, so one of these edges is added with probability at least $1/2$. (By contrast, if all edges were undirected, then a ball starting from the clique would explore the cycle in both directions up to some random threshold, and hence the edge $(s',s'')$ would not necessarily be added to $\Erem$.)

{\em Our approach:}
Our approach does not follow the random delay approach.
Roughly, our algorithm simulates a variation of the directed LDD algorithm~\cite{bernstein2022negative}.
The sequential algorithm carves out the graph with disjoint balls in an arbitrary order $B_{v_1}, \ldots, B_{v_\ell}$.
Doing so sequentially is inefficient in distributed models, so we instead show how to efficiently compute an index $i$, such that $B_{v_1} \cup \ldots \cup B_{v_{i-1}} \cup B_{v_i}$ and its complement are proportional in size. This yields a recursive algorithm with $O(\log(n))$ parallel rounds, and because our final ordering is mimicking a valid sequential ordering from \cite{bernstein2022negative}, we are able to argue that every edge $e$ is added to $\Erem$ with a small probability.  For more details, see \ref{subsec:LDDoverview}.

\begin{figure}[t!]
    \centering
    \begin{minipage}{.8\columnwidth}
    \centering
        \includegraphics[width=0.5\columnwidth]{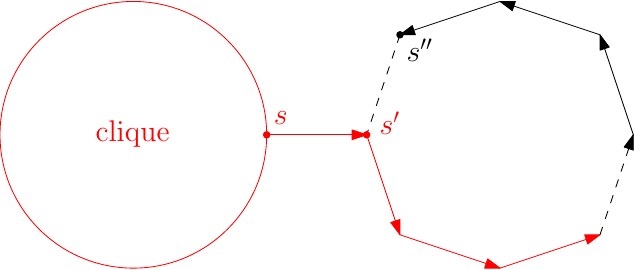}
        \captionof{figure}{The figure shows the case when one vertex in the clique starts growing a ball before all vertices in the cycle and includes $s,s'$ into the ball, while $s''$ is not included. $s''$ will be induced in another outgoing ball later. The dashed line marked the edges included in \Erem as a result.}\label{fig:bottlenecks}
    \end{minipage}%
\end{figure}

\paragraph*{Obstacle 2: Algorithms for Average Case vs Worst Case Inputs. } Recall that once we reach Phase 3 of the algorithm, only the edges of $\Erem$ can be negative; since the directed LDD guarantees that every edge is added to $\Erem$ with small probability, this implies that every shortest path contains few negative edges \emph{in expectation}.
Bernstein et al. \cite{bernstein2022negative} show a simple sequential algorithm that efficiently computes shortest paths in such a graph.
This algorithm works even when there are shortest paths with many negative edges, so long as the average number over shortest paths is small.

Unfortunately, such an algorithm does not seem possible in other settings.
Instead, we have to settle for a weaker subroutine that requires \emph{all} shortest paths to have few negative edges. (Technically speaking, it works in a general graph, but only returns correct distances to vertices $v$ for which the shortest $sv$-path has few negative edges.)
This subroutine is too weak to directly handle Phase 3 from \cite{bernstein2022negative}.
In order to execute the framework above with our weaker subroutine, we need to introduce a more refined recursive structure for \scaledown, which is the cause of our extra $n^{o(1)}$ factor in the time bounds.

\paragraph*{Obstacle 3: SCCs and Their Topological Order in \congest. }
Phase 2 of our algorithm requires computing SCCs and a toplogical ordering among them. Schudy~\cite{schudy2008finding} gives an algorithm for computing SCCs and their topological order in the parallel model that uses $O(\log^2 n)$ calls to non-negative-weight SSSP and yet, somewhat surprisingly, there has been no such algorithm formally written for \congest.
Directly porting the framework of \cite{schudy2008finding} to \congest is non-trivial; the congestion on any particular edge could be prohibitively large.
We remedy this state of affairs by implementing the framework in the \DM, which abstracts away from such low-level details and can be compiled into a \congest algorithm.

\paragraph*{Roadmap} In the next two sections, we present the following self-contained technical contributions: directed low diameter decomposition in parallel, distributed, and quantum models; and topological sort in the \congest model.
Subsequent sections are devoted to explaining how these algorithms are used to solve negative-weight SSSP in parallel, distributed, and quantum models.

\section{Low Diameter Decomposition}\label{sec:lowediameterdecomposition}

In this section, we provide the low diameter decomposition (LDD) algorithm on directed graphs with non-negative integer weight (Algorithm~\ref{alg:lowdiamterdecomposition}). We restate the lemma we want to prove in this section below.

\lddl*

\paragraph*{Organization.} 
In Section~\ref{subsec:LDDoverview}, we present a comprehensive overview of our algorithm. Section~\ref{subsec:LDDalg} is dedicated to establishing the correctness of Algorithm~\ref{alg:lowdiamterdecomposition}. A key component of this algorithm is the $FindBalancedSet$ subroutine, encapsulated within Algorithm~\ref{alg:smallsizedecomposition}, detailed in Section~\ref{sec:smallsize}. Implementation specifics of Algorithm~\ref{alg:lowdiamterdecomposition}, along with analysis of the running times for these implementations, are provided in Section~\ref{sec:ldd-imp-pc}. Finally, in Section~\ref{subsec:lddopenproblems}, we discuss some open questions related to directed low diameter decomposition.

\subsection{Algorithm Overview}\label{subsec:LDDoverview}
Our low diameter decomposition algorithm is presented in Algorithm~\ref{alg:lowdiamterdecomposition}. In this subsection, we provide an overview of Algorithm~\ref{alg:lowdiamterdecomposition}. The algorithm contains two phases:

\paragraph*{Phase 1: Mark vertices as light or heavy.} This phase is identical to the sequential algorithm introduced in \cite{bernstein2022negative}. After this phase, each vertex $v$ will get one of the following three marks: \emph{in-light}, \emph{out-light}, \emph{heavy}. It is guaranteed that w.h.p., if a vertex $v$ is marked as (i) \emph{in-light}, then $|\ingball(v,d/4)|\le .7|V|$, (ii) \emph{out-light}, then $|\outgball(v,d/4)|\le .7|V|$, (iii) \emph{heavy}, then $|\ingball(v,d/4)|> .5|V|$ and $|\ingball(v,d/4)|> .5|V|$. The algorithm for finding these labels can be summarized as follows. We select $\Theta(\log n)$ nodes from the graph uniformly at random and execute the SSSP algorithm starting from these nodes. The proportion of sampled nodes that are at a distance of no more than $d/4$ from a vertex $v$ represents the size of $\ingball(v,d/4)$, while the proportion of nodes to which $v$ is at a distance of no more than $d/4$ corresponds to the size of $\outgball(v,d/4)$.
See Algorithm~\ref{alg:lowdiamterdecomposition} Phase 1 for the details of how to get the marks, and~\Cref{cla:correctestimate} for the proof of the guarantees.

\paragraph*{Phase 2: Create sub-problems with small sizes.} We denote the set of \emph{in-light} vertices by $\Vin$, the set of \emph{out-light} vertices by $\Vout$, and the set of \emph{heavy} vertices by $\Vheavy$.
Sequentially carving our balls centered on light vertices, as in \cite{bernstein2022negative}, would not be efficient in the models we consider.
We would like to find sets which make for an efficient recursion.
To this end, we first apply subroutine $FindBalancedSet$ (Algorithm~\ref{alg:smallsizedecomposition}) on $\Vin, \Vout$. $FindBalancedSet$ on $V_{in}$ (or $V_{out}$) will create a random vertex set $A_{in}$ (or $A_{out}$) having the following properties:
\begin{enumerate}
    \item \textit{(Light boundary)} It is guaranteed that each edge $e$ is included in $\incut(\Ain)$ (or $\outcut(\Aout)$) with probability $O(w(e)\log(n)/d)$. Note that this differs from~\Cref{lem:LDDcorrectness} by a $\log n$ factor. 
    \item \textit{(Balanced or contains $V_*$)} For $*\in\{in,out\}$, we have (i) $|A_*|\le .9|V|$, and (ii) either $|A_*|\ge .1|V|$ or $V_*\subseteq A_*$.
    If $.1|V| \le |A_*| \le .9|V|$, we say $A_*$ is balanced.
    In other words, the only case that $A_*$ is not balanced (too small) is that $V_*$ is completely contained in $A_*$.
\end{enumerate}

Now we consider two cases.

\textbf{Case 1: $\Ain$ or $\Aout$ is balanced.} For convenience, we only consider the case when $\Ain$ is balanced, i.e. $.1|V|\le |\Ain|\le .9|V|$. The case where $\Aout$ is balanced is similar. In this case, we recursively call $\EremI \leftarrow LowDiameterDecomposition(G[\Ain],d)$ and $\EremII \leftarrow LowDiameterDecomposition(G[V\backslash \Ain],d)$, and return $\incut(\Ain)\cup \EremI \cup \EremII $ as $\Eref$. Now, we verify the output guarantees.

\begin{enumerate}
    \item \textit{(Time cost)} Since each recursion layer decreases the size of the graph by a constant factor, the depth of the recursion tree is bounded by $O(\log n)$.
    \item \textit{(Low diameter)} Consider an SCC $C$ of the subgraph $E-\Eref$. Since $\incut(\Ain)\subseteq \Eref$, it must be the case that $C\subseteq \Ain$ or $C\subseteq V\backslash \Ain$. In both cases, $C$ is included in a recursive call.
    \item \textit{($\Eref$ guarantee)} Each edge $e$ is included in $\incut(\Ain)$ with probability $O(w(e)\log(n)/d)$. Each edge $e$ can also be included in the returned edge set of a recursive call. The depth of the recursion tree is bounded by $O(\log n)$, therefore, an edge is included in $\Eref$ with probability $O(w(e)\log^2n/d)$.
\end{enumerate}

\textbf{Case 2: Both $\Ain,\Aout$ are not balanced.} In this case, we have $\Vin\subseteq \Ain, \Vout\subseteq \Aout$. We call  $\EremI \leftarrow LowDiameterDecomposition(G[\Ain],d)$, and $\EremII \leftarrow \\ LowDiameterDecomposition(G[\Aout\backslash \Ain],d)$, then return $\incut(\Ain)\cup \outcut(\Aout)\cup \EremI \cup \EremII $ as $\Eref$. %
Now we verify the output guarantees.

\begin{enumerate}
    \item \textit{(Time cost)} Notice that $|\Ain\cup \Aout|\le .2|V|$, thus, each recursion layer decreases the size of the graph by a constant factor; the depth of the recursion tree is bounded by $O(\log n)$.
    \item \textit{(Low diameter)} Consider an SCC $C$ of the subgraph $E-\Eref$. Since $\incut(\Ain),\outcut(\Aout)\subseteq \Eref$, it must be the case that $C\subseteq \Ain$ or $C\subseteq \Aout\backslash \Ain$ or $C\subseteq V\backslash (\Ain\cup \Aout)\subseteq\Vheavy$. In both the first two cases, $C$ is included in a recursive call. In the third case, remember that each vertex $v\in \Vheavy$ has the property that $|\ingball(v,d/4)|> .5|V|$ and $|\outgball(v,d/4)|> .5|V|$. Thus, any two vertices in $\Vheavy$ have mutual distance at most $d/2$ and so $C$ has weak diameter at most $d$.
    
    \item \textit{($\Eref$ guarantee)} Each edge $e$ is included in $\incut(\Ain)$ or $\outcut(\Aout)$ with probability $O(w(e)\log (n)/d)$. Each edge $e$ can also be included in the returned edge set of a recursive call. The depth of the recursion tree is bounded by $O(\log n)$, therefore, an edge is included in $\Eref$ with probability $O(w(e)\log^2n/d)$.
\end{enumerate}

\subsubsection{Overview of FindBalancedSet} Remember that 
FindBalancedSet takes $\Vin$ or $\Vout$ as input and outputs a set $\Ain$ or $\Aout$ that satisfies properties \textit{light boundary} and \textit{balanced} described above. For convenience, we only consider the case when $\Vin$ is the input. Write $\Vin=\{v_1,v_2,...,v_{\ell}\}$ (an arbitrary order).
The algorithm contains two steps. 

\begin{description}
    \item[Step 1.] For each $i\in\left[\ell\right]$, sample an integer $d_i$ following a certain geometric distribution. The detailed definition is given by~\Cref{def:ge}. For now, we can think of the distribution in the following way: suppose a player is repeating identical independent trials, where each trial succeeds with probability $\Theta(\frac{\log n}{d})$, then $d_i$ is the number of failed trails before the first success.

    \item[Step 2.] Find the smallest $i\in[\ell]$ such that $\left|\cup_{j\le i}\ingball(v_j,d_j)\right|>0.1|V|$, denoted as $k$. If $k$ does not exist, i.e. $\left|\cup_{j\in[\ell]}\ingball(v_j,d_j)\right|\le 0.1|V|$, set $k=\ell$. Note that for a fixed $i$, we can compute $\left|\cup_{j\le i}\ingball(v_j,d_j)\right|$ by a single SSSP call (as opposed to computing each ball sequentially, which is inefficient); we can then binary search to find $k$. Return $\cup_{j\le k}\ingball(v_j,d_j)$ as $A$.
\end{description}

\textbf{Property \emph{balanced or contains $\Vin$}.} According to the definition of $d_i$, one can show that $d_i<d/4$ w.h.p., which implies $|\ingball(v_i,d_i)|\le .7|V|$ (because $v_i$ is light.). Since $k$ is the smallest integer such that $\left|\cup_{j\le k}\ingball(v_j,d_j)\right|>0.1|V|$, it must be the case $\left|\cup_{j\le k}\ingball(v_j,d_j)\right|<0.8|V|$. Moreover, if $\left|\cup_{j\le k}\ingball(v_j,d_j)\right|>0.1|V|$ is not true, then $k=\ell$ and $\Vin\subseteq \Ain$.

\textbf{Property \emph{light boundary}.} This is the most technical part and the rest of this subsection is devoted to sketching the proof idea.%

Notice that the only randomness of $FindBalancedSet$ comes from $d_1,d_2,...,d_{\ell}$. For convenience, write $\mathbf{d}=(d_1,d_2,...,d_{\ell})$. Since $\incut(A)$ only depends on $\mathbf{d}$, we may define $\incut(A)_{\mathbf{d}}$ as the edge set $\incut(A)$ generated by the algorithm with $\mathbf{d}$ as the randomness.

To analyze the light boundary property, we will describe another algorithm that, given $\mathbf{d}=(d_1,d_2,...,d_{\ell})$, outputs an edge set $E_{\mathbf{d}}$, such that 
\begin{enumerate}
    \item $\incut(A)_\mathbf{d}\subseteq E_{\mathbf{d}}$ always holds for any $\mathbf{d}$, and
    \item an edge $e$ is included in $E_{\mathbf{d}}$ with probability $O(w(e)\log(n)/d)$.
\end{enumerate}

\textit{(The algorithm to generate $E_{\mathbf{d}}$)} Initially set $E_{\mathbf{d}}=\emptyset$. For iterations $i=1,2,...,\ell$, do
\begin{itemize}
    \item Mark all edges $(u,v)$ with $u,v\in \ingball(v_i,d_i)$ as "invulnerable", and add all edges in $\incut(\ingball(v_i,d_i))$ that are not invulnerable to $E_{\mathbf{d}}$. 
\end{itemize}

Note that $E_{\mathbf{d}}$ is produced by a sequential algorithm, which is easier to analyze, but the algorithm never actually computes $E_{\mathbf{d}}$.
We can show that $\incut(A)\subseteq E_{\mathbf{d}}$ is always true: Recall that $A=\cup_{j\le k}\ingball(v_j,d_j)$. Any edge in $\incut(A)$ is not invulnerable before the end of the $k$-th iteration; any edge in $\incut(A)$ is also on the boundary of some $B_j$ for $j\le k$, which means it has already been added to $E_{\mathbf{d}}$ before the end of the $k$-th iteration.

The last thing is to show that an edge $e$ is included in $E_{\mathbf{d}}$ with probability $O(w(e)\log(n)/d)$. To this end, consider the following alternative explanation of the procedure when we do the $i$-th iteration: $v_i$ gradually grows the radius of the ball centered on $v_i$, each round increases the radius by $1$, and stops with probability $\Theta(\log(n) / d)$. This is exactly how $d_i$ is defined. Observe that each edge $(u,v)$ will be included in $E_{\mathbf{d}}$ if and only if the first $v_i$ that grows its ball to reach $v$ failed to reach $u$ (if it reached $u$, then this edge is marked as invulnerable and will never be added to $E_{\mathbf{d}}$).
By the memoryless property of the geometric distribution, this happens with probability $\Theta(w(e)\cdot \log(n)/d)$. 

\begin{algorithm}
	\caption{$\Eref\leftarrow LowDiameterDecomposition(G,d)$}\label{alg:lowdiamterdecomposition}
	\KwData{Non-negative weighted directed graph $G=(V,E,w)$, an integer $d$.}
	\KwResult{A random set of edges $\Eref\subseteq E$. (See~\Cref{lem:LDDcorrectness} for the properties of the output.)}
 
        \BlankLine
        If $G$ is an empty graph, return $\emptyset$\;
        Let $n$ and $c$ be defined as in~\Cref{rem:nandc}\;
        \tcp*[h]{\textcolor{blue}{Phase 1: mark vertices as light or heavy}}
        
	Sample $\lceil c\log n\rceil$ vertices in $V$ uniformly at random, denoted as $S$\;\label{LDDline2}
	For each $v\in S$, use $\Oracle(G,v)$ to find $\ingball(v,d/4)$ and $\outgball(v,d/4)$\;\label{line:ldd_s-balls}
 
        For each $v \in V$, compute $\ingball(v,d/4) \bigcap S$ and $\outgball(v,d/4) \bigcap S$ using Line~\ref{line:ldd_s-balls}\;

            \ForEach(\label{line:ldd-marking-loop}){$v \in V$}{
            
            If $|\ingball(v,d/4) \bigcap S| \leq .6 |S|$, mark $v$ \emph{in-light} \tcp*[f]{whp $|\ingball(v,d/4)| \leq .7|V(G)|$}
                
            Else if $|\outgball(v,d/4) \bigcap S| \leq .6|S|$, mark $v$ \emph{out-light}\tcp*[f]{whp $|\outgball(v,d/4)| \leq .7|V(G)|$}
                
            Else mark $v$  \emph{heavy} \tcp*[f]{whp $|\ingball(v,D/4)| > .5|V(G)|$ and $|\outgball(v,D/4)| > .5|V(G)|$}
                
            }\label{LDDline3}
	
        \BlankLine
        \tcp*[h]{\textcolor{blue}{Phase 2: creates sub-problems with small sizes}}
        
        Denote the set of \emph{in-light} vertices by $\Vin$, the set of \emph{out-light} vertices by $\Vout$\;
	$\Ain\leftarrow FindBalancedSet(G,\Vin,d,in)$, $\Eremin \leftarrow \incut(\Ain)$\;\label{LDDline4}
	$\Aout\leftarrow FindBalancedSet(G,\Vout,d,out)$, $\Eremout \leftarrow\outcut(\Aout)$\;\label{LDDline5}

        \tcp*[h]{\textcolor{blue}{Case 1: One of $\Ain,\Aout$ is balanced.}}
        
	\If{$A_*$ ($*$ can be $in$ or $out$) has size between $.1|V|$ and $.9|V|$}
        {
            $\EremI \leftarrow LowDiameterDecomposition(G[A_*],d)$\; \label{recurse1}
            $\EremII \leftarrow LowDiameterDecomposition(G[V\backslash A_*],d)$\;\label{recurse2}
            \Return{$\Erems\bigcup \EremI \bigcup \EremII $}\;\label{return1}
        }

        \tcp*[h]{\textcolor{blue}{Clean up: Check that $V\backslash(\Ain\bigcup \Aout)$ have small weak diameter.}}
        
	Pick an arbitrary vertex $u\in V\backslash(\Ain\cup \Aout)$. Use $\Oracle(G,u)$ to find $\ingball(u,d/2),\outgball(u,d/2)$\;
        
        \If{$V\backslash(\Ain\cup \Aout)\not\subseteq \ingball(u,d/2)\bigcap\outgball(u,d/2)$ or $|\Ain\cup \Aout|\ge .5|V|$}
        { \Return $E$}\label{LDDlineborder}

        \tcp*[h]{\color{blue}Case 2: both $\Ain,\Aout$ are small.}
        
        $\EremI \leftarrow LowDiameterDecomposition(G[\Ain],d)$\;\label{recurse3}
        $\EremII \leftarrow LowDiameterDecomposition(G[\Aout\backslash \Ain],d)$\;\label{recurse4}
            
	\Return{$\Eremin \bigcup \Eremout \bigcup \EremI \bigcup \EremII $}\;\label{return3}
\end{algorithm}

\subsection{Correctness of the Algorithm}\label{subsec:LDDalg}

In this secion, we will show the correctness of Algorithm~\ref{alg:lowdiamterdecomposition}.
\begin{remark}\label{rem:nandc}
Throughout this section, we use $n$ to denote a global variable which always refers to the size of the graph in which we call low diameter decomposition. Introducing the parameter $n$ ensures that in the analysis "with high probability" is in terms of $n$. In addition, we always use $c$ to denote a sufficiently large constant.
\end{remark}

\begin{proof}[Proof of~\Cref{lem:LDDcorrectness} (correctness)]
One can verify that each recursion will decrease the size of the graph by at least $1$. Therefore, the algorithm will terminate.

We use induction on the size of the input graph $G$ to show that the following statement is true, thus proving the lemma.

\paragraph*{Induction hypothesis.} The output of Algorithm~\ref{alg:lowdiamterdecomposition} satisfies
	\begin{enumerate}
            \item each SCC of the subgraph $E-\Eref$ has weak diameter at most $d$ in $G$, i.e. if $u,v$ are two vertices in the same SCC, then $dist_G(u,v)\le d$ and $dist_G(v,u)\le d$.
		\item for any $e\in E$, we have $\Pr[e\in \Eref]=\frac{c^3w(e)\log|V|\log n}{d}+\frac{|V|}{n^9}$.
	\end{enumerate}
	\paragraph*{Base case.} When $G$ contains $0$ vertices, the algorithm returns $\Eref=\emptyset$, the induction hypothesis holds. 
	
	\paragraph*{Induction.} 
        We first prove (2). There are two possibilities for $\Eref$: either $\Eref\subseteq \Eremin \bigcup \Eremout \bigcup \EremI \bigcup \EremII $ or $\Eref=E$. 
        According to~\Cref{lem:correctnessLowsizedecomposition}, each edge $e$ is included in $\Eremin $ or $\Eremout $ with probability $\frac{2c^2w(e)\log n}{d}$. 
        
        We first need the following claims to bound the size of the recursive call. 
        
	\begin{claim}\label{cla:correctestimate}
            With high probability in $n$, for any $v\in \Vin$, we have $|\ingball(v,d/4)|\le .7|V|$; for any $v\in \Vout$, we have $|\outgball(v,d/4)|\le.7|V|$; for any $v\in V\backslash(\Vout\cup \Vin)$, we have $|\ingball(v,d/4)|,|\outgball(v,d/4)|>.5|V|$.
        \end{claim}
        \begin{proof}[Proof of~\Cref{cla:correctestimate}]
            Notice that $S$ contains $\lceil c\log n\rceil$ vertices sampled uniformly at random. If a vertex $v$ with $|\ingball(v,d/4)|>.7|V|$ is included in $\Vin$, that means $|\ingball(v,d/4) \bigcap S| \leq .6 |S|$. However, the expectation of the number of vertices in $\ingball(v,d/4) \bigcap S$ is at least $.7|S|$ since $S$ is uniformly sampled. Using a simple Chernoff bound, by taking $c$ sufficiently large, the event does not happen with high probability. The same arguments hold for vertices in $\Vin, V\backslash(\Vout\cup \Vin)$.
        \end{proof}

	\begin{claim}\label{cla:sizeofrecursive}
		With high probability in $n$, we have $|\Ain|,|\Aout|\le 0.9|V|$. 
	\end{claim}
 
	\begin{proof}[Proof of~\Cref{cla:sizeofrecursive}]
		According to~\Cref{lem:correctnessLowsizedecomposition}, we just need to prove that with high probability in $n$, for any $v\in \Vout$, we have $|\ingball(v,d/4)|\le.7|V|$ and for any $v\in \Vin$, we have $|\outgball(v,d/4)|\le.7|V|$. This can be deduced by~\Cref{cla:correctestimate}.

	\end{proof}

        The following claim shows that $\Eref=E$ happens with a small probability.
        \begin{claim}\label{cla:nothappen}
            With high probability in $n$, line~\ref{LDDlineborder} is not executed.
        \end{claim}
        
        \begin{proof}[Proof of~\Cref{cla:nothappen}]
            If line~\ref{LDDlineborder} is executed, then both $\Ain,\Aout$ has size not between $.1|V|$ and $.9|V|$ (algorithm does not return in line~\ref{return1}) and either we have $V\backslash(\Ain\cup \Aout)\not\subseteq \ingball(u,d/2)\bigcap\outgball(u)$, or we have $|\Ain\cup \Aout|\ge.5|V|$. According to~\Cref{cla:sizeofrecursive}, with high probability in $n$, we have $|\Ain|,|\Aout|\le .1|V|$, which means $|\Ain\cup \Aout|<.5|V|$. According to~\Cref{lem:correctnessLowsizedecomposition} item (2), we have $\Vin\subseteq \Ain, \Vout\subseteq \Aout$. Thus, $V\backslash(\Ain\cup \Aout)\subseteq V\backslash(\Vin\cup \Vout)$. According to~\ref{cla:correctestimate}, with high probability in $n$, for any two vertices $u,v\in V\backslash(\Vout\cup \Vin)$, we know both $u$ and $v$ can reach and can be reached by at least $.5|V|$ vertices within distance $d/4$. Therefore, $u$ and $v$ can reach each other within distance $d/2$, which means $V\backslash(\Ain\cup \Aout)\subseteq \ingball(u,d/2)\bigcap\outgball(u)$.
        \end{proof}
	Now we are ready to compute $\Pr[e\in \Eref]$. Notice that each edge can be included in at most $1$ recursive call. The following inequality bounds the total probability of an edge being in $\Eref$. The first term is according to the induction hypothesis, the second term is the probability of being included in $\EremI ,\EremII $, the last term is the small failing probability of \Cref{cla:sizeofrecursive}, and the small failing probability of~\Cref{cla:nothappen}.
	\[\left(\frac{c^3w(e)\log\left(0.9|V|\right)\log n}{d}+\frac{0.9|V|}{n^9}\right)+\frac{2c^2w(e)\log n}{d}+\frac{1}{n^9}\le\frac{c^3w(e)\log(|V|)\log n}{d}+\frac{|V|}{n^9}\]
	
	for sufficiently large $c$. 
	
	Then we prove (1). Consider an SCC $C$ of the subgraph $E-\Eref$. If the algorithm returns by line~\ref{return1}, since $\incut(\Ain)\subseteq \Eref$, it must be the case that $C\subseteq \Ain$ or $C\subseteq V\backslash \Ain$. In both cases, $C$ is included in a recursive call and $C$ has small weak diameter in the induced subgraph according to the induction hypothesis, which also holds in the original graph. If the algorithm return by line~\ref{LDDlineborder}, then each SCC is a single node. If the algorithm returns by line~\ref{return3}, since $\incut(\Ain),\outcut(\Aout)\subseteq \Eref$, it must be the case that $C\subseteq \Ain$ or $C\subseteq \Aout\backslash \Ain$ or $C\subseteq V\backslash (\Ain\cup \Aout)\subseteq\Vheavy$. In both the first two cases, $C$ is included in a recursive call. In the third case, remember that each vertex $v\in \Vheavy$ has the property that $|\ingball(v,d/4)|> .5|V|$ and $|\outgball(v,d/4)|> .5|V|$. Thus, any two vertices in $\Vheavy$ have mutual distance at most $d/2$, $C$ has weak diameter at most $d$.
\end{proof}

In order to bound the number of oracle calls, we use the following lemma to bound the depth of recursion. 

\begin{lemma}\label{lem:depthofrecursion}
    The recursion depth of $LowDiameterDecomposition(G=(V,E),d)$ (Algorithm~\ref{alg:lowdiamterdecomposition}) is $O(\log |V|)$. Any recursion call is on an induced subgraph of $G$, and any two recursion calls in the same recursion layer are on vertex-disjoint induced subgraphs.
\end{lemma}

\begin{proof}[Proof of \Cref{lem:depthofrecursion}]
    Consider an execution $LowDiameterDecomposition(G[V'],d)$. If the algorithm enters line~\ref{recurse1}, then we know $|A_*|\le .9|V'|$ and $|V\backslash A_*|\le .9|V'|$. If the algorithm enters line~\ref{recurse3}, then we konw that $|\Ain\cup \Aout|<.5|V'|$. In other words, the maximum number of vertices in all recursion calls in the next layer is at most $0.9$ times the size of the previous layer. Thus, recursion ends at the $O(\log n)$-th layer.
    
    One can verify that each LowDiameterDecomposition($G,d$) generates two recursive calls on subgraphs induced by either $A_*, V\backslash A_*$ or $\Ain, \Aout\backslash \Ain$, where each pair is trivially vertex disjoint. Thus, any two recursion calls in the same recursion layer are on vertex-disjoint induced subgraphs.

\end{proof}

The running time of~\Cref{lem:LDDcorrectness} is proved in~\Cref{sec:ldd-imp-pc}.

\subsection{Find Balanced Set}\label{sec:smallsize}
\begin{definition}[Truncated geometric distribution]\label{def:ge}
	We say $x$ follows the geometric distribution with parameter $p\in(0,1)$ truncated at $t\in\mathbb{N}$, denoted by $x\sim GE[p]_{\le l}$, if $x\in\mathbb{N},x \leq t$ and $\Pr[x=k]=(1-p)^{k}p\cdot \frac{1}{1-(1-p)^{t+1}}$.
\end{definition}

\begin{algorithm}
\caption{$A\leftarrow FindBalancedSet(G,V',d,*)$}\label{alg:smallsizedecomposition}
\KwData{Non-negative weighted directed graph $G=(V,E,w)$, a vertex set $V'\subseteq V$ and an integer $d$ satisfying $|\sgball(v,d/4)|\le .7|V|$ for any $v\in V'$.}
\KwResult{A set of vertices $A\subseteq V$ satisfying~\Cref{lem:correctnessLowsizedecomposition}.}
Suppose $V'=\{v_{1},v_{2},...,v_{\ell}\}$. Each vertex $v_{i}$ samples $d_i\sim GE[\min\left((c\log n)/d,1\right)]_{\le \lfloor d/4\rfloor}$ (see~\Cref{def:ge}) \;\label{SSline1}
Find the smallest $i\in[\ell]$ such that $\left|\cup_{j\le i}\sgball(v_j,d_j)\right|>0.1|V|$, denoted as $k$. If $k$ does not exist, i.e. $\left|\cup_{j\in[\ell]}\sgball(v_j,d_j)\right|\le 0.1|V|$, set $k=\ell$. (See~\Cref{rem:binarysearch} for implementation)\;
 \label{SSline2}

\Return $\cup_{j\le k}\sgball(v_{j},d_j)$, \;\label{SSline3}
\end{algorithm}

\begin{remark}\label{rem:binarysearch}
    Line~\ref{SSline2} can be implemented in the following way by calling $O(\log n)$ times of $\Oracle$. Note that the function $f(i)=\left|\cup_{j\le i}\sgball(v_j,d_j)\right|$ is an increasing function. To find the smallest $i\in[\ell]$ such that $f(i)>0.1|V|$, we can use binary search, which requires $O(\log n)$ queries to the value of $f(i)$. To compute $f(i)$, for simplicity, we assume the input graph is $A_{in}$. Let $G'=(V\cup\{s\},E\cup\{(v_j,s)\mid j\le i\},w\cup w')$ where $w'((v_j,s))=d-d_j$, then we call $\Oracle(G',s)$ to compute $f(i)$ and one can verify that $f(i)=|\ingball(s,d)|-1$. 
\end{remark}
\begin{lemma}\label{lem:correctnessLowsizedecomposition}
	If the input of Algorithm~\ref{alg:smallsizedecomposition} satisfies $|\sgball(v,d/4)|\le .7|V|$ for any $v\in V'$, then the outputs satisfy
	\begin{enumerate}[itemindent=1em, label=(\ref{lem:correctnessLowsizedecomposition}\alph*)]
		\item for any $e\in E$, we have $\Pr[e\in \Eref]\le\frac{c^2w(e)\log n}{d}$,\label{lem:correctnessldditem1}
		\item either $|A|>0.1|V|$, or $V'\subseteq A$,\label{lem:correctnessldditem2}
		\item $|A|\le .9|V|$,\label{lem:correctnessldditem3}
	\end{enumerate}
\end{lemma}
\begin{proof}[Proof of~\Cref{lem:correctnessLowsizedecomposition}]

	We first prove~\ref{lem:correctnessldditem1}. If $p=1$, then $d<c\log n$, which means $\frac{cw(e)\log n}{d}\ge 1$ as long as $w(e)>0$. In this case, any edge $e$ with non-zero weight satisfies~\ref{lem:correctnessldditem1}. Also notice that zero-weight edges will never be added to $\Eref$. In the following arguments we only consider the case when $p<1$.
	
	For each node $u\in V$, let $I_u$ denote the smallest $i$ such that $u\in B^+_{d_i}(v_{i})$. If such $i$ does not exist, let $I_u=\ell+1$. Consider an edge $e=(u,v)$. If $e\in \Eref$, we first argue that $I_u<I_v$: according to the algorithm description, there must exists $k\in[\ell]$ such that $u\in \cup_{j\le k}B^+_{d_j}(v_j)$, in which case we have $I_u\le k$ and $I_v>k$. Thus, $I_u<I_v$ must hold. Now we focus on bounding the probability of $I_u<I_v$. 
	
	Denote event $A_i$ as $d_i-w(e)<dist(v_i,u)\le d_i$ We have
	
	\[\Pr[I_u<I_v]\le \sum_{i\in[\ell]}\Pr[A_i,I_u=i]=\sum_{i\in[\ell],\Pr[I_u=i]\not=0}\Pr[I_u=i]\cdot \Pr[A_i\mid I_u=i].\]
	Explanation:  If $I_u<I_v\le \ell+1$, then $I_u=i$ must happen for some $i\in[\ell]$, which means $dist(v_i,u)\le d_i$. If $dist(v_i,u)\le d_i-w(e)$, since $(u,v)$ is an edge with weight $w(e)$, we have $dist(v_i,v)\le d_i$, contradicting the fact that $I_u<I_v$. Thus, $A_i$ must happen. 
	
	For each $i$ with $dist(v_i,u)\ge d/8$, we claim that $\Pr[I_u=i]\le\frac{1}{n^9}$. That is because $I_u=i$ implies $d_i\ge dist(v_i,u)\ge d/8$. Remember that $d_i\sim GE[(c\log n)/d]_{\le d/4}$, which means $\Pr[d_i\ge d/8]\le\frac{1}{n^9}$, for sufficiently large $c$. Therefore, we can write
	
	\[\Pr[I_u<I_v]\le \sum_{i\in[\ell],\Pr[I_u=i]\not=0,dist(v_i,u)<d/8}\Pr[I_u=i]\cdot \Pr[A_i\mid I_u=i]+\frac{1}{n^8}.\]
	
	In what follows, we will prove $\Pr[A_i\mid I_u=i]\le 2pw(e)$ for any $i\in[\ell],\Pr[I_u=i]\not=0,dist(v_i,u)<d/8$.
	
	Let $S$ contain all the tuples $\mathbf{d}=(d'_{1},d'_{2},...,d'_{i-1})$ such that $\Pr[d_1=d'_1,d_2=d'_2,...,d_{i-1}=d'_{i-1}]\not=0$ and $u\not\in \cup_{j\le i-1}B^+_{d'_{j}}(v_j)$. One can see that if $I_u=i$ happens, then $d_{1},...,d_{i-1}$ must get value $\mathbf{d}$ for some $\mathbf{d}\in S$. Denote this event as $D_\mathbf{d}$. Thus, we have
	
	\[\Pr[A_i\mid I_u=i]=\sum_{\mathbf{d}\in S}\Pr[D_\mathbf{d},A_i\mid I_u=i]=\sum_{\mathbf{d}\in S}\Pr[D_{\mathbf{d}}\mid I_u=i]\cdot \Pr\left[A_i\mid I_u=i, D_\mathbf{d}\right].\]
	
	We will bound the probability $\Pr\left[A_i\mid I_u=i, D_\mathbf{d}\right]$ for any $\mathbf{d}$ and for any $i\in[\ell],\Pr[I_u=i]\not=0,dist(v_i,u)<d/8$. Notice that event $I_u=i,D_\mathbf{d}$ is equivalent to event $dist(v_i,u)\le d_i,D_\mathbf{d}$. Thus, we get
	\[\Pr\left[A_i\mid I_u=i, D_\mathbf{d}\right]=\Pr\left[A_i\mid dist(v_i,u)\le d_i, D_\mathbf{d}\right]=\frac{\Pr\left[d_i-w(e)<dist(v_i,u)\le d_i\mid D_\mathbf{d}\right]}{\Pr\left[dist(v_i,u)\le d_i\mid D_\mathbf{d}\right]}.\]
	
	Notice that $D_{\mathbf{d}}$ is independent of the random variable $d_i$, thus, the above term equals to

\begin{align*}
        \frac{\Pr\left[d_i-w(e)<dist(v_i,u)\le d_i\right]}{\Pr\left[dist(v_i,u)\le d_i\right]}.
\end{align*}

	Remember that $i$ is an index such that $\Pr[I_u=i]\not=0$, which means $dist(v_i,u)\le \lfloor d/4\rfloor$ (otherwise, since $d_i$ is a geometric distribution truncated at $\lfloor d/4\rfloor$, it is impossible that $I_u=i$). Thus, the above term is at most
 \begin{align*}
     \frac{\sum_{k\in[dist(v_i,u),dist(v_i,u)+w(e)]}\Pr[d_i=k]}{\sum_{k\ge dist(v_i,u)}\Pr[d_i=k]}&\le\frac{w(e)\cdot (1-p)^{dist(v_i,u)}\cdot p}{(1-p)^{dist(v_i,u)}-(1-p)^{\lfloor d/4\rfloor+1}}\\
	&=\frac{w(e)p}{1-(1-p)^{\lfloor d/4\rfloor+1-dist(v_i,u)}}\le 2w(e)p.
 \end{align*}

	By combining everything together, we get
	
	\[\Pr[(u,v)\in \Eref]\le\Pr[I_u<I_v]\le 2w(e)p+\frac{1}{n^8}\le 3w(e)p\le \frac{c^2w(e)\log n}{d}.\]

	Then we prove \ref{lem:correctnessldditem2}. Recall that $k$ is the smallest $i\in[\ell]$ such that $\left|\cup_{j\le i}\sgball(v_j,d_j)\right|>0.1|V|$ if such $i$ exists, in which case we have $|A|>0.1|V|$; or we set $k=\ell$, in which case we have $V'\subseteq A$.
	
	Then we prove \ref{lem:correctnessldditem3}. We can see that $d_i\le d$ holds for any $i\in[\ell]$. According to the input guarantee, we have $|B^+_{d_i}(v_i)|\le.7|V|$ for any $i\in[\ell]$. Since we have $\left|\cup_{j< k}\sgball(v_j,d_j)\right|\le 0.1|V|$ and $|B^+_{d_k}(v_k)|\le .7|V|$, we have $|A|\le 0.1|V|+0.7|V|\le.9|V|$.
    \end{proof}

    \subsection{Implementation in Various Models (Proof of~\Cref{lem:LDDcorrectness})}
    \label{sec:ldd-imp-pc}

In this section, we will give the implementation details of Algorithm~\ref{alg:lowdiamterdecomposition}. We start with the implementation of the PRAM model.
    
    \subsubsection{Proof of \Cref{lem:LDDcorrectness} Parallel Running Time} 
    
    In Phase 1, we sample $\tilde{O}(1)$ nodes and call $\Oracle$ for sampling nodes. In Phase 2, $FindBalancedSet$ calls $\Oracle(G, s)$ $\tilde{O}(1)$ times. Finally, we will recurse the whole process on the subproblem. Based on \Cref{lem:depthofrecursion}, in each level of recursion, each node will be only in one subproblem and the recursion depth is at most $\tilde{O}(1)$. Combining them gives us \Cref{lem:LDDcorrectness}. 

    \subsubsection{Proof of \Cref{lem:LDDcorrectness} \congest{} Running Time}   
    \paragraph*{Definition of $\Oraclem$. }There is a challenge when we call $\Oracle$ in \congest{} model. For example,~\Cref{rem:binarysearch} is one of the places where we call $\Oracle$, but one can notice that we call $\Oracle$ on a graph with a super source $s$ connecting to many vertices in our communication network. Running SSSP in this new graph cannot be trivially simulated by \congest{} algorithm in the original network, since the new edges cannot transfer information in the original graph. Thus, we define the following oracle which allows the SSSP to start with a super source.
    
    \begin{definition}\label{def:oraclem}
        The oracle $\Oraclem(G,S,x)$ has inputs (i) $G=(V,E,w)$ is a directed graph with  polynomially bounded weighted function $w:E\to\mathbb{N}$, (ii) $S\subseteq V$ specifies the vertices that the super source is connected to, (iii) $x:S\to\mathbb{N}$ specifies the weight of edges from the super source to each vertex in $S$. $\Oraclem(G,S,x)$ returns a distance vector $(d_v)_{v\in V}$ defined as follows. Let $G'=(V\cup\{s\},E\cup\{(s,v)\mid v\in S\}, w\cup w')$ where $w'((s,v))=x(v)$ for each $v\in S$, then $d_v=dist_{G'}(s,v)$. 
    \end{definition}
    
    The following theorem is from Rozhoň \etal~\cite{abs-2210-16351}.

    \begin{theorem}[Collary 1.9 of Rozhoň \etal~\cite{abs-2210-16351}]\label{thm:distributedNNSSSPalgorithm}
        There exists a randomized distributed algorithm that solves SSSP with non-negative polynomially bounded integer edge weight 
    in $\tOh{n^{1/2}+D+n^{2/5+o(1)}D^{2/5}}$ rounds, where D denotes the undirected hop-diameter, and works with high probability.
    \end{theorem}
    Using \Cref{thm:distributedNNSSSPalgorithm}, we can answer $\Oraclem(G, S, x)$ in $\tOh{n^{1/2}+D+n^{2/5+o(1)}D^{2/5}}$ rounds. The only difficulty comes from the fact we have a virtual source for $\Oraclem$. In Rozhoň \etal~\cite{abs-2210-16351}, first, they construct a new graph $\hat{G}$ by adding another virtual source to the input graph. Then they reduce the exact SSSP to approximate SSSP on the graph with one virtual source. Note that in our case, we add the virtual source to graph $G'$, where $G'$ is the graph defined in \Cref{def:oraclem} and already has one virtual source. Fortunately, after adding another virtual source to $G'$, our $\hat{G}$ can have only one virtual source by combining edges going through $s$. The exact SSSP is still reduced to the approximate SSSP on the graph with one virtual source. This gives us the following theorem,
    \begin{theorem}\label{thm:distributedNNSSSP}
        There exists a randomized distributed algorithm answering $\Oraclem$ in \\ $\tOh{n^{1/2}+D+n^{2/5+o(1)}D^{2/5}}$ rounds, where D denotes the undirected hop-diameter, and works with high probability.
    \end{theorem}

    For \congest{} model, we will prove the following corollary, which reveals the \DM{} implementation of Algorithm~\ref{alg:lowdiamterdecomposition}.
    
    \begin{corollary};\label{cor:lddcongestminor}
        There is an algorithm given a directed polynomially boundeded positive integer weighted graph $G=(V,E,w)$, that computes a low diameter decomposition (as described in~\Cref{lem:LDDcorrectness}) within $\tOh{1}$ rounds in the \DM , $\tOh{1}$ rounds in the \congest model and $\tOh{1}$ times of $\Oraclem$ calls.
    \end{corollary}
    \begin{proof}
        
        The algorithm contains $O(\log n)$ recursive layers. We will consider implementing recursive instances in each layer simultaneously, in $\tOh{1}$ steps in \DM{} and $\tOh{1}$ times of $\Oraclem$ calls. To avoid confusion, we always use $G$ to denote the original graph that we want to do low diameter decomposition on. Before each recursive layer, we assume we know the inputs for these recursive calls are induced subgraphs $G[V_1],...,G[V_{\ell}]$, and at the end of the recursive layer, the input to the next recursive layer $G[V_1],...,G[V_{\ell'}]$ is computed.
        
        We first consider Phase 1. Suppose in this layer, the recursive instances are on vertex sets $V_1,V_2,...,V_{\ell}$. We first contract all edges $(u,v)$ with both end points in the same set $u,v\in V_i$ (we will guarantee later that each vertex set $V_i$ is connected). In this way, each $V_i$ becomes a node in the minor graph. To sample $\lceil c\log n\rceil$ vertices, each vertex uniformly at random samples an integer in $[n^9]$, then use binary search to find the threshold $t$ where there are exactly $\lceil c\log n\rceil$ vertices that have sampled integer greater than $k$. Counting the number of vertices that have sampled integers greater than $k$ can be done by one aggregation step in each $V_i$. It might be the case that there is no such threshold $k$ (several vertices get the same sample integer), in which case we stop the whole algorithm and output $E$ as $\Erem$ (such an event happens with probability at most $1/n^8$, so it will not affect the final probability too much). Now each vertex knows whether it is in $S$ or not. Other lines of phase 1 can be done by $O(1)$ times of calling to $\Oraclem$ inside each $G[V_i]$ (remember that $G[V_i]$ is the input graph for one recursive call, and other lines of Phase 1 requires only finding balls using $\Oraclem$ calls). We claim that one $\Oraclem$ on 
    $G$ suffices to simulate $\ell$ calls to $\Oraclem$ on each of $G[V_i]$: set the 
     weight of edges inside each $G[V_i]$ to be the corresponding weight of edges determined by $\Oraclem$, other edges to have infinite (large enough) weight, and the source set is the combination of all the source set in each $V_i$. In this way, any path crossing different $V_i$ cannot be the shortest path, thus, the distances are correctly computed for each $G[V_i]$. 
    
        Then we consider Phase 2. First, we need to show the implementation of FindBalancedSet (Algorithm~\ref{alg:smallsizedecomposition}). Similar to the implementation above, for each $V_i$, we can get the arbitrary order $v_1,...,v_\ell$ by sampling from $[n^9]$ for each vertex, and the set $\{v_1,v_2,...,v_i\}$ can be found by binary search the threshold $k$ where there are exactly $i$ vertices that has sampled integer greater than $k$. Another part of FindBalancedSet is $O(\log n)$ calls to $\Oraclem$ inside each $V_i$. We already showed how to implement this by $\tOh{1}$ times of $\Oraclem$ calling on $G$. After two calls to FindBalancedSet, each node knows whether it is in $\Ain,\Aout$ or not.
    
        For case 1, we first need to count the size $|\Ain|,|\Aout|$, which can be done by one aggregation step: contracting all edges with both endpoints inside $\Ain$. To do the recursive call, i.e., to enter the next layer of the recursion, we need to find the recursive vertex sets $V_1,...,V_{\ell'}$ in the next recursive layer. we achieve this by contracting edges inside induced subgraph $G[A_*], G[V\backslash A_*]$. Notice that after contracting these edges, there are not necessarily two connected components in the next recursive layer: $V\backslash A_*$ could be unconnected, which means several instances will be created. However, this does not affect the outcome of our algorithm, since each edge in $\Eref$ is included in one of the recursive instances or $\Eref_{*}$, and each SCC is completely included in one of the recursive instances.
    
        For "clean up", the arbitrary vertex $u$ can be picked by one aggregation step, $\ingball(u,d/2),\outgball(u,d/2)$ can be found by one $\Oraclem$ call on $G$, other vertex sets sizes computation can be done by aggregation steps.
    
        For case 2, to do the recursive call, we will contract edges inside induced subgraph $G[\Ain], G[\Aout\backslash \Ain]$. The same problem happens: $\Aout \backslash \Ain$ does not necessarily be connected. We can recurse on each connected component of $\Aout\backslash \Ain$ which will not affect the output, as we have already argued above.
    \end{proof}

    By combining~\Cref{cor:lddcongestminor,thm:minoraggregation}, we get~\Cref{lem:LDDcorrectness} \congest{} running time.
    
    By combining~\Cref{thm:distributedNNSSSP}, we can get the following corollary.
    \begin{corollary}\label{cor:lddcongest}
        There is a \congest algorithm given a directed graph $G=(V,E,w)$ with polynomially bounded non-negative weights, computes a low diameter decomposition (as described in~\Cref{lem:LDDcorrectness}) within $\tOh{n^{1/2}+D+n^{2/5+o(1)}D^{2/5}}$ rounds. 
    \end{corollary}

    \subsubsection{Proof of~\Cref{lem:LDDcorrectness} Quantum Query Running time}\label{subsubsec:lddquantum} The output of $LowDiameterDecomposition$ is $E^{rem}$, which can have size $m$, too large for quantum query model. Thus, instead of returning $E^{rem}$, the algorithm will return a partition of $V$ denoted as $V_1,V_2,...,V_\ell$, and $E^{rem}$ contains all the edges $(u,v)$ where $u\in V_i,v\in V_j$ for some $i<j$. We need to re-write~\Cref{lem:LDDcorrectness} as follows.

    \begin{restatable}[Low-Diameter Decomposition in quantum query model]{lemma}{lddl}\label{lem:LDDquantum}

Let $G=(V,E,w)$ be a directed graph with a polynomially bounded weight function $w:E\to\mathbb{N}$ and let $d$ be a positive integer. There
 exists a quantum query algorithm $\LDD(G,d)$ with following guarantees:
    \begin{itemize}
        \item INPUT: An $n$-node $m$-edge, graph $G = (V, E, w)$ with non-negative integer edge weight and a positive integer $d$.
        \item OUTPUT: (proved in~\Cref{subsec:LDDalg}) a partition of $V$ denoted as $V_1,...,V_\ell$.
    satisfying:
        \begin{itemize}
            \item Each $V_i$ has weak diameter at most $d$ in $G$, i.e. if $u,v$ are two vertices in $V_i$, then $dist_G(u,v)\le d$ and $dist_G(v,u)\le d$.
            \item For any $(u,v)\in E$, we have $\Pr[u\in V_i,v\in V_j,i<j]=O\left(\frac{w(e)\log^2n}{d}+\frac{1}{n^8}\right)$.
        \end{itemize}
        \item RUNNING TIME: Assuming there exists a quantum edge query algorithm answering SSSP in $Q(m,n)$ rounds, then $\LDD(G,d)$ takes $\tilde{O}(Q(m,n))$ queries with high probability.
    \end{itemize}
\end{restatable}

    We need to show how to do phases 1, 2, ''clean up'' phase, and $FindBalancedSet$ in the quantum query model. 
    
    Phase 1 makes $O(\log n)$ calls to $\Oracle$. After these calls, we can get $\ingball(v,d/4) \bigcap S$ and $\outgball(v,d/4) \bigcap S$. After that, other operations do not need to do any query to the graph. After Phase 1, we successfully identified two sets $\Vin,\Vout$.

    Phase 2 first makes two calls to $FindBalancedSet$ (which we will elaborate on later) to get $\Ain,\Aout$, and then makes recursive calls to $LowDiameterDecomposition$ on two disjoint vertex sets $A_*,V\backslash A_*$. Then will return a vertex partition on $A_*$ and $V\backslash A_*$ separately, which combine together to a vertex partition of $V$. The returned partition of $V$ follows the following order (i) we put the partition of $A_*$ before $V\backslash A_*$ if $*=out$, (ii) we put the partition of $V\backslash A_*$ before $A_*$ if $*=in$. It is easy to see that the edges connected former sets to larger sets is a subset of $E^{rem}_*\cup E^{rem}_1\cup E^{rem}_2$.

    Phase ''Clean up'' uses one call to $\Oracle$ and get two vertex sets $\ingball(u,d/2),\outgball(u,d/2)$. If it enters Line \ref{LDDlineborder}, then instead of returning $E$ we simply return a partition of $V$ where each set of this partition is a single node of $V$. Then it makes two calls to $LowDiameterDecomposition$. Instead of returning $E^{rem}_{in}\cup E^{rem}_{in}\cup E^{rem}_1\cup E^{rem}_2$, we return the partition of $\Aout\backslash \Ain$, followed by $V\backslash (\Ain\cup\Aout)$, then the partition for $\Ain$. Edges connecting former sets to later sets are subsets of $E^{rem}_{in}\cup E^{rem}_{in}\cup E^{rem}_1\cup E^{rem}_2$.

    Now we consider $FindBalancedSet$. After each node sampled $d_i$, the algorithm uses binary search where it makes in total $O(\log n)$ calls to $\Oracle$ \footnote{although in the algorithm it is multisource, this can be done easily in quantum query model by adding a dummy source connecting to all the sources $v_j$ with an edge with weight $Max-d_j$ where $Max$ is a sufficiently large constant, and find all nodes with distance $Max$ to the single source}. Then we can return the set $\cup_{j\le k}\sgball(v_{j},d_j)$.

    The running time is bounded by~\Cref{lem:depthofrecursion}, where it says the depth of the recursion is $O(\log n)$, where the size of the calls in each layer sum up to a single original graph. It makes $\tOh{1}$ calls to $\Oracle$.
    
    For correctness, it is easy to see that edges connecting former sets to later sets (edges $(u,v)$ where $u\in V_i,v\in V_j$ with $i<j$) is a subset of the edges set returned by Algorithm~\ref{alg:lowdiamterdecomposition}. Thus, $\Pr[u\in V_i,v\in V_j,i<j]=O\left(\frac{w(e)\log^2n}{d}+\frac{1}{n^8}\right)$ follows from the correctness of Algorithm~\ref{alg:lowdiamterdecomposition}. Moreover, every $V_i$ has weak diameter at most $d$ because each $V_i$ can either be a singleton (returned by Line~\ref{LDDlineborder}), or the set $V\backslash (\Ain\cup \Aout)$ (returned by Line~\ref{return3}), in which case the diameter is bounded by $d$.
\subsection{Open Problems}
\label{subsec:lddopenproblems}
There are three directions in which we see our result may be improved or extended. First, our decomposition only gives a so-called weak-diameter guarantee where the diameter is measured according to the distances in the input graph instead of distances within the cluster. Second, our cutting probability is an $O(\log n)$-factor higher compared to what one can get in undirected graphs, and it would be nice to cut each edge with probability roughly $O(\frac{\log n}{d})$ instead of $O(\frac{\log^2 n}{d})$. Third, it would be nice to get a deterministic variant of the low-diameter decomposition where the total number of edges that are cut is small.

\section{Topological Sort with respect to SCC}
\label{sec:scc-top}

A subroutine used by by \cite{bernstein2022negative} for sequential negative-weight SSSP computes strongly connected components (SCCs), outputting them in a topological order.
While there are known efficient parallel algorithms for this \cite{schudy2008finding} which reduces the problem to non-negative-weight SSSP, and while our quantum implementation of negative-weight SSSP does not need to use such a subroutine, there has been a gap in the \congest model where we do need to call such a subroutine in our negative-weight SSSP algorithm.
We adapt the parallel framework of Schudy~\cite{schudy2008finding} to the \DM (which then compiles down to \congest) in this section, showing \Cref{lem:scc-topsort-congest}.

\subsection{Overview of Schudy's Framework}
Schudy's framework finds SCCs using $O(\log^2 n)$ calls to $\Oraclem$ (see \Cref{def:oraclem}, where it is first used for \LDD).

At a high level, the framework is based on that of Coppersmith \etal~\cite{coppersmith2003divide}; pick a random vertex $v_p$, which identifies an SCC (all vertices that reach and can be reached from $v_p$), and three topologically orderable recursive instances (the remaining vertices which (i) can reach $v_p$, (ii) are reachable from $v_p$, (iii) can neither reach nor are reachable from $v_p$).
The snag here is that the recursive instances may be large, and consequently the algorithm is inefficient in the worst case.
Schudy~\cite{schudy2008finding} fixes this by employing a more careful selection of recursive instances which shrink by a constant multiplicative factor (hence the algorithm has a logarithmic recursion depth) and still satisfy topological orderability (hence the algorithm is correct).

Algorithm~\ref{alg:scc-top} gives a model-independent overview of the framework, with the addition of a labelling $(r_v)_{v \in V}$ which identifies the SCC to which the vertices belong (\cite{schudy2008finding} instead outputs the SCCs as a topologically sorted list).
To help with assigning a valid labelling, our algorithm takes in two more arguments, $\ell$ and $N$.
Intuitively, $\ell$ and $N$ are used to define a valid range of labels that may be assigned to the SCCs of a recursive instance.
$\ell$ is roughly the total number of vertices that are in preceding recursive instances, and $N$ is the number of vertices in $G$, the graph in the top-level recursive instance.
These together define the range from $\ell$ to $\ell$ plus the number of vertices in the recursive instances, dilated by a factor of $N^2$ so that different SCCs get different labels with high probability.

\begin{algorithm}
    \newcommand{\nextnr}{\stepcounter{AlgoLine}\ShowLn}
    \caption{$(r_v)_{v\in V}\leftarrow$ \SccTop$(G=(V,E), \ell, N)$}
    \label{alg:scc-top}
    \KwData{
        \begin{itemize}
            \item Directed Graph $G=(V,E)$.
                Internally, we treat $G$ as a weighted graph with edge weights all $0$.
                Accordingly, $\outgball(v,0)$ contains all vertices that $v$ can reach and $\ingball(v,0)$ contains all vertices that can reach $v$.
            \item Integers $\ell, N$, which bound $r_v$.
                Think of $N$ as the number of vertices of the graph at the top layer of recursion.
        \end{itemize}
    }
    \KwResult{Ordering $(r_v)_{v\in V}$ of vertices such that
        \begin{itemize}
            \item with high probability, $r_u = r_v$ if and only iff $u$ and $v$ are in the same SCC;
            \item $r_u > r_v$ when $(u,v) \in E$ and $u$ and $v$ belong to different SCCs;
            \item $\ell \cdot N^2 \le r_v \le (\ell + n - 1) \cdot N^2$.
        \end{itemize}
    }
    Let $v_1,v_2,...,v_n$ be a uniformly random permutation of all vertices in $V$\;\label{ustl1}
    
    For $i\in[n]$, define $S_i=\cup_{j\le i}\outgball(v_j,0)$, i.e., all the vertices that are reachable from $v_1,...,v_i$.
    Let $p\in[n]$ be the smallest index such that the induced subgraph $G[S_p]=(S_p,E_p)$ satisfies $|S_p|+|E_p|\ge\frac{n+m}{2}$ (we can efficiently find $p$ using a binary search, where each iteration makes one call to $\Oraclem$ with the virtual source attached to $v_1, v_2, \ldots, v_i$ by $0$ weight edges)\;\label{ustl2}
    
    \BlankLine

    Let $A\leftarrow S_{p-1}$, $B\leftarrow\outgball(v_p,0)$, $C\leftarrow \ingball(v_p,0)\cap\outgball(v_p,0)$\;\label{ustl3}

    Let $R_1 \gets V\backslash (A\cup B)$, $R_2 \gets A \setminus B$, $R_3 \gets C$, $R_4 \gets B \setminus (A \cup C)$, $R_5 \gets (A \cap B) \setminus C$\;\label{ustl5}
    
    \tcp*[h]{\color{blue} Base case.}
    
    If $|C| = |V|$, let $r$ be an integer drawn uniformly at random from $[\ell \cdot N^2, (\ell + n - 1) \cdot N^2]$ and set $(r_v)_{v \in V} \gets r$. \textbf{Return} $(r_v)_{v \in V}$\;\label{ustl4}

    \tcp*[h]{\color{blue} Recursion.}
    
    Let $(r_v)_{v\in R_1}\leftarrow \SccTop(G[R_1], \ell + \sum_{j=2}^{5}|R_j|, N)$\; \label{ustl6}

    Let $(r_v)_{v\in R_2}\leftarrow \SccTop(G[R_2], \ell + \sum_{j=3}^{5}|R_j|, N)$\; \label{ustl7}

    Let $(r_v)_{v\in R_3}\leftarrow \SccTop(G[R_3], \ell + \sum_{j=4}^{5}|R_j|, N)$\; \label{ustl8}

    Let $(r_v)_{v\in R_4}\leftarrow \SccTop(G[R_4], \ell + |R_5|, N)$\; \label{ustl9}

    Let $(r_v)_{v\in R_5}\leftarrow \SccTop(G[R_5], \ell, N)$\; \label{ustl10}
    
    \BlankLine
    \textbf{Return} $(r_v)_{v \in V}$\; \label{ustl11}
\end{algorithm}

Before we show an efficient implementation of Algorithm~\ref{alg:scc-top} in the \DM, we need two results from \cite{schudy2008finding}.
The first asserts its correctness, and the second asserts its efficiency.

\begin{proposition}[Paraphrasing Claim~5 and Lemma~7 in \cite{schudy2008finding}]
\label{prop:schudy-edges}
    $R_1, R_2, R_3, R_4, R_5$ (discovered in Algorithm~\ref{alg:scc-top}) partition $V$, and there are no edges going from one set to a lower numbered set.
\end{proposition}

\begin{proposition}[Paraphrasing Lemmas~8 and 9, Section~5.2 in \cite{schudy2008finding}]
\label{prop:schudy-rec}
    The recursion depth of Algorithm~\ref{alg:scc-top} is $O(\log n)$ with high probability.
\end{proposition}

\begin{proposition}
\label{prop:schudy-fin}
    Using $O(\log^2 n)$ calls to $\Oraclem$, Algorithm~\ref{alg:scc-top} outputs a topological order with respect to strongly connected components.
    More specifically, it outputs a polynomially-bounded labelling $(r_v)_{v\in V}$ such that, with high probability
    \begin{enumerate}
        \item $r_u = r_v$ if and only if $u$ and $v$ are in the same strongly connected component;
        \item when the SCC that $u$ belongs to has an edge towards the SCC that $v$ belongs to, $r_u > r_v$.
    \end{enumerate}
\end{proposition}
\begin{proof}
    Correctness of Algorithm~\ref{alg:scc-top} follows from observing that for any recursive instance, the labels assigned to its vertices are in the range $[\ell \cdot N^2, (\ell + n - 1) \cdot N^2]$ which is disjoint from and ordered with the ranges of other recursive instances in the same level by \Cref{prop:schudy-edges}.

    The number of calls to $\Oraclem$ being $O(\log^2 n)$ follows from \Cref{prop:schudy-rec} (we can cut off the algorithm and output \textsc{Fail} if the recursion depth is too large), and there being $O(\log n)$ calls to $\Oraclem$ in each recursive layer (we can run all executions of line~\ref{ustl2} in one recursive layer simultaneously).
\end{proof}

\subsection{Implementation in the \congest Model}
\label{subsec:scc-top-congest}

We are now ready to restate Algorithm~\ref{alg:scc-top} in the \DM, but with one crucial difference in our implementation.
Let $\mathsf{CC}(G[S])$ denote the connected components of $G[S]$, listed in any order.
Where Algorithm~\ref{alg:scc-top} recurses into $G[R_1], G[R_2], G[R_3], G[R_4], G[R_5]$ (in order), our implementation recurses into $\mathsf{CC}(G[R_1]), \mathsf{CC}(G[R_2]), \mathsf{CC}(G[R_3]), \mathsf{CC}(G[R_4]), \mathsf{CC}(G[R_5])$ (in order).
That is, one recursive call is made for each connected component.
This way, recursive instances in our implementation can correspond with connected subgraphs and, thus, super nodes in the \DM.

\begin{corollary}
\label{cor:topsort-impl}
    There is an algorithm on a directed graph $G = (V,E)$ which computes, with high probability, a ranking $(r_v)_{v\in V}$ of vertices such that, 
    \begin{itemize}
        \item $r_u = r_v$ if and only if $u$ and $v$ are in the same strongly connected component of $G$;
        \item if $u \in S$ and $v \in T$, where $S$ and $T$ are different strongly connected components of $G$, and there is an edge from $S$ to $T$, then $r_u > r_v$.
    \end{itemize}
    This algorithm takes $\tOh{1}$ rounds in the \DM and makes $\tOh{1}$ calls to $\Oraclem$.
\end{corollary}
\begin{proof}
    We describe Algorithm~\ref{alg:scc-top} using \DM steps.
    Let us focus on the implementation of one layer of recursion.
    Let $G[V_1], G[V_2], \ldots, G[V_k]$ be the recursive instances in this layer, listed in order.
    For now, suppose for all $i \in [k]$ that $G[V_i]$ is connected.

    We contract all edges with both endpoints in the same recursive instance to get a graph of super nodes.
    Let us now refine our focus to a particular super node $V_i$ at this layer of recursion.
    We subscript names in Algorithm~\ref{alg:scc-top} with $i$ to make this clear.

    \textit{Line~\ref{ustl1}. }
    Each vertex independently samples a uniform random number in $[n_i^3]$.
    Using a simple first moment method, it can be seen that with high probability no two vertices in $V_i$ sample the same number, which induces a uniformly random ordering of the vertices in $V_i$.

    \textit{Line~\ref{ustl2}. }
    We can then find $p_i$ via binary searching over $[n_i^3]$, with each iteration making a call to $\Oraclem(G, \cup_i S_i , 0)$ where $S_i$ is the set of all vertices in $V_i$ whose number is less than their respective binary search threshold. The binary searches in super nodes are independent of each other, but they coordinate one call to $\Oraclem$ to execute a threshold-check (see \Cref{rem:binarysearch} for a similar example).
    Edges joining different super nodes are taken to have infinite weight on this call to $\Oraclem$.
    The number of vertices plus edges of the graph induced by reachable vertices can be computed using an aggregation step.

    \textit{Lines~\ref{ustl3}, \ref{ustl4}, and \ref{ustl5}. }
    Next, the sets $A_i, B_i, C_i$ can be found with three calls to $\Oraclem$ (making sure to reverse edge directions for $C_i$).
    Now every vertex will know its membership in $(R_1)_i, (R_2)_i, (R_3)_i, (R_4)_i, (R_5)_i$ which partition $V_i$.
    Use an aggregation step to find $|(R_j)_i|$, and set the ranks of all vertices to $r$ if only $|(R_3)_i| > 0$; the random integer $r$ can be sampled by using an aggregation step to elect a leader, having the leader sample $r$, and using another aggregation step to broadcast $r$.

    \textit{Lines~\ref{ustl6} onwards. }
    Recall that in this implementation, recursive instances of the next layer are connected components $\mathsf{CC}(G[(R_j)_i])$ of $G[(R_j)_i]$.
    Each connected component of $G[(R_j)_i]$ sets its $\ell$ parameter as if it were in the recursive instance $G[(R_j)_i]$ (so there will only be five different $\ell$ parameters branching off from $G[V_i]$, even if there are much more than five connected components).
    The $\ell$ parameters can be found using the already computed values of $|(R_j)_i|$.
    
    Uncontract all super nodes and recurse down into the next layer.

    \textbf{Correctness. }
    By observing that each label is taken uniformly at random from an interval of length at least $|V|^2$, SCC labels are distinct with high probability (one may again use a first moment method to see this).

    If $S$ and $T$ are SCCs with an edge from $S$ to $T$, there must be a first time in the recursion that they are separated.
    That is, $S \subseteq R_{i^*}$ and $T \subseteq R_{j^*}$ for $i^* < j^*$ (the inequality comes from \Cref{prop:schudy-edges}) on some level of the recursion.
    Denote the $\ell$ parameters for $G[R_{i^*}]$ and $G[R_{j^*}]$ with $\ell_{R_{i^*}}$ and $\ell_{R_{j^*}}$ respectively.
    Then one can see that the intervals $[\ell_{R_{i^*}}, \ell_{R_{i^*}} + |R_{i^*}| - 1]$ and $[\ell_{R_{j^*}}, \ell_{R_{j^*}} + |R_{j^*}| - 1]$ are disjoint and $\ell_{R_{i^*}} > \ell_{R_{j^*}}$.
    Thus, $r_u > r_v$ for $u \in S$ and $v \in T$.

    \textbf{Complexity. }
    By \Cref{prop:schudy-rec}, Algorithm~\ref{alg:scc-top} has $O(\log n)$ levels of recursion with high probability.
    Vertices can halt and output \textsc{Fail} after $\Theta(\log n)$ levels of recursion, and so this algorithm succeeds with high probability using $\tOh{1}$ rounds of the \DM.
    Each layer of recursion involves $O(\log n)$ calls to $\Oraclem$, running the binary searches together, and hence in total there are $\tOh{1}$ calls.
\end{proof}

\begin{remark}
\label{rem:scc-using-reach}
    Before concluding this section, it is worth noting that we may instead use an oracle for Reachability with a virtual source, instead of $\Oraclem$ which we use here for clarity and convenience.
    Recursive calls to Algorithm~\ref{alg:scc-top} on subgraphs $G' \subseteq G$ are currently made by setting weights of edges to $0$ or $1$ in accordance with the edge being present or not present in the subgraph; this way, the top-level graph $G$ is used as the communication network throughout and hence the round complexity remains in terms of $n$ and $D$.
    With reachability, we can continue to use $G$ as the communication network by simulating \SccTop on $G'$ with a copy of $V(G')$ where every vertex in $V(G')$ is connected to its copy, and every edge in $E(G) \setminus E(G')$ is in the copy of $V(G')$.
    Suffice to say, the construction shows that reachability in a subgraph $G' \subseteq G$ is as hard as reachability in a graph $G$.
\end{remark}

To conclude this section, we complete the proofs for \Cref{lem:scc-topsort-congest} and \Cref{cor:scc-topsort-congest}, which we restate here for convenience.

\scctopcongest*
\begin{proof}
    This follows immediately from \Cref{cor:topsort-impl} and \Cref{thm:minoraggregation}.
\end{proof}

\begin{corollary}
\label{cor:scc-topsort-congest}
    There is a \congest algorithm that, given a directed graph $G=(V,E)$, outputs SCCs in topological order (same conditions as \Cref{lem:scc-topsort-congest}) within $\tOh{n^{1/2}+D+n^{2/5+o(1)}D^{2/5}}$ rounds. 
\end{corollary}
\begin{proof}
    This follows immediately from \Cref{lem:scc-topsort-congest} and \Cref{thm:distributedNNSSSP}.
\end{proof}

\section{The Framework}
\label{sec:framework}
In this section, we outline the crucial parts of our negative-weight SSSP algorithm, and how they fit together.
First, we define a basic shortest-path oracle for non-negative weights that all of our subroutines have access to. 
Then, we describe the computational models in which our algorithms are implemented.
Then we specify the formal inputs/outputs of the key subroutines needed by our negative-weight SSSP algorithm, leaving their implementation details for later sections. Finally, we give a description and full pseudocode for the two main algorithms, in order to illustrate how all of our subroutines interact with each other.
The goal of this section is to help the reader broadly understand our negative-weight SSSP algorithm and its interface.

\subsection{Shortest Path Oracles}
Recall that our main result is a reduction from negative-weight SSSP to shortest paths with non-negative weights. Thus, throughout the paper, we assume that we have access to an oracle for the latter, both in the parallel and distributed models.

\begin{definition}[$(\dist_{G}(s,v))_{v\in V}\leftarrow \Oracle(G=(V,E,w),s)$]\label{def:oracle}
    The non-negative single source shortest path oracle $\Oracle$ takes inputs (i) A directed graph $G=(V, E,w)$ with non-negative polynomially bounded integer edge weights (ii) A vertex $s \in V$, and returns the distance from $s$ to all vertices in $V$. It outputs shortest distances from $s$ to every $v \in V$. (In the \congest model, every vertex $v$ learns distance $\dist(s,v)$.)
\end{definition}

\paragraph*{Running times. } When we discuss the running time of a subroutine, we will simply specify the total number of calls it makes to $\Oracle$, since this will always be the dominant term. The final goal is an algorithm that makes $n^{o(1)}$ calls to $\Oracle$. 

Some of the subroutines that follow actually require oracles with a few additional minor technical properties, which we go over in their respective sections.

\subsection{Formal Input and Output Guarantees of Subroutines}
\label{sec:subroutine-guarantees}
Recall from Section~\ref{sec:overview} 
that almost all of the work happens in the \scaledown algorithm. The \scaledown algorithm in turn uses several subroutines; we now formally define their input/output guarantees. 

\paragraph*{Low-Diameter Decomposition (LDD) (Details in \Cref{sec:lowediameterdecomposition}). } 
The first subroutine is $\LDD$
\begin{itemize}
    \item INPUT: An $n$-node $m$-edge, directed graph $G = (V, E, w)$ with non-negative integer edge weights and a positive integer $d$.
    \item OUTPUT: A subset of edges $\Eref\subseteq E$ satisfying:
    \begin{itemize}
        \item Each SCC of the subgraph $E \setminus \Eref$ has weak diameter at most $d$ in $G$, i.e. if $u,v$ are two vertices in the same SCC, then $\dist_G(u,v)\le d$ and $\dist_G(v,u)\le d$;
        \item For any $e\in E$, we have $\Pr[e\in \Eref]=O\left(\frac{w(e)\log^2n}{d}+\frac{1}{n^8}\right)$.
    \end{itemize}
    \item NUMBER OF $\Oracle$ CALLS: $\tOh{1}$.
\end{itemize}

\textbf{Remark:} In the \congest model, the input requirements translates to every vertex $v \in V$ knowing $n, m, d$, their own neighborhood $N(v)$, and the weights on their incident edges $w(v,x)$ for $x \in N(v)$; the output guarantee translates to every vertex $v \in V$ knowing the membership of all its incident edges in $\Eref$.

\paragraph*{SCCs and Their Topological Ordering (Details in \Cref{sec:scc-top}). }
The second subroutine, $\SccTop$, computes a topological sort of SCCs.

\begin{itemize}
    \item INPUT: An $n$-node $m$-edge, directed graph $G = (V, E)$.
    \item OUTPUT: A polynomially bounded ranking $(r_v)_{v\in V}$ of vertices such that, 
        \begin{itemize}
            \item $r_u = r_v$ if and only if $u$ and $v$ are in the same SCC of $G$;
            \item If $u \in S$ and $v \in T$, where $S$ and $T$ are different SCCs of $G$, and there is an edge from $S$ to $T$, then $r_u > r_v$.
        \end{itemize}
    \item NUMBER OF $\Oracle$ CALLS: $\tOh{1}$.
\end{itemize}

\textbf{Remark:} In the \congest model, the input requirement translates to every vertex $v \in V$ knowing $n, m$, and their own neighborhood $N(v)$; the output guarantee translates to every vertex $v \in V$ knowing $r_v$.

\paragraph*{$\FixDag$ (Details in \Cref{sec:fix-dag}). }
The third subroutine is $\FixDag$, which computes a price that makes non-negative the edges between SCCs in a graph; this is used in Phase 2 of \scaledown.
\begin{itemize}
    \item INPUT: An $n$-node $m$-edge, directed graph $G = (V, E, w)$ with polynomially bounded non-negative edge weights on edges contained entirely inside an SCC, and a ranking  $(r_v)_{v\in V}$ that satisfies the output guarantee of $\SccTop$.
    \item OUTPUT: A polynomially bounded price function $\psi$ such that $G_{\psi}$ has non-negative edge weights for every edge.
    \item NUMBER OF $\Oracle$ CALLS: $0$.
\end{itemize}
\textbf{Remark:} In the \congest model, the input requirement translates to every vertex $v \in V$ knowing $n,m$, their own rank $r_v$, their own neighborhood $N(v)$, and the weights on their incident edges $w(v,x)$ for $x \in N(v)$; the output guarantee translates to every vertex $v \in V$ knowing the value of $\psi(v)$.

\paragraph*{$\EstDist$ (Details in \Cref{section:estdist}). }
The fourth subroutine is $\EstDist$.
This is used to return a distance estimate of single source shortest paths from source $s$ that is accurate for any vertex $v$ for which the shortest $sv$ path contains few negative-weight edges. This is used to compute the price function of Phase 3 of \scaledown.

\begin{itemize}
    \item INPUT: An $n$-node $m$-edge, directed graph $G = (V, E, w)$ with polynomially bounded integer edge weights, a source vertex $s$, and an accuracy parameter $h$.
    \item OUTPUT: A distance estimate $\tilde{d}$ such that for all $v \in V$ we have $\tilde{d}(v) \ge \dist_G(s,v)$, with equality if $\eta_G(v) \leq h$. (See \Cref{def:eta} for  $\eta_G(v)$.)
    \item NUMBER OF $\Oracle$ CALLS: $O(h)$.
\end{itemize}

\textbf{Remark:} In the \congest model, the input requirement translates to every vertex $v \in V$ knowing $n,m,s,h$, their own neighborhood $N(v)$, and the weights on their incident edges $w(v,x)$ for $x \in N(v)$; the output guarantee translates to every vertex $v \in V$ knowing the value of $\tilde{d}(v)$.

\textbf{Remark:} There is an important distinction between $\EstDist$ used here, and its analog $\ElimNeg$ in the sequential algorithm of \cite{bernstein2022negative}.
In \cite{bernstein2022negative}, instead of merely returning distance estimates, they are able to directly return the true distance to all vertices in Phase 3.
As a result, their algorithm $\ElimNeg$ immediately yields a price function that renders all edges non-negative (by \Cref{lem:prelim_nonnegative}); by contrast, our algorithm $\EstDist$ does not yield such a price function because some of the distance estimates are inaccurate. At a high level, the running time of their $\ElimNeg$  is $O(\log n \paren{n + \sum_{v \in V} \eta_G(v)})$, which scales with the \emph{average} number of negative edges along shortest paths.
On the other hand, our $\EstDist$ would need $O(\max_{v \in V} \eta_G(v))$ calls to $\Oracle$ to fix all non-negative edges, which scales with the \emph{worst case} number of negative edges along shortest paths.
This can be prohibitively large, so we instead limit $\EstDist$ to making $O(h)$ calls to $\Oracle$, and only guarantee accurate distance estimates for vertices $v$ with $\eta_G(v) \leq h$.

\subsection{The \scaledown algorithm}
\label{sec:scaledown-high-lvl}

We now discuss the algorithm \scaledown, which is the primary component of the overall algorithm.  The main guarantee of this procedure is to take a graph $G$ where all edge weights are $\geq -2B$ and output a price function $\phi$ such that the edge weights of $G_{\phi}$ are all $\geq -B$. For the sake of efficiency, the algorithm follows a recursive structure, where each step in the recursion reduces the number of negative edges on shortest paths from the dummy source; this number of negative edges is capture by the parameter $\Delta$ in the input-output guarantees below. 

For clarity, the remainder of this section assumes that the input graph does not contain a negative cycle.

\begin{itemize}
    \item INPUT: An $n$-node $m$-edge, directed graph $G = (V, E, w)$ with polynomially bounded integer edge weights, a parameter $B$ such that $w(e) \ge -2B$ for all edges $e \in E$, and a parameter $\Delta \le n$ such that $\eta(G^B) \le \Delta$. 
    \item OUTPUT: A polynomially bounded price function $\phi$ such that $w_\phi(e) \ge -B$ for all $e \in E$.
    \item NUMBER OF $\Oracle$ CALLS (this includes all subcalls by subroutines): $n^{o(1)}$.
\end{itemize}
Note that $\Delta = n$ always satisfies the input guarantee, so the initial call to \scaledown simply sets $\Delta = n$.

\subsubsection{Description of ScaleDown}
The algorithm \scaledown makes use of all the subroutines from Section \ref{sec:subroutine-guarantees} above. To show how these subroutines fit together, we now give pseudocode and a description of the algorithm $\scaledown$. We leave the formal analysis for \Cref{sec:scaledown_analysis}, because much of it is similar to the sequential analysis of \cite{bernstein2022negative}.

\paragraph*{Model-independent pseudocode. }
Our pseudocode is written in a general form and does not refer to any specific model of computation (e.g. parallel or distributed). This is because all the technical work happens inside the subroutines of Section \ref{sec:subroutine-guarantees}, and we will discuss model-specific implementation of these subroutines in their respective sections. The remaining steps of \scaledown (aside from the subroutines) are straightforward to implement in any model. 

\begin{algorithm}
	\caption{Algorithm for $ScaleDown(G=(V,E,w),\Delta,B)$}\label{alg:scaledown}
        \If{$\Delta \leq 2^{\sqrt{\log n}}$}{
            $\tilde{d} \leftarrow \EstDist(G^B_s, s, 2^{\sqrt{\log n}})$ \\
            $\phi \leftarrow \tilde{d}$ except that the entry for $s$ is deleted\\
            \Return $\phi$
        }
	Let $G^B_{\geq 0} := (V,E,w^B_{\geq 0})$ where $w^B_{\geq 0}(e) := \max\{0,w^B(e)\}$ for all $e \in E$ \\
      \For{$i=1,\ldots,10\log(n)$}{ \label{line:scaledown-loop}
        //  \textcolor{blue}{Phase 0: Decompose $V$ to SCCs $V_1,V_2,\ldots$ with weak diameter $\big\lfloor \frac{\Delta}{2^{\sqrt{\log n}}}\big\rfloor B$ in $G$} \\

        $\Erem \leftarrow \LDD \left(G^B_{\geq 0},\lfloor \Delta/2^{\sqrt{\log n}} \rfloor B \right)$ \\ 
        $(r_v)_{v \in V} \leftarrow \text{SCC+Topsort}((V,E \setminus \Erem))$ \\ \label{line:scaledown-scc}
        Denote the SCCs (found using $(r_v)_{v \in V}$) of $G^B \backslash \Erem$ with $V_1, V_2, ...$ \\
        // (\Cref{lem:expectation_is_good}) If $\eta(G^B) \leq \Delta$, then for all $v \in V_i$, $E[|P_{G^B}(v) \cap \Erem|] = O(\log^2 (n) 2^{\sqrt{\log n}})$. \\

        // \textcolor{blue}{Phase 1: Make edges inside the SCCs $G^B[V_i]$ non-negative (whp)} \\
        Let $H = \cup_j G[V_j]$ \\
        // (\Cref{lem:scaledown_bounding_eta}) If $G$ has no negative-weight cycle, then $\eta(H^B) \leq \lfloor\Delta/2^{\sqrt{\log n}} \rfloor$.
        $\phi_1 \leftarrow ScaleDown \left(H,\lfloor \Delta /2^{\sqrt{\log n}}\rfloor,B \right)$ \\ \label{line:scaledown-recurse}

        // \textcolor{blue}{Phase 2: Make all edges in $G^B \setminus \Erem$ non-negative (whp)} \\
        $\psi_2 \leftarrow FixDAGEdges(G^B_{\phi_1} \setminus \Erem, r)$\\
        $\phi_2 \leftarrow \phi_1 + \psi_2$ \\

        // \textcolor{blue}{Phase 3: For each $v \in V$, $\tilde{d}^{(i)}(v) = \dist_{G^B_s}(s,v)$ with probability at least $1/2$}
        \\ $\phi'_2 \leftarrow \phi_2$ but we additionally define $\phi'_2(s) = 0$\\
        $\tilde{d}_3 \leftarrow \EstDist((G^B_s)_{\phi'_2}, s, h)$ for $h = O(\log^2(n) 2^{\sqrt{\log n}})$ being sufficiently large \\
        $\tilde{d}^{(i)} := \tilde{d}_3 + \phi_2$ except that the entry for $s$ is deleted\label{line:scaledown-di} \quad \quad \quad  \quad \text{// $\tilde{d}^{(i)}(v) = \dist_{G_s^B}(s,v)$ with probability at least $1/2$}
      }
      For each node $v\in V$, let $\phi(v) = \min_{i = 1,\ldots,10\log(n)}\tilde{d}^{(i)}(v)$\label{line:scaledown-dii} \quad \quad \quad \text{// $\phi(v) = \dist_{G_s^B}(s,v)$ with probability at least $1- n^{-10}$}\\
      \Return $\phi$ 
\end{algorithm}

\paragraph*{Base Case: $\Delta \leq 2^{\sqrt{\log n}}$. }
Recall that the input has to satisfy $\eta(G^B) \leq \Delta$. Therefore, by definition, for each vertex $v \in V$, there exists a shortest $sv$-path in $G^B_s$ using at most $\Delta \leq 2^{\sqrt{\log n}}$ negative edges. We can therefore directly compute $\dist_{G^B_s}(s,v)$ for every vertex $v \in V$ by calling $\EstDist(G^B_s,s,2^{\sqrt{\log n}})$. Applying Lemma \ref{lem:prelim_nonnegative} then yields a price function $\phi$ with $w^B_\phi(e) \geq 0$ for all $e \in E$, which in turn implies $w_\phi(e) \geq -B$, as desired.

\paragraph*{Recursive Case: $\Delta > 2^{\sqrt{\log n}}$. }
Algorithm \ref{alg:scaledown} runs in $10 \log(n)$ iterations (Line~\ref{line:scaledown-loop}). In each iteration $i$, a distance estimate $\tilde{d}^{(i)}(v)$ is computed for all $v \in V$. The distance estimate satisfies that $\tilde{d}^{(i)}(v) \geq \dist_{G^B_s}(s,v)$ for all $v \in V$ and $\tilde{d}^{(i)}(v) = \dist_{G^B_s}(s,v)$ with probability at least $1/2$. Therefore, $\tilde{d}^{(i)}(v) = \dist_{G^B_s}(s,v)$ for some $i$ with probability at least $1 - (1/2)^{10 \log n} = 1 -n^{-10}$. Each iteration involves a recursive call; see \Cref{fig:old-rec,fig:new-rec} for illustrations comparing the recursion structure of \cite{bernstein2022negative} with that of our algorithm.

\begin{figure}
    \centering
    \begin{minipage}{.5\columnwidth}
        \includegraphics[width=1\columnwidth]{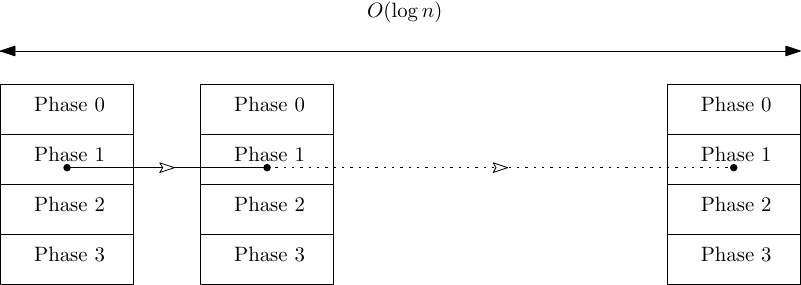}
        \captionof{figure}{Recursion structure of \scaledown in \cite{bernstein2022negative}.}\label{fig:old-rec}
    \end{minipage}%
    \begin{minipage}{.5\columnwidth}
        \includegraphics[width=1\columnwidth]{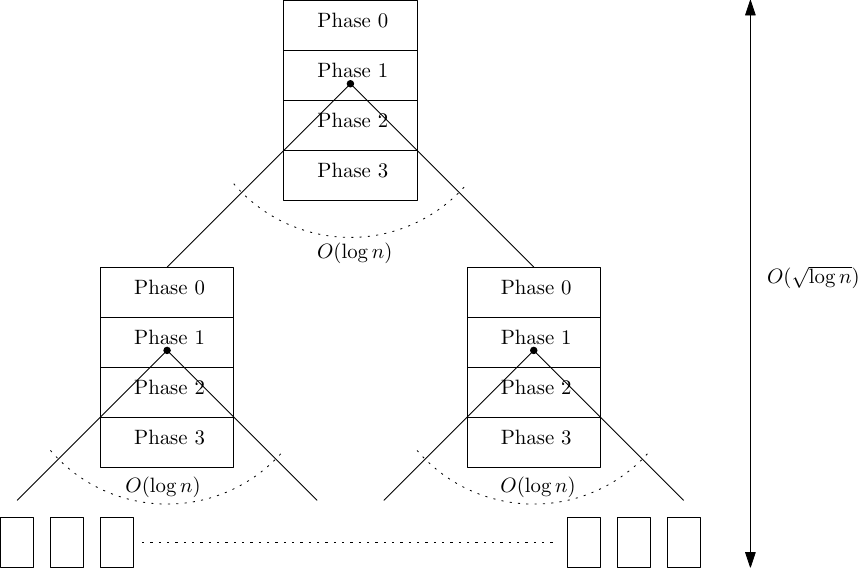}
        \captionof{figure}{Recursion structure of \scaledown in this paper.}\label{fig:new-rec}
    \end{minipage}
\end{figure}

\paragraph*{Description of an Iteration. }
In Phase $0$, we compute a low-diameter decomposition of $G^B_{\geq 0}$ by invoking 
$\Erem \leftarrow \LDD(G^B_{\geq 0}, \lfloor \Delta / 2^{\sqrt{\log n}}\rfloor B)$ and then do $V_1, V_2, \ldots \leftarrow SCC+TopSort(G^B \setminus \Erem)$; the decomposition guarantees that each $V_j$ has weak diameter at most $\lfloor \Delta/2^{\sqrt{\log n}}\rfloor B$ in $G$.

In Phase $1$, the goal is to compute a price function $\phi_1$ such that for any edge $e$ inside an SCC $V_j$ of $G \setminus \Erem$, we have $w^B_{\phi_1}(e) \geq 0$ with high probability. To do so, we recursively call $\scaledown(H,\lfloor \Delta/2^{\sqrt{\log n}}\rfloor,B)$, where $H$ is the union of all the subgraphs $G[V_j]$ induced by SCCs. The key step of the analysis is to show that this recursive call with a lower setting of $\Delta$ satisfies the input condition of $\scaledown$; that is, we need to show that if we consider any induced component $G[V_j]$ with dummy source $s$, then the shortest path from $s$ to any vertex in $G^B[V_j]$ contains at most $\lfloor \Delta/2^{\sqrt{\log n}}\rfloor$ negative edges. (Note that this is true only of shortest paths in $G^B[V_j]$, not in $G[V_j]$, which is the reason the algorithm uses $G^B$.) This claim was proved in the sequential algorithm of \cite{bernstein2022negative}; we reprove it in our formal analysis in \Cref{sec:scaledown_analysis}.

Next, in Phase $2$, we make all remaining edges in $G^B \setminus \Erem$ non-negative.
Observe that the edges inside each SCC $V_j$ are non-negative from Phase $1$, with high probability, and so the remaining negative edges will be among those connecting one SCC to another.
The subroutine \FixDag described above outputs, with high probability, a price function $\psi_2$ such that $w^B_{\phi_1 + \psi_2}(e) \geq 0$ for every edge $e \in E \setminus \Erem$ (\Cref{lem:scaledown_phase_2}).

Before describing Phase $3$, we observe the following. Let $\phi'_2$ be the price function after step $2$. We know that
all remaining negative edges in $G^B_{\phi'_2}$ are contained in $\Erem$.
We now take advantage of the fact that our LDD subroutine described above guarantees that each edge has a small probability of appearing in $\Erem$, so any given path has few edges from $\Erem$ in expectation. In particular, one can show that given any shortest path $P_{G^B}(v)$ from dummy source $s$ to some $v \in V$ in $(G^B_s)_{\phi'_2}$, we have that $E[|P_{G^B}(v) \cap \Erem|] = O(\log^2(n)2^{\sqrt{\log n}})$ (see \Cref{lem:expectation_is_good} in the formal analysis below.) This implies that with probability at least $1/2$ the number of negative edges on the path $P_{G^B}(v)$ in $(G^B_s)_{\phi'_2}$ is at most $h$ for a sufficiently large $h = O(\log^2 (n) 2^{\sqrt{\log n}})$. 

The key line of Phase 3 is computing distance estimates $\tilde{d}_3 \leftarrow \EstDist((G^B_s)_{\phi'_2},s,h)$, where $s$ is the dummy source and $h = O(\log^2 (n) 2^{\sqrt{\log n}})$. By the above discussion, for any $v \in V$, this distance estimate is correct with probability $\geq 1/2$.
\footnote{In actuality, the algorithm needs to make a small adjustment to $\tilde{d}_3$ to account for the price function of the dummy source (see Line~\ref{line:scaledown-di}); this is a minor technical detail that is handled in the formal analysis (see \Cref{sec:scaledown_analysis}).}

\paragraph*{Aggregating distance estimates between iterations. }
Recall that \scaledown loops over $10\log(n)$ independent and identical iterations (see Line~\ref{line:scaledown-loop}). As discussed above, the guarantee of iteration $i$ is that $\tilde{d}^{(i)}(v) = \dist_{G^B_s}(v)$ with probability at least $1/2$. Moreover, by the output guarantees of $\EstDist$ (\Cref{sec:subroutine-guarantees}), we always have that $\tilde{d}^{(i)}(v) \geq \dist_{G^B_s}(v)$. Thus, aggregating the $10\log(n)$ distance estimates, we clearly have that with high probability, for all $v \in V$, $\min_i \{\tilde{d}^{(i)}(v)\} = \dist_{G^B_s}(v)$. By \Cref{lem:prelim_nonnegative}, setting $\phi(v)$ to be this distance yields a price function such that all edge-weight are non-negative (see Line~\ref{line:scaledown-dii}).

We note that this need for $10\log(n)$ iterations is new to our paper and did not appear in the sequential \scaledown algorithm of \cite{bernstein2022negative}. In both our algorithm and theirs, once the algorithm gets to Phase 3, it is already the case that for any $v \in V$ the \emph{expected} number of negative edges on the shortest path from $s$ to $v$ is small. This is sufficient for the subroutine $\ElimNeg$ from the sequential algorithm to efficiently compute distances to \emph{all} vertices $v$. But our parallel/distributed subroutine $\EstDist$ can only compute distances to those vertices $v$ that satisfy this expected bound, so we need to aggregate estimates from $O(\log(n))$ iterations with independent randomness.

\paragraph*{Bounding Recursive Invocations. }
The need for $10\log(n)$ iterations leads to a slightly more complicated recursive analysis of the runtime, but the overall intuition is quite straightforward. The algorithm recursively calls itself $O(\log n)$ times in total, each time with parameter $\Delta_{rec} = \lfloor \Delta/2^{\sqrt{\log n}}\rfloor$. As $\Delta \leq n$, the recursion depth is $O(\sqrt{\log n})$ and thus the total number of recursive invocations is $O(\log n)^{O(\sqrt{\log n})} = 2^{O(\sqrt{\log n} \log \log n)}$. Ignoring the recursive calls, the non-negative SSSP oracle $\Oracle$ is called $O(2^{\sqrt{\log n}})$ times and therefore the total number of invocations to $\Oracle$ is $2^{O(\sqrt{\log n}\log \log n)}O(2^{\sqrt{\log n}}) = 2^{O(\sqrt{\log n}\log \log n)} = n^{o(1)}$, as desired. See \Cref{sec:scaledown_analysis} for the formal runtime analysis.

\subsection{Algorithm $\spmain$}
\paragraph*{$\spmain$ (Details in \Cref{sec:SPMain}). }
The algorithm $\spmain$ is a simple outer shell which uses $\scaledown$ to compute the final shortest distances. We take the algorithm $\spmain$ with essentially no modifications from \cite{bernstein2022negative}.
Pseudocode for $\spmain$ can be found in Algorithm~\ref{alg:SPMain}.

The input and output guarantees of $\spmain$ are as follows:
\begin{itemize}
    \item INPUT: An $n$-node $m$-edge, directed graph $G_{in} = (V, E, w)$ with polynomially bounded integer edge weights, and a source vertex $s_{in}$.
    \item OUTPUT: Error if $G_{in}$ contains a negative-weight cycle \footnote{There is a blackbox reduction in \cite{bernstein2022negative} (see Section~7 there) that extends this to a Las Vegas algorithm that reports a negative-weight cycle (instead of outputting Error)}.
        Otherwise, $\dist_{G_{in}}(s_{in}, v)$ for all $v \in V$.
    \item NUMBER OF $\Oracle$ CALLS (this includes all subcalls by subroutines): $n^{o(1)}$.
\end{itemize}

\subsubsection{Description of $\spmain$}

$\spmain$ first scales all the weight up by $2n$ to ensure that weights remain integral after calls to $\scaledown$. It then calls $\scaledown$ $O(\log n)$ times.
The calls to $\scaledown$ guarantee that the weights of the graph have been rescaled so that (i) shortest paths are unchanged, and (ii) edge weights are no smaller than $-1$.
Adding $+1$ to all weights thus yields a graph with non-negative edge weights, and we show that shortest paths are yet unchanged.
We can now invoke $\Oracle$ and unscale the weights to get distances from $s$. See \Cref{sec:SPMain} for the formal analysis. 

\begin{algorithm}
	\caption{Algorithm for $SPMain(G_{in} =(V,E,w_{in}),s_{in})$}\label{alg:SPMain}
    $\bar{w}(e) \leftarrow w_{in}(e) \cdot 2n$ for all $e \in E$, $\bar{G} \leftarrow (V,E,\bar{w})$, $B \leftarrow - \min_{e \in E} \bar{w}(e)$.  \\\label{line:spmain-find-b}
    Round $B$ up to nearest power of $2$ \\
    $\phi_0(v) = 0$ for all $v \in V$ \\
    \For{$i = 1$ to $t := \log_2(B)$}{
        $\psi_i \leftarrow ScaleDown((\bar{G})_{\phi_{i-1}},\Delta := n, B/2^i)$ \\\label{line:spmain-scaledown}
          \text{// \Cref{claim:SPMain_Gstar_nonnegative_weights}: $w_{\phi_i}(e) \geq -B/2^i$ for all $e \in E$  w.h.p. if $G$ does not contain a negative weight cycle} 
        $\phi_i \leftarrow \phi_{i-1} + \psi_i$    
            }
         $G^* \leftarrow (V,E,w^*)$ where $w^*(e) \leftarrow \bar{w}_{\phi_t}(e) + 1$ for all $e \in E$ \\
        // Observe: If $G_{in}$ does not contain a negative-weight cycle, then w.h.p. $G^*$ in above line has only strictly positive weights \\
        \If{$G^*$ contains a negative-weight edge}{ \label{line:spmain-neg-edge}
         \Return ERROR
        } 
        $d^* \leftarrow \Oracle(G^*,s_{in})$ \\\label{line:spmain-sssp}
        $d_{in}(v) \leftarrow \lfloor(d^*(v) - \phi_t(s_{in}) + \phi_t(v)) / (2n) \rfloor$ for all $v \in V$ \\
        \Return $d_{in}$
\end{algorithm}

\section{\FixDag}
\label{sec:fix-dag}

\subsection{Algorithm Overview and Analysis}
This section goes over \FixDag and a proof of \Cref{lem:fix-dag}.
The high level idea of \FixDag is very simple.
Let $G=(V,E,w)$ be a directed graph where edges contained in SCCs have non-negative weights, and let $(r_v)_{v \in V}$ be a labelling of vertices such that
\begin{enumerate}
    \item $r_u = r_v$ if and only if $u$ and $v$ are in the same strongly connected component;
    \item when the SCC that $u$ belongs to has an edge towards the SCC that $v$ belongs to, $r_u > r_v$.
\end{enumerate}
Finally, let $-B$ be the smallest (i.e. most negative) weight in $G$.
Then, we simply add a price $\psi(v)$ of $B \cdot r_v$ to every vertex.
Algorithm~\ref{alg:fix-dag-indep} formalizes this idea.
\begin{algorithm}
    \caption{$\psi(\cdot) \leftarrow \FixDag(G=(V,E,w), (r_v)_{v \in V})$}
    \label{alg:fix-dag-indep}
    \KwData{
        \begin{itemize}
            \item A weighted directed Graph $G=(V,E,w)$ where edges contained in SCCs have non-negative weights.
            \item A labelling $(r_v)_{v \in V}$ respecting a topological order of SCCs of $G$.
        \end{itemize}
    }
    \KwResult{Price function $\psi$ such that $G_\psi$ has non-negative weights.}
    \BlankLine
    $-B \gets \min(0,\min_{e \in E}(w(e)))$\;
    For each $v \in V$, let $\psi(v) \gets B \cdot r_v$\;
    Return $\psi$\;
\end{algorithm}

Finally, we are ready to assert the correctness of Algorithm~\ref{alg:fix-dag-indep}.
\begin{lemma}{\FixDag}
\label{lem:fix-dag}
    Let $G=(V,E,w)$ be a directed graph with polynomially-bounded integer weights, where for all $(u,v) \in E$ where $u$ and $v$ are in the same strongly connected component, $w(u,v) \ge 0$.
    Let $(r_v)_{v\in V}$ be a polynomially-bounded labelling which respects a topological ordering
    of SCCs of $G$.
    Given $G$ and $(r_v)_{v\in V}$, Algorithm~\ref{alg:fix-dag-indep} outputs a polynomially bounded price function $\psi : V \rightarrow \mathbb{Z}$ such that $w_{\psi}(u,v) \ge 0$ for all $(u,v) \in E$.
    Algorithm~\ref{alg:fix-dag-indep} makes no oracle calls.
\end{lemma}
\begin{proof}
    Suppose $(u,v) \in E$ is an edge contained in an SCC.
    Then $w_{\psi}(u,v) = w(u,v) \ge 0$ since $w(u,v) \ge 0$ and $r_u = r_v$.

    Suppose, on the other hand, $(u,v) \in E$ is an edge such $u$ and $v$ are in different SCCs.
    Then $w_{\psi}(u,v) \ge B \cdot (r_u - r_v - 1) \ge 0$ since $w(u,v) \ge -B$ and $r_u > r_v$.

    Finally, it is clear that Algorithm~\ref{alg:fix-dag-indep} does not make any call to $\Oracle$.
\end{proof}

\subsection{Implementation in Various Models}
We now show that there are efficient implementations of Algorithm~\ref{alg:fix-dag-indep} in both parallel and distributed models.

\paragraph*{Parallel implementation. }
A direct implementation of Algorithm~\ref{alg:fix-dag-indep} in the parallel model gives the following corollary.
\begin{corollary}
\label{cor:fix-dag-parallel}
    Let $G=(V,E,w)$ be a directed graph with polynomially-bounded integer weights, where for all $(u,v) \in E$ where $u$ and $v$ are in the same strongly connected component, $w(u,v) \ge 0$.
    Let $(r_v)_{v\in V}$ be a polynomially-bounded labelling which respects a topological ordering 
    of SCCs of $G$.
    Given $G$ and $(r_v)_{v\in V}$, there is a parallel algorithm that outputs a price function $\psi : V \rightarrow \mathbb{Z}$ such that $w_{\psi}(u,v) \ge 0$ for all $(u,v) \in E$ with $O(m)$ work and $O(1)$ span.
\end{corollary}

\paragraph*{Distributed implementation. }
Algorithm~\ref{alg:fix-dag-indep} takes one round in the \DM: contract the whole graph and use a consensus step to compute the minimum weight edge.
Then each vertex updates their price.
This leads to the following corollary.
\begin{corollary}
\label{cor:fix-dag-congest}
    Let $G=(V,E,w)$ be a directed graph with polynomially-bounded integer weights, where for all $(u,v) \in E$ where $u$ and $v$ are in the same strongly connected component, $w(u,v) \ge 0$.
    Let $(r_v)_{v\in V}$ be a polynomially-bounded labelling which respects a topological ordering
    of SCCs of $G$.
    Given $G$ and $(r_v)_{v\in V}$, there is an algorithm in the \DM that outputs a price function $\psi : V \rightarrow \mathbb{Z}$ such that $w_{\psi}(u,v) \ge 0$ for all $(u,v) \in E$ in one round.
\end{corollary}

\paragraph*{Quantum edge query implementation.}

Since~\Cref{lem:LDDquantum} will return $V_1,...,V_\ell$ (which does not nessesarily be SCCs of a graph), we need to change the input to Algorithm \ref{alg:fix-dag-indep} a bit, as follows.

\textbf{Inputs:} A weighted graph graph $G=(V,E,w)$ and a partition $V_1,...,V_\ell$ of $V$, where edges contained in each $V_i$ have non-negative weights, and there are no edges from $V_j$ to $V_i$ with $j>i$. 

The output is still a price function such that $G_\psi$ has no negative weight edges. For this, we can just set $\psi(v)\leftarrow B\cdot i$ for every $v\in V_i$. 
\section{\EstDist}
\label{section:estdist}

In this section, we provide the \EstDist algorithm (Algorithm \ref{alg:estdist}) and its proof of correctness.

\subsection{Algorithm Overview}

\EstDist is encapsulated in Algorithm~\ref{alg:estdist}.
At a high level, the algorithm takes $O(h)$ epochs, where $h$ is given as input.
Epoch $i$ computes a distance estimate $\tilde{d}_i$ such that $d(v) \le \tilde{d}_i(v)$, with equality if there is a shortest path from $s$ to $v$ with at most $i$ negative edges.
Each epoch consists of a Bellman-Ford step which relaxes negative edges once, and then an Oracle step for recomputing shortest paths in light of newly relaxed negative edges.

\begin{algorithm}
	\caption{Algorithm for $EstDist(G=(V,E,w),s,h)$}\label{alg:estdist}
         Let $H$ be the graph we obtain from $G$ by adding $B := \max(0,-(h+1) \cdot \min_{e \in E} w(e))$ to the weight of each outgoing edge from $s$.
         \\
        Let $H_0$ be the graph we obtain from $H$ by making all negative weight edges to be zero weight. \\
        $\tilde{d}_0 \leftarrow \Oracle(H_0,s)$ \\
        \For{$i = 1,2,\ldots,h$}{
            // \textcolor{blue}{Bellman-Ford step} \\
            $\tilde{d}^{(1)}_i(v) \leftarrow  \min(\tilde{d}_{i-1}(v),\min_{u \in V \setminus \{s\} \colon (u,v) \in E} \tilde{d}_{i-1}(u) + w(u,v))$ for every $v \in V$   \\
            Let $H_i$ be the graph we obtain from $H$ by making all negative-weight edges zero weight. Then, for every vertex $v \in V \setminus \{s\}$, add one edge from $s$ to $v$ with weight 
            $\tilde{d}^{(1)}_i(v)$.\\
            // \textcolor{blue}{Oracle step} \\
            $\tilde{d}_i \leftarrow \Oracle(H_i,s)$ 
            }
            $\tilde{d}(v) \leftarrow \tilde{d}_h(v) - B$ for every $v \in V \setminus \{s\}$ \\
            $\tilde{d}(s) \leftarrow 0$ \\
            \Return $\tilde{d}$
\end{algorithm}

\subsection{Analysis}
The main result of this section is the following lemma.
\begin{lemma}[\EstDist]\label{lem:estdist}
Let $G = (V,E,w)$ be a directed graph with polynomially-bounded integer weights, $s \in V$ and $h \in \mathbb{N}$. Assume that $dist_G(s,v) < \infty$ for all $v \in V$. Given $G$, $s$ and $h$ as input, Algorithm~\ref{alg:estdist} outputs a distance estimate $\tilde{d} \colon V \mapsto \mathbb{Z}$ such that for every $v \in V$, $\tilde{d}(v) \geq dist_G(s,v)$ and $\tilde{d}(v) = dist_G(s,v)$ if there exists a shortest path connecting $s$ and $v$ that contains at most $h$ negative edges. 
Moreover, Algorithm~\ref{alg:estdist} performs $h+1$ oracle calls to the non-negative weight distance oracle $\Oracle$.
\end{lemma}

\begin{proof}
First, \Cref{claim:estdist_nonnegative_Hi} verifies that we only give graphs with non-negative weight edges to the distance oracle $\Oracle$.

By using induction, we first show that for every $i \in \{0,1,\ldots,h\}$ and $v \in V \setminus \{s\}$, it holds that $\tilde{d}_i(v) \geq dist_H(s,v)$. The base case $i = 0$ trivially follows. Now, consider some $i \in \{1,2,\ldots,h\}$. Assume that $\tilde{d}_{i-1}(v) \geq dist_H(s,v)$ for every $v \in V \setminus \{s\}$. By using triangle inequality, we therefore get

\begin{align*}
\tilde{d}^{(1)}_i(v) &= \min(\tilde{d}_{i-1}(v),\min_{u \in V \setminus \{s\} \colon (u,v) \in E} \tilde{d}_{i-1}(u) + w(u,v)) \\
&\geq \min(dist_H(s,v),\min_{u \in V \colon (u,v) \in E} dist_H(s,u) + w(u,v)) \geq \dist_H(s,v)
\end{align*}
for every $v \in V \setminus \{s\}$.
Thus, it directly follows from the way $H_i$ is constructed that 
\[\tilde{d}_i(v) = dist_{H_i}(s,v)\geq dist_H(s,v),\]

as needed.

In particular, we get for every $v \in V \setminus \{s\}$ that

\[\tilde{d}(v) = \tilde{d}_h(v) - B \geq dist_H(s,v) - B \geq dist_G(s,v),\]

as needed. Also, it trivially follows that $\tilde{d}(s) = 0 \geq dist_G(s,s)$.

Next, let $v \in V \setminus \{s\}$ be a node such that there exists a shortest path $P$ connecting $s$ and $v$ with at most $h$ negative edges in $G$. By \Cref{claim:estdist_relate_H_to_G}, $P$ is also a shortest path connecting $s$ and $v$ in $H$.
We show that $\tilde{d}_h(v) = dist_H(s,v)$.
To do so, let $u$ be a node contained in $P$ such that there are at most $i$ negative edges between $s$ and $u$ on the path $P$ in $H$. We show that this implies $\tilde{d}_i(u) = dist_H(s,u)$ by induction on $i$.
The base case $i = 0$ trivially follows from the way we obtain $H_0$ from $H$.

Next, consider some fixed $i \in \{1,2,\ldots,h\}$ and assume that it holds for $i-1$.
Let $u$ be a node contained in $P$ such that there are at most $i$ negative edges between $s$ and $u$ on the path $P$ in $H$. 
If the number of negative edges is strictly less than $i$, then by induction we get $\tilde{d}_{i-1}(u) = dist_H(s,u)$ and therefore also $\tilde{d}_i(u) \leq \tilde{d}_i^{(1)}(u) \leq \tilde{d}_{i-1}(u) \leq dist_H(s,u)$. 
Now, assume that the number of negative edges is exactly $i$. Let $e = (x,y)$ be the $i$-th negative edge on the path $P$ in $H$, i.e., the last negative edge before $u$.
In particular, the number of negative edges on the path $P$ until $x$ is strictly less than $i$, and therefore we can use the induction hypothesis to conclude that $\tilde{d}_{i-1}(x) = dist_H(s,x)$. As $P$ is a shortest path in $H$, it holds that

\[\tilde{d}^{(1)}_i(y) \leq \tilde{d}_{i-1}(x) + w(x,y) = dist_H(s,y).\]

As $(x,y)$ is the last negative edge on the path $P$ before $u$, the whole path segment from $y$ to $u$ only consists of nonnegative edges and is therefore present in the graph $H_i$. As $P$ is a shortest path in $H$, the segment has a length of $dist_H(s,u) - dist_H(s,y)$. Therefore, we get

\[d_{H_i}(u) \leq \tilde{d}^{(1)}_i(y) + (dist_H(s,u) - dist_H(s,y)) = dist_H(s,u).\]

As we have shown above that $\tilde{d}_i(u) \geq dist_H(s,u)$, we therefore get $\tilde{d}_i(u) = dist_H(s,u)$, which finishes the induction.

In particular, for every $v \in V \setminus \{s\}$ such that there exists a shortest path $P$ connecting $s$ and $v$ with at most $h$ negative edges in $G$, we get that

\[\tilde{d}(v) = \tilde{d}_h(v) - B = dist_H(s,v) - B = dist_G(s,v),\]

where the last equality follows from \Cref{claim:estdist_relate_H_to_G}.

Finally, Algorithm \ref{alg:estdist} indeed calls the oracle $\Oracle$ $h+1$ times, which finishes the proof.
\end{proof}

\begin{claim}
    \label{claim:estdist_relate_H_to_G}
    For each vertex $v \in V \setminus \{s\}$ with $dist_G(s,v) > -\infty$, it holds that $dist_G(s,v) = dist_H(s,v) - B$.
\end{claim}
\begin{proof}
Consider some vertex $v \in V \setminus \{s\}$ with $dist_G(s,v) > -\infty$. As we additionally assume that $dist_G(s,v) < \infty$, there is a shortest path $P$ from $s$ to $v$ in $G$ of length $dist_G(s,v)$. This path $P$ has exactly one outgoing edge from $s$ and therefore is of length $dist_G(s,v) + B$ in $H$. Therefore, $dist_H(s,v) \leq dist_G(s,v) + B$. On the other hand, any $sv$-path in $H$ has at least one outgoing edge from $s$ and therefore the weight of this path in $G$ is at least smaller by an additive $B$. Hence, $dist_G(s,v) = dist_H(s,v) - B$, as needed.
\end{proof}

\begin{claim}
    \label{claim:estdist_nonnegative_Hi}
    For every $i \in \{0,1,\ldots,h\}$, $H_i$ only contains non-negative weight edges.
\end{claim}
\begin{proof}
We first prove by induction that for every $i \in \{0,1,\ldots,h\}$,

\[\min_{v \in V \setminus \{s\}}\tilde{d}_i(v) \geq \max(0,-(h-i)\min_{e \in E} w(e)) \geq 0.\]

We start with the base case $i = 0$.
From the way $H_0$ is defined, it follows that

\begin{equation*}
\label{eq:estdist_1}
\min_{v \in V \setminus \{s\}} \tilde{d}_0(v) \geq \min_{e \in E} w(e) + B = \min_{e \in E} w(e) +  \max(0,- (h+1) \cdot \min_{e \in E} w(e)) \geq \max(0,- (h-0) \cdot \min_{e \in E} w(e)).
\end{equation*}

Now, consider some fixed $i \in \{1,2,\ldots,h\}$ and assume that it holds for $i-1$. From the way $\tilde{d}^{(1)}_i$ and $H_i$ are defined and the induction hypothesis, we get

\begin{align*}
\min_{v \in V \setminus \{s\}} \tilde{d}_i(v) &\geq \min_{v \in V \setminus \{s\}} \tilde{d}^{(1)}_i(v) \geq \min_{v \in V \setminus \{s\}} \tilde{d}_{i-1}(v) + \min(0,\min_{e \in E} w(e)) \\
&\geq \max(0,- (h-(i-1)) \cdot \min_{e \in E} w(e)) + \min(0,\min_{e \in E} w(e)) \\
&= \max(0,- (h-i) \cdot \min_{e \in E} w(e)),
\end{align*}
which finishes the induction.
In particular, for every $i \in \{1,\ldots,h\}$,

\[\min_{v \in V \setminus \{s\}} \tilde{d}^{(1)}_i(v) \geq \max(0,- (h-i) \cdot \min_{e \in E} w(e)) \geq 0\]

and therefore $H_i$ only contains non-negative weight edges, as needed.
\end{proof}

\subsection{Implementation in Various Models}

\paragraph*{Parallel Implementation. } Algorithm~\ref{alg:estdist} is naturally parallelized; we have the following corollary in the parallel model.
\begin{corollary}
\label{cor:estdist-parallel}
Let $G = (V,E,w)$ be a directed graph with polynomially-bounded integer weights, $s \in V$ and $h \in \mathbb{N}$. Assume that $dist_G(s,v) < \infty$ for all $v \in V$ and there is a parallel algorithm answering SSSP in $W(m, n)$ work and $S(m, n)$ span.
Given $G$, $s$ and $h$ as input, then there is a parallel algorithm that outputs a distance estimate $\tilde{d} \colon V \mapsto \mathbb{Z}$ such that for every $v \in V$, $\tilde{d}(v) \geq dist_G(s,v)$ and $\tilde{d}(v) = dist_G(s,v)$ if there exists a shortest path connecting $s$ to $v$ that contains at most $h$ negative edges, with $O(W(m, n)h)$ work and $O(S(m,n)h)$ span.
\end{corollary}

\paragraph*{Distributed Implementation. }
Algorithm~\ref{alg:estdist} changes the weight of each edge and calls SSSP on the new graph. Each node can change its weight by itself in the \congest model. Each Bellman-Ford step can be accomplished in $1$ round and we have to call SSSP oracle $O(h)$ times. This leads to the following corollary.
\begin{corollary}
\label{cor:estdist-congest}
Let $G = (V,E,w)$ be a directed graph with polynomially-bounded integer weights, $s \in V$ and $h \in \mathbb{N}$. Assume that $dist_G(s,v) < \infty$ for all $v \in V$ and there is a \congest algorithm answering SSSP in $T(n, D)$ rounds.
Given $G$, $s$, and $h$ as input, then there is a \congest algorithm that outputs a distance estimate $\tilde{d} \colon V \mapsto \mathbb{Z}$ such that for every $v \in V$, $\tilde{d}(v) \geq dist_G(s,v)$ and $\tilde{d}(v) = dist_G(s,v)$ if there exists a shortest path connecting $s$ to $v$ that contains at most $h$ negative edges, with $T(n, D) h$ rounds.
\end{corollary}

\paragraph*{Quantum Query Implementation.} $\min_{e\in E}w(e)$ can be computed in $O(\sqrt{m})$ or $O(n)$ queries by using the following well-known minimum finding lemma. 

\begin{lemma}[\cite{ahuja1999quantum,durr1996quantum}]\label{lem:quantummin}
    Given quantum query access to a list of unordered items of length $x$, there is a quantum query algorithm that finds the minimum/maximum amongst the items with high probability by making $O(\sqrt{x})$ queries.
\end{lemma}

Every quantum matrix or list query in the updated $H$ or $H_0$ can be done by constant queries to the original $G$ (notice that $H$ adds some certerin weight to edges, and $H_0$ takes the maximum of $0$ and each edge as the final weight). When doing the Bellman-Ford step, for each node $v$ we use~\Cref{lem:quantummin} to update the distance in $\sqrt{deg_G(v)}$ or $\sqrt{n}$ queries. In total it is $\sum_{v\in V}\sqrt{deg_G(v)}\le \sqrt{mn}$ For each oracle step. We get the following corollary.

\begin{corollary}
\label{cor:quantumestdist}
Let $G = (V,E,w)$ be a directed graph with polynomially-bounded integer weights, $s \in V$ and $h \in \mathbb{N}$. Assume that $dist_G(s,v) < \infty$ for all $v \in V$ and there is a quantum edge query algorithm answering SSSP in $Q(m,n)$ rounds.
Given $G$, $s$, and $h$ as input, then there is a quantum query algorithm that outputs a distance estimate $\tilde{d} \colon V \mapsto \mathbb{Z}$ such that for every $v \in V$, $\tilde{d}(v) \geq dist_G(s,v)$ and $\tilde{d}(v) = dist_G(s,v)$ if there exists a shortest path connecting $s$ to $v$ that contains at most $h$ negative edges, with $\tOh{Q(m,n)}$ queries. 
\end{corollary}

\section{\scaledown}
\label{sec:scaledown_analysis}
In this section, we will give an analysis of the \scaledown function (see Algorithm~\ref{alg:scaledown}). Then combining the subroutines described in previous sections, we give the parallel and distributed implementation of \scaledown.

\subsection{Analysis}
The formal statement for the correctness and oracle-complexity of \scaledown is as follows.

\begin{theorem}[\scaledown]
    \label{thm:scaledown}
    Let $G = (V,E,w)$ be a weighted directed graph, $\Delta \leq n$ and $B \in \mathbb{N}$. The input has to satisfy that $w(e) \geq -2B$ for all $e \in E$. If the graph $G$ does not contain a negative cycle, then the input must also satisfy $\eta(G^B) \leq \Delta$; that is, for every $v \in V$ there is a shortest $sv$-path in $G^B_s$ with at most $\Delta$ negative edges (\Cref{def:GBs,def:eta}). \\
    Then, $ScaleDown(G,\Delta,B)$ returns a polynomially bounded potential $\phi$ such that if $G$ does not contain a negative cycle, then $w_\phi(e) \geq -B$ for all $e \in E$, with high probability. $ScaleDown(G,\Delta,B)$ calls the non-negative SSSP oracle $\Oracle{}$  $2^{O(\sqrt{\log n} \log \log n)}$ times.
\end{theorem}
\begin{proof}
It directly follows from \Cref{lem:scaledown_base_case} and \Cref{lem:scaledown_induction_step} that $\phi$ satisfies the conditions stated in \Cref{thm:scaledown}.
It remains to show that the negative-weight shortest path oracle $\Oracle$ is called $2^{O(\sqrt{\log n} \log \log n)}$ times in total.

We first upper bound the total number of recursive invocations of \scaledown.
As $\Delta \leq n$, the recursion depth is upper bounded by $O(\sqrt{\log n})$. As \scaledown recursively calls itself $O(\log n)$ times, the total number of calls is upper bounded by $\log(n)^{O(\sqrt{\log n})} = 2^{O(\sqrt{\log n} \log \log n)}$.

Next, we show that in a single call the total number of invocations to $\Oracle$ is upper bounded by $2^{O(\sqrt{\log n})}$.
\LDD calls $\Oracle$ for $\tOh{1}$ times (\Cref{lem:LDDcorrectness}).
The same holds for \SccTop (\Cref{prop:schudy-fin}), and \FixDag makes $0$ calls to $\Oracle$ (\Cref{lem:fix-dag}).
Finally, in the base case \EstDist makes $O(2^{\sqrt{\log n}})$ calls to $\Oracle$ and in Phase 3 it makes $O(\log^2(n)2^{\sqrt{\log n}}) = 2^{O(\sqrt{\log n})}$ calls (plugging appropriate values of $h$ into \Cref{lem:estdist}).
Hence, the total number of calls to $\Oracle$ is indeed upper bounded by $2^{O(\sqrt{\log n} \log \log n)}$.
\end{proof}

\begin{lemma}
\label{lem:scaledown_base_case}
If $\Delta \leq 2^{\sqrt{\log n}}$ and $G$ does not contain a negative weight cycle, then $w_{G_\phi}(e) \geq -B$ for every $e \in E$ and $\phi$ is polynomially bounded.
\end{lemma}
\begin{proof}
It follows from the output guarantees of \EstDist (see \Cref{lem:estdist}) that $\phi(v) = \tilde{d}(v) = dist_{G^B_s}(s,v)$ for every $v \in V$. Therefore, \Cref{lem:prelim_nonnegative} implies that for every $e \in E$, $w_{G^B_\phi}(e) \geq 0$ and therefore $w_{G_\phi}(e) \geq -B$, as desired. 
\end{proof}

The following lemmas come from Bernstein \etal~\cite{bernstein2022negative}.
\begin{lemma}[Lemma 4.3 of Bernstein \etal~\cite{bernstein2022negative}]
\label{lem:scaledown_distance_in_SCC}
For every $j$ and every $u,v \in V_j$, $dist_G(u,v) \leq \lfloor \Delta/2^{\sqrt{\log n}}\rfloor B$.
\end{lemma}

\begin{lemma}[Lemma 4.4 of Bernstein \etal~\cite{bernstein2022negative}]
\label{lem:expectation_is_good}
If $\eta(G^B) \leq \Delta$, then for every $v \in V$, $E[|P_{G^B}(v) \cap \Erem|] = O(\log^2 (n) 2^{\sqrt{\log n}})$.
\end{lemma}

\begin{lemma}[Lemma 4.5 of Bernstein \etal~\cite{bernstein2022negative}]
\label{lem:scaledown_bounding_eta}
If $G$ has no negative cycle, then $\eta(H^B) \leq \lfloor \Delta /2^{\sqrt{\log n}} \rfloor$.
\end{lemma}
During Phase 1, we perform the recursive call $ScaleDown(H,\lfloor \Delta/2^{\sqrt{\log n}}\rfloor, B)$. Lemma \ref{lem:scaledown_bounding_eta} implies the input satisfies the requirements of \scaledown (see \Cref{thm:scaledown}). 
Hence, we can assume by induction (because of the base case proven in \Cref{lem:scaledown_base_case}) that \scaledown outputs a price function $\phi_1$ satisfying the following: 
\begin{corollary}
\label{cor:scaledown_phase_1}
If $G$ has no negative-weight cycle, then all edges in $G^B_{\phi_1}[V_j]$ are non-negative for every $j$, with high probability.
\end{corollary}

\paragraph*{Phase 2: Make all edges in $G^B \setminus \Erem$ non-negative, with high probability. }

\begin{lemma}
\label{lem:scaledown_phase_2}
Assume that $G$ has no negative-weight cycle. Also, assume that all edge weights in $G^B_{\phi_1}[V_j]$ are non-negative and polynomially bounded for every $j$, which happens with high probability. Then, all edge weights in $G^B_{\phi_2} \setminus \Erem$ are non-negative and polynomially bounded with high probability.
\end{lemma}
\begin{proof}
\scaledown calls $\Psi_2 \leftarrow \FixDag(G^B_{\phi_1} \setminus \Erem,r)$.
As we assume that all edge weights in $G^B_{\phi_1}[V_j]$ are polynomially bounded for every $j$, it follows that for every $(u,v) \in E \setminus \Erem$ that $w_{G^B_{\phi_1}}(u,v) \geq 0$, which is the first input condition of \FixDag according to \Cref{lem:fix-dag}. The second condition is that $(r_{v})_{v \in V}$ is a polynomially-bounded labelling which respects a topological ordering of SCCs of $G^B_{\phi_1} \setminus \Erem$. It follows from setting $(r_v)_{v \in V} \leftarrow SCC+Topsort((V,E \setminus \Erem))$ and the output guarantees of \SccTop that this condition is satisfied with high probability.
If this condition is indeed satisfied, then the output guarantee of \FixDag in \Cref{lem:fix-dag} gives that all edge weights in $(G^B_{\phi_1} \setminus \Erem)_{\Psi_2} = G^B_{\phi_2} \setminus \Erem$ are non-negative and polynomially bounded, as desired. 
\end{proof}

\paragraph*{Phase 3: Compute $dist_{G^B_s}(s,v)$ for every $v$ with probability at least one half. }

\begin{lemma}
\label{lem:scaledown_phase3}
For every $v \in V$, it holds that $\tilde{d}^{(i)}(v) \geq dist_{G^B_s}(s,v)$. Moreover, if $G$ does not contain a negative cycle, then $\tilde{d}^{(i)}(v) = dist_{G^B_s}(s,v)$ with probability at least $1/2$.
\end{lemma}
\begin{proof}

For every $v \in V$, we have

\begin{align*}
\tilde{d}^{(i)}(v) &= \tilde{d}_3(v) + \phi_2(v) \\
&\geq dist_{(G^B_s)_{\phi'_2}}(s,v) + \phi_2(v) & \text{(\Cref{lem:estdist})}\\
&= dist_{G^B_s}(s,v) + \phi'_2(s) - \phi'_2(v) + \phi_2(v) \\
&= dist_{G^B_s}(s,v),
\end{align*}

which shows the first part of \Cref{lem:scaledown_phase3}. Moreover, the calculations above also imply that if $\tilde{d}_3(v) = dist_{(G^B_s)_{\phi'_2}}(s,v)$, then $\tilde{d}^{(i)}(v) = dist_{G^B_s(s,v)}$. 
We next show that if $G$ does not contain a negative weight cycle, then $\tilde{d}_3(v) = dist_{(G^B_s)_{\phi'_2}}(s,v)$ with probability at least $1/2$, which then shows the second part of \Cref{lem:scaledown_phase3}. 

According to \Cref{lem:estdist}, $\tilde{d}_3(v) = dist_{(G^B_s)_{\phi'_2}}(s,v)$ if there exists a shortest path connecting $s$ and $v$ in $(G^B_s)_{\phi'_2}$ with at most $h$ negative edges.

Recall that $P_{G^B}(v)$ is a shortest $sv$-path in $G^B_s$. As $G^B_s$ and $(G^B_s)_{\phi'_2}$ are equivalent according to \Cref{lem:prelim_eqivalence_price_function}, this implies that $P_{G^B}(v)$ is also a shortest $sv$-path in $(G^B_s)_{\phi'_2}$. It therefore suffices to show that $P_{G^B}(v)$ has at most $h$ negative edges in $(G^B_s)_{\phi'_2}$ with probability at least $1/2$. 
Combining \Cref{cor:scaledown_phase_1} and \Cref{lem:scaledown_phase_2}, we get that with high probability $\Eneg(G^B_{\phi_2}) \subseteq \Erem$, i.e. each negative edge in $G^B_{\phi_2}$ is contained in $\Erem$, with high probability. Therefore, each negative edge in $(G^B_s)_{\phi'_2}$ is either in $\Erem$ or an outgoing edge from $s$, with high probability. The path $P_{G^B}(v)$ contains exactly one outgoing edge from $s$. Therefore, if $\Eneg(G^B_{\phi_2}) \subseteq \Erem$ and $|P_{G^B}(v) \cap \Erem| \leq h-1$, then $P_{G^B}(v)$ contains at most $h$ negative edges. 
In \Cref{lem:expectation_is_good}, we have shown that $\E[|P_{G^B}(v) \cap \Erem|] = O(\log^2(n)2^{\sqrt{\log n}})$. Therefore, for $h = O(\log^2(n)2^{\sqrt{\log n}})$ being sufficiently large, it holds that $h-1 \geq 3\E[|P_{G^B}(v) \cap \Erem|]$ and therefore a simple Markov bound implies $Pr[|P_{G^B}(v) \cap \Erem| \geq h-1] \leq 1/3$.

Thus, we get

\[Pr[\tilde{d}_3(v) \neq dist_{(G^B_s)_{\phi'_2}}(s,v)] \leq Pr[\Eneg(G^B_{\phi_2}) \not \subseteq \Erem] + Pr[|P_{G^B}(v) \cap \Erem| \geq h-1] \leq 0.5,\]

as desired.
\end{proof}

\begin{lemma}
    \label{lem:scaledown_induction_step}
 If $\Delta > 2^{\sqrt{\log n}}$ and $G$ does not contain a negative weight cycle, then $w_{G_\phi}(e) \geq -B$ for every $e \in E$ and $\phi$ is polynomially bounded.
\end{lemma}
\begin{proof}
\Cref{lem:scaledown_phase3} together with setting $\phi = \min_{i=1,\ldots,10\log (n)} \tilde{d}^{(i)}$ implies that $\phi(v) = dist_{G^B_s}(s,v)$ for every $v \in V$ with high probability. If that's indeed the case, then \Cref{lem:prelim_nonnegative} implies that for every $e \in E$, $w_{G^B_\phi}(e) \geq 0$ and therefore $w_{G_\phi}(e) \geq -B$, as desired. 
\end{proof}

\subsection{Implementation in Various Models}
\label{sec:impl-scaledown}

\paragraph*{Parallel Implementation. }
\scaledown is naturally parallelized, given that its subroutines are parallelized.
\begin{corollary}
\label{cor:scaledown-parallel}
Let $G = (V,E,w), \Delta, B$ satisfy the input requirements for \scaledown. Assume that $dist_G(s,v) < \infty$ for all $v \in V$ and there is a parallel algorithm answering SSSP in $W(m, n)$ work and $S(m, n)$ span.
There is a parallel implementation of \scaledown that succeeds with $O(W(m, n)(\log n)^{O(\sqrt{\log n})})$ work and $\tilde{O}(S(m, n)2^{\sqrt{\log n}})$ span with high probability.
\end{corollary}
\begin{proof}
    Using \Cref{lem:LDDcorrectness} for \LDD, \Cref{prop:schudy-fin} for finding a topological ordering of SCCs, \Cref{cor:fix-dag-parallel} for \FixDag, and finally \Cref{cor:estdist-parallel} for \EstDist, parallel \scaledown calls $\Oracle$ $O(2^{\sqrt{\log n}})$ times.
    Note that when we recurse \scaledown on the new graph, the graph contains at most $O(m)$ edges and $O(n)$ vertex, so each subproblem calls $\Oracle$ at most $O(2^{\sqrt{\log n}})$ times. Each \scaledown calls $O(\log n)$ subproblems and the recursion depth is at most $O(\sqrt{\log n})$. 
    In total, \scaledown calls $\Oracle$ $2^{\sqrt{\log n}}\times (\log n)^{\sqrt{\log n}}$ times, and so it takes $W(m, n)(\log n)^{\sqrt{\log n}}$ work.
    
    For the span, although we need to run Phase 0 - Phase 3 $\log n$ times, we can run them simultaneously, and it only takes $S(m, n) 2^{\sqrt{\log n}}$ span for each level of recursion. The recursion depth is $O(\sqrt{\log n})$. Therefore, the span of \scaledown is $O(S(m, n) 2^{\sqrt{\log n}})$.
\end{proof}

\paragraph*{Distributed Implementation. }
The distributed implementation for \scaledown is slightly more complicated, and involves some subtlety for the following reason: \scaledown makes recursive calls to subgraphs (e.g. $E \setminus \Erem$), which do not necessarily have the same undirected hop diameter as the input graph; however, we want to bound the round complexity in terms of $D$.
\begin{corollary}
\label{cor:scaledown-dist}
Let $G = (V,E,w), \Delta, B$ satisfy the input requirements for \scaledown. Assume that $dist_G(s,v) < \infty$ for all $v \in V$ and there is a \congest algorithm answering SSSP in $T(n,D)$ rounds.
There is a \congest implementation of \scaledown that succeeds with $O((T(n,D) + \sqrt{n} + D)(\log n)^{O(\sqrt{\log n})})$ rounds with high probability.
\end{corollary}
\begin{proof}
    We implement \scaledown in the \DM.
    First, observe that subroutine calls on subgraphs (denoted here with $G'$) of $G$ in lines~\ref{line:scaledown-scc} and \ref{line:scaledown-recurse} can use $G$ as the communication network and hence we can measure the complexity of every line as if run on an $n$ vertex $D$ hop-diameter graph.
    \DM rounds on $G'$ can be straightforwardly simulated by $G$, and $\Oraclem$ calls on $G'$ can be run using $G$ by setting the weights of edges in $E(G) \setminus E(G')$ to be a sufficiently high polynomial in $n$ (which precludes them from being part of any shortest path).

    \textbf{Oracle calls in ScaleDown. }
    The number of calls to $\Oraclem$ follows directly from \Cref{thm:scaledown}.
    It remains to bound the number of \DM rounds.

    \textbf{\DM rounds in ScaleDown. }
    The base case, when $\Delta \le 2^{\sqrt{\log n}}$, only uses calls to $\Oraclem$ and makes up zero rounds.
    Let us hence focus on implementing just one iteration of the loop (line~\ref{line:scaledown-loop}).
    If we can show that this takes $\tOh{1}$ rounds, we are done since across all recursive instances there are $(\log n)^{O(\sqrt{\log n})}$ iterations.

    Computation of \LDD (i.e. $\Erem$) takes $\tOh{1}$ rounds, by \Cref{cor:lddcongestminor}.
    Similarly, computation of \SccTop (i.e. $(r_v)_{v \in V}$) takes $\tOh{1}$ rounds, by \Cref{cor:topsort-impl}.
    Computation of \FixDag (i.e. $\psi_2$) takes exactly $1$ round in the \DM by \Cref{cor:fix-dag-congest}.
    The remaining lines of the algorithm are all internal computations within vertices, or calls to $\Oraclem$, and have no bearing on the number of \DM rounds.
    In all, one iteration consequently takes $\tOh{1}$ rounds.

    To tie things up, one iteration of \scaledown takes $\tOh{1}$ rounds, there are $(\log n)^{O(\sqrt{\log n})}$ iterations, and $O(\log n)$ calls to \scaledown from which the number of \DM rounds is $(\log n)^{O(\sqrt{\log n})}$.
    Using \Cref{thm:minoraggregation} to get a \congest algorithm finishes up the proof.
\end{proof}

\paragraph*{Implementation in quantum query model.} 

We have the following corollary.

\begin{corollary}
\label{cor:quantumscaledown}
Let $G = (V,E,w), \Delta, B$ satisfy the input requirements for \scaledown. Assume that $dist_G(s,v) < \infty$ for all $v \in V$ and there is a quantum edge query algorithm answering SSSP in $Q(m,n)$ queries.
There is a quantum query implementation of \scaledown that succeeds with $Q(m,n)(\log n)^{O(\sqrt{\log n})}$ queries with high probability.
\end{corollary}

Now, we detail the implementation. If \(\Delta \leq 2^{\sqrt{\log n}}\), we invoke \EstDist, which, according to~\Cref{cor:quantumestdist}, has a fast implementation. It is important to note that every query to \(G^{B}_{\geq 0}\) can be realized with constant queries to \(G\). We now elaborate on each phase of the implementation.

\textbf{Phase 0:} This phase decomposes \(V\) into \(V_1, V_2, \ldots\), which are not necessarily the strongly connected components (SCCs) of a graph, as per~\Cref{lem:LDDquantum}. Instead of returning \(E^{rem}\), we directly return \(V_1, V_2, \ldots\), obviating the need for an SCC+Topsort algorithm.

\textbf{Phase 1:} A recursive call is made for the union of induced subgraphs of \(V_i\) (denoted as \(H\)). Each query to \(H\) can be implemented in constant queries to \(G\) as follows: for each edge \((u, v)\), we simply verify whether \(u, v\) are both in the same \(V_i\); if yes, then that edge is returned; otherwise, no edge is returned.

\textbf{Phase 2:} This phase involves calling \FixDag, which is feasible since we already possess \(V_1, V_2, \ldots\).

\textbf{Phase 3:} \EstDist is called, which, as indicated by~\Cref{cor:quantumestdist}, can be performed efficiently.

\paragraph*{}
\section{$\spmain$ (the Outer Shell)}
\label{sec:SPMain}

In this section, we finally give an analysis of the \spmain function (see Algorithm~\ref{alg:SPMain}), and complete the proofs of \Cref{thm:mainparallelreduction,thm:maindistributedreduction}.

\subsection{Analysis}

\begin{theorem}[\spmain]
\label{thm:spmain}
Let $G_{in} = (V,E,w_{in})$ be a directed graph with polynomially bounded integer edge weights and $s_{in} \in V$. Algorithm
\ref{alg:SPMain} takes as input $G_{in}$ and $s_{in}$ and has the following guarantee:

\begin{enumerate}
\item If $G_{in}$ has a negative-weight cycle, then Algorithm \ref{alg:SPMain} outputs ERROR, \\
\item If $G_{in}$ has no negative-weight cycle, then Algorithm \ref{alg:SPMain} computes $dist_G(s_{in},v)$ for every node $v \in V$ with high probability and otherwise it outputs ERROR.
\end{enumerate}
Algorithm \ref{alg:SPMain} invokes the non-negative weight SSSP oracle $\Oracle$ $n^{o(1)}$ times.
\end{theorem}
\begin{proof}
First, consider the case that $G_{in}$ has a negative-weight cycle.
As  we obtain $\bar{G}$ from $G_{in}$ by multiplying each edge weight by $2n$, this implies that there exists a cycle with weight at most $-2n$ in $\bar{G}$. Together with \Cref{lem:prelim_eqivalence_price_function}, this implies that $\bar{G}_{\phi_t}$ also contains a cycle with weight at most $-2n$. Thus, there exists an edge in $\bar{G}_{\phi_t}$ with weight at most $-2$ and this edge has a negative weight in $G^*$. Therefore, Algorithm \ref{alg:SPMain} indeed outputs ERROR. Next, assume that $G_{in}$ does not contain a negative-weight cycle. Then, according to \Cref{claim:SPMain_Gstar_nonnegative_weights}, the graph $G^*$ does not contain a negative-weight edge with high probability. If that's indeed the case, then the algorithm computes $dist_{G_{in}}(s,v)$ for every node $v \in V$ according to \Cref{claim:SPMain_correct_output}, as desired.
It remains to discuss the number of oracle call invocations. As we assume that the edge weights are polynomially bounded, Algorithm \ref{alg:SPMain} invokes ScaleDown $O(\log n)$ times. Each invocation calls the non-negative weight distance oracle $\Oracle$ $n^{o(1)}$ times according to \Cref{thm:scaledown}. Therefore, the total number of oracle invocations is indeed $n^{o(1)}$, which finishes the proof.
\end{proof}

\begin{claim}
    \label{claim:SPMain_correct_output}
    Assume that $G^*$ has non-negative weights. Then, $d_{in}(v) = dist_{G_{in}}(s_{in},v)$ for every vertex $v \in V$. 
\end{claim}
\begin{proof}
We have

\begin{align*}
d_{in}(v) &= \lfloor(d^*(v) - \phi_t(s_{in}) + \phi_t(v)) / (2n) \rfloor \\
&= \lfloor(dist_{G^*}(s_{in},v) - \phi_t(s_{in}) + \phi_t(v)) / (2n) \rfloor & \text{$d^*(v) = dist_{G^*}(s_{in},v)$}\\
&\geq \lfloor(dist_{\bar{G}_{\phi_t}}(s_{in},v) - \phi_t(s_{in}) + \phi_t(v)) / (2n) \rfloor & \text{$w^*(e) \geq \bar{w}_{\phi_t}(e)$ for all $e \in E$} \\
&= \lfloor dist_{\bar{G}}(s_{in},v) / (2n) \rfloor & \text{\Cref{lem:prelim_eqivalence_price_function}}\\
&= \lfloor dist_{G_{in}}(s_{in},v) \rfloor = dist_{G_{in}}(s_{in},v)
\end{align*}

and similarly,

\begin{align*}
d_{in}(v) &= \lfloor(dist_{G^*}(s_{in},v) - \phi_t(s_{in}) + \phi_t(v)) / (2n) \rfloor \\
&\leq \lfloor(dist_{\bar{G}_{\phi_t}}(s_{in},v) + n - \phi_t(s_{in}) + \phi_t(v)) / (2n) \rfloor \\
&= \lfloor (dist_{\bar{G}}(s_{in},v) + n) / (2n) \rfloor \\
&= \lfloor dist_{G_{in}}(s_{in},v) + 1/2 \rfloor = dist_{G_{in}}(s_{in},v).
\end{align*}
\end{proof}

\begin{claim}
    \label{claim:SPMain_Gstar_nonnegative_weights}
    Assume that $G_{in}$ does not contain a negative-weight cycle. Then, the following holds with high probability:
    For all $e \in E$ and $i \in [0,t:= \log_2(B)]$ we have that $\bar{w}_i$ is integral and that $\bar{w}_i(e) \geq -B/2^i$ for all $e \in E$. In particular, $\bar{w}_t(e) \geq -1$ for all $e \in E$ and therefore the graph $G^*$ has non-negative weights.
\end{claim}
\begin{proof}
We prove the claim by induction on $i$. The base case $i = 0$ directly follows from the way $B$ is defined.
Now, assume by induction that the claim holds for $\bar{G}_{\phi_{i-1}}$. The call to $ScaleDown(\bar{G}_{\phi_{i-1}},\Delta := n, B/2^i)$ satisfies the necessary input properties (see \Cref{thm:scaledown}) and in particular $\bar{G}_{\phi_{i-1}}$ does not contain a negative-weight cycle.

Thus, by the output guarantee of \scaledown we have that $(\bar{w}_{\phi_{i-1}})_{\psi_i}(e) \geq (B/2^{i-1})/2 = B/2^i$. The claim follows because as noted in \Cref{def:pricefunction}, $(\bar{w}_{\phi_{i-1}})_{\psi_i} = \bar{w}_{\phi_{i-1} + \psi_i} = \bar{w}_{\phi_i}$.
\end{proof}

\subsection{Implementation in Various Models (Completing Main Theorems)}

Finally, we are ready to wrap up our main results.
\mp*
\begin{proof}
    This follows from the parallel implementation of \scaledown (\Cref{cor:scaledown-parallel}), which is the non-trivial part of \spmain to implement.
\end{proof}
\md*
\begin{proof}
    This follows from the distributed implementation of \scaledown (\Cref{cor:scaledown-dist}), which is the non-trivial part of \spmain to implement.
\end{proof}

\mq*
\begin{proof}
    Notice that in Algorithm~\ref{alg:SPMain}, $B$ can be found in $\sqrt{mn}$ or $\sqrt{n^{1.5}}$ queries using~\Cref{lem:quantummin}. Also notice that every query to $\bar{G}$ can be implemented by constant queries to $G$. The same holds for $G^*$ as well. Other non-trivial parts of \spmain follow from the quantum implementation of \scaledown (\Cref{cor:quantumscaledown}), and one call to $\Oracle$.
\end{proof}

\bibliographystyle{plainurl}
\bibliography{local}

\begin{thebibliography}{10}

\bibitem{ahuja1999quantum}
Ashish Ahuja and Sanjiv Kapoor.
\newblock A quantum algorithm for finding the maximum.
\newblock {\em arXiv preprint quant-ph/9911082}, 1999.

\bibitem{Awerbuch89}
B.~Awerbuch, M.~Luby, A.V. Goldberg, and S.A. Plotkin.
\newblock Network decomposition and locality in distributed computation.
\newblock In {\em 30th Annual Symposium on Foundations of Computer Science},
  pages 364--369, 1989.
\newblock \href {https://doi.org/10.1109/SFCS.1989.63504}
  {\path{doi:10.1109/SFCS.1989.63504}}.

\bibitem{AMV20}
Kyriakos Axiotis, Aleksander Madry, and Adrian Vladu.
\newblock Circulation control for faster minimum cost flow in unit-capacity
  graphs.
\newblock In Sandy Irani, editor, {\em 61st {IEEE} Annual Symposium on
  Foundations of Computer Science, {FOCS} 2020, Durham, NC, USA, November
  16-19, 2020}, pages 93--104. {IEEE}, 2020.
\newblock \href {https://doi.org/10.1109/FOCS46700.2020.00018}
  {\path{doi:10.1109/FOCS46700.2020.00018}}.

\bibitem{bartal1996probabilistic}
Yair Bartal.
\newblock Probabilistic approximation of metric spaces and its algorithmic
  applications.
\newblock In {\em Proceedings of 37th Conference on Foundations of Computer
  Science}, pages 184--193. IEEE, 1996.

\bibitem{BeckerEL20}
Ruben Becker, Yuval Emek, and Christoph Lenzen.
\newblock Low diameter graph decompositions by approximate distance
  computation.
\newblock In {\em {ITCS}}, volume 151 of {\em LIPIcs}, pages 50:1--50:29.
  Schloss Dagstuhl - Leibniz-Zentrum f{\"{u}}r Informatik, 2020.

\bibitem{bernstein2020near}
Aaron Bernstein, Maximilian~Probst Gutenberg, and Christian Wulff-Nilsen.
\newblock Near-optimal decremental sssp in dense weighted digraphs.
\newblock In {\em 2020 IEEE 61st Annual Symposium on Foundations of Computer
  Science (FOCS)}, pages 1112--1122. IEEE, 2020.

\bibitem{bernstein2019distributed}
Aaron Bernstein and Danupon Nanongkai.
\newblock Distributed exact weighted all-pairs shortest paths in near-linear
  time.
\newblock In {\em Proceedings of the 51st Annual ACM SIGACT Symposium on Theory
  of Computing}, pages 334--342, 2019.

\bibitem{bernstein2022negative}
Aaron Bernstein, Danupon Nanongkai, and Christian Wulff-Nilsen.
\newblock Negative-weight single-source shortest paths in almost-linear time.
\newblock {\em arXiv preprint arXiv:2203.03456}, 2022.

\bibitem{bernstein2019decremental}
Aaron Bernstein, Maximilian Probst, and Christian Wulff-Nilsen.
\newblock Decremental strongly-connected components and single-source
  reachability in near-linear time.
\newblock In {\em Proceedings of the 51st Annual ACM SIGACT Symposium on theory
  of computing}, pages 365--376, 2019.

\bibitem{BerzinaDFLS04}
Aija Berzina, Andrej Dubrovsky, Rusins Freivalds, Lelde Lace, and Oksana
  Scegulnaja.
\newblock Quantum query complexity for some graph problems.
\newblock In {\em {SOFSEM}}, volume 2932 of {\em Lecture Notes in Computer
  Science}, pages 140--150. Springer, 2004.

\bibitem{bringmann2023negative}
Karl Bringmann, Alejandro Cassis, and Nick Fischer.
\newblock Negative-weight single-source shortest paths in near-linear time: Now
  faster!
\newblock {\em arXiv preprint arXiv:2304.05279}, 2023.

\bibitem{exactcf2023}
Nairen Cao and Jeremy Fineman.
\newblock Parallel exact shortest paths in almost linear work and square root
  depth.
\newblock In {\em {SODA}}. {SIAM}, 2023.

\bibitem{distributedhopsets}
Nairen Cao, Jeremy~T. Fineman, and Katina Russell.
\newblock Brief announcement: An improved distributed approximate single source
  shortest paths algorithm.
\newblock In {\em Proceedings of the 2021 ACM Symposium on Principles of
  Distributed Computing}, PODC'21, page 493–496, New York, NY, USA, 2021.
  Association for Computing Machinery.
\newblock \href {https://doi.org/10.1145/3465084.3467945}
  {\path{doi:10.1145/3465084.3467945}}.

\bibitem{cfrnagetivesssp}
Nairen Cao, Jeremy~T. Fineman, and Katina Russell.
\newblock Parallel shortest paths with negative edge weights.
\newblock In {\em Proceedings of the 34th ACM Symposium on Parallelism in
  Algorithms and Architectures}, SPAA '22, page 177–190, New York, NY, USA,
  2022. Association for Computing Machinery.
\newblock \href {https://doi.org/10.1145/3490148.3538583}
  {\path{doi:10.1145/3490148.3538583}}.

\bibitem{chechik2016decremental}
Shiri Chechik, Thomas~Dueholm Hansen, Giuseppe~F Italiano, Jakub
  {\L}{\k{a}}cki, and Nikos Parotsidis.
\newblock Decremental single-source reachability and strongly connected
  components in {O}$(m\sqrt{n})$ total update time.
\newblock In {\em 2016 IEEE 57th Annual Symposium on Foundations of Computer
  Science (FOCS)}, pages 315--324. IEEE, 2016.

\bibitem{cfl2022}
Li~Chen, Rasmus Kyng, Yang~P. Liu, Richard Peng, Maximilian~Probst Gutenberg,
  and Sushant Sachdeva.
\newblock Maximum flow and minimum-cost flow in almost-linear time.
\newblock In {\em 2022 IEEE 63rd Annual Symposium on Foundations of Computer
  Science (FOCS)}, pages 612--623, 2022.
\newblock \href {https://doi.org/10.1109/FOCS54457.2022.00064}
  {\path{doi:10.1109/FOCS54457.2022.00064}}.

\bibitem{CMSV}
Michael~B. Cohen, Aleksander M\k{a}dry, Piotr Sankowski, and Adrian Vladu.
\newblock Negative-weight shortest paths and unit capacity minimum cost flow in
  \~{O}($m^{10/7} \log w$ ) time: (extended abstract).
\newblock In {\em Proceedings of the Twenty-Eighth Annual ACM-SIAM Symposium on
  Discrete Algorithms}, SODA '17, page 752–771, USA, 2017. Society for
  Industrial and Applied Mathematics.

\bibitem{coppersmith2003divide}
Don Coppersmith, Lisa Fleischer, Bruce Hendrickson, and Ali Pinar.
\newblock A divide-and-conquer algorithm for identifying strongly connected
  components.
\newblock 2003.

\bibitem{DurrHHM06}
Christoph D{\"{u}}rr, Mark Heiligman, Peter H{\o}yer, and Mehdi Mhalla.
\newblock Quantum query complexity of some graph problems.
\newblock {\em {SIAM} J. Comput.}, 35(6):1310--1328, 2006.

\bibitem{durr1996quantum}
Christoph Durr and Peter Hoyer.
\newblock A quantum algorithm for finding the minimum.
\newblock {\em arXiv preprint quant-ph/9607014}, 1996.

\bibitem{ElkinNe16}
M.~Elkin and O.~Neiman.
\newblock Hopsets with constant hopbound, and applications to approximate
  shortest paths.
\newblock In {\em 57th Annual Symposium on Foundations of Computer Science
  (FOCS)}, pages 128--137, Los Alamitos, CA, USA, oct 2016. IEEE Computer
  Society.
\newblock URL: \url{https://doi.ieeecomputersociety.org/10.1109/FOCS.2016.22},
  \href {https://doi.org/10.1109/FOCS.2016.22}
  {\path{doi:10.1109/FOCS.2016.22}}.

\bibitem{fn18}
S.~{Forster} and D.~{Nanongkai}.
\newblock A faster distributed single-source shortest paths algorithm.
\newblock In {\em 2018 IEEE 59th Annual Symposium on Foundations of Computer
  Science (FOCS)}, pages 686--697, 2018.
\newblock \href {https://doi.org/10.1109/FOCS.2018.00071}
  {\path{doi:10.1109/FOCS.2018.00071}}.

\bibitem{distributednegativessp2021}
Sebastian Forster, Gramoz Goranci, Yang~P. Liu, Richard Peng, Xiaorui Sun, and
  Mingquan Ye.
\newblock Minor sparsifiers and the distributed laplacian paradigm.
\newblock In {\em 2021 IEEE 62nd Annual Symposium on Foundations of Computer
  Science (FOCS)}, pages 989--999, 2022.
\newblock \href {https://doi.org/10.1109/FOCS52979.2021.00099}
  {\path{doi:10.1109/FOCS52979.2021.00099}}.

\bibitem{Ghaffari23}
Mohsen Ghaffari, Christoph Grunau, Bernhard Haeupler, Saeed Ilchi, and Václav
  Rozhoň.
\newblock {\em Improved Distributed Network Decomposition, Hitting Sets, and
  Spanners, via Derandomization}, pages 2532--2566.
\newblock URL:
  \url{https://epubs.siam.org/doi/abs/10.1137/1.9781611977554.ch97}, \href
  {http://arxiv.org/abs/https://epubs.siam.org/doi/pdf/10.1137/1.9781611977554.ch97}
  {\path{arXiv:https://epubs.siam.org/doi/pdf/10.1137/1.9781611977554.ch97}},
  \href {https://doi.org/10.1137/1.9781611977554.ch97}
  {\path{doi:10.1137/1.9781611977554.ch97}}.

\bibitem{GKM17}
Mohsen Ghaffari, Fabian Kuhn, and Yannic Maus.
\newblock On the complexity of local distributed graph problems.
\newblock In {\em Proceedings of the 49th Annual ACM SIGACT Symposium on Theory
  of Computing}, STOC 2017, page 784–797, New York, NY, USA, 2017.
  Association for Computing Machinery.
\newblock \href {https://doi.org/10.1145/3055399.3055471}
  {\path{doi:10.1145/3055399.3055471}}.

\bibitem{0001Z22}
Mohsen Ghaffari and Goran Zuzic.
\newblock Universally-optimal distributed exact min-cut.
\newblock In {\em {PODC}}, pages 281--291. {ACM}, 2022.

\bibitem{Goldberg}
Andrew~V. Goldberg.
\newblock Scaling algorithms for the shortest paths problem.
\newblock {\em SIAM Journal on Computing}, 24(3):494--504, 1995.
\newblock \href
  {http://arxiv.org/abs/https://doi.org/10.1137/S0097539792231179}
  {\path{arXiv:https://doi.org/10.1137/S0097539792231179}}, \href
  {https://doi.org/10.1137/S0097539792231179}
  {\path{doi:10.1137/S0097539792231179}}.

\bibitem{jambulapati2019parallel}
Arun Jambulapati, Yang~P Liu, and Aaron Sidford.
\newblock Parallel reachability in almost linear work and square root depth.
\newblock In {\em 2019 IEEE 60th Annual Symposium on Foundations of Computer
  Science (FOCS)}, pages 1664--1686. IEEE, 2019.

\bibitem{Johnson77}
Donald~B. Johnson.
\newblock Efficient algorithms for shortest paths in sparse networks.
\newblock {\em J. ACM}, 24(1):1–13, jan 1977.
\newblock \href {https://doi.org/10.1145/321992.321993}
  {\path{doi:10.1145/321992.321993}}.

\bibitem{ks97}
Philip~N Klein and Sairam Subramanian.
\newblock A randomized parallel algorithm for single-source shortest paths.
\newblock {\em Journal of Algorithms}, 25(2):205 -- 220, 1997.
\newblock URL:
  \url{http://www.sciencedirect.com/science/article/pii/S0196677497908889},
  \href {https://doi.org/https://doi.org/10.1006/jagm.1997.0888}
  {\path{doi:https://doi.org/10.1006/jagm.1997.0888}}.

\bibitem{LinialS93}
Nathan Linial and Michael~E. Saks.
\newblock Low diameter graph decompositions.
\newblock {\em Comb.}, 13(4):441--454, 1993.
\newblock \href {https://doi.org/10.1007/BF01303516}
  {\path{doi:10.1007/BF01303516}}.

\bibitem{MillerPX13}
Gary~L. Miller, Richard Peng, and Shen~Chen Xu.
\newblock Parallel graph decompositions using random shifts.
\newblock In {\em {SPAA}}, pages 196--203. {ACM}, 2013.

\bibitem{peleg1999near}
David Peleg and Vitaly Rubinovich.
\newblock A near-tight lower bound on the time complexity of distributed mst
  construction.
\newblock In {\em 40th Annual Symposium on Foundations of Computer Science
  (Cat. No. 99CB37039)}, pages 253--261. IEEE, 1999.

\bibitem{RozhonEGH22}
V{\'{a}}clav Rozhon, Michael Elkin, Christoph Grunau, and Bernhard Haeupler.
\newblock Deterministic low-diameter decompositions for weighted graphs and
  distributed and parallel applications.
\newblock In {\em {FOCS}}, pages 1114--1121. {IEEE}, 2022.

\bibitem{rozhon2020polylogarithmic}
V{\'a}clav Rozho{\v{n}} and Mohsen Ghaffari.
\newblock Polylogarithmic-time deterministic network decomposition and
  distributed derandomization.
\newblock In {\em Proceedings of the 52nd Annual ACM SIGACT Symposium on Theory
  of Computing}, pages 350--363, 2020.

\bibitem{abs-2210-16351}
V{\'{a}}clav Rozhon, Bernhard Haeupler, Anders Martinsson, Christoph Grunau,
  and Goran Zuzic.
\newblock Parallel breadth-first search and exact shortest paths and stronger
  notions for approximate distances.
\newblock {\em CoRR}, abs/2210.16351, 2022.

\bibitem{REGH2022}
Václav Rozhoň, Michael Elkin, Christoph Grunau, and Bernhard Haeupler.
\newblock Deterministic low-diameter decompositions for weighted graphs and
  distributed and parallel applications.
\newblock In {\em 2022 IEEE 63rd Annual Symposium on Foundations of Computer
  Science (FOCS)}, pages 1114--1121, 2022.
\newblock \href {https://doi.org/10.1109/FOCS54457.2022.00107}
  {\path{doi:10.1109/FOCS54457.2022.00107}}.

\bibitem{rozhon2022undirected}
Václav Rozhoň, Christoph Grunau, Bernhard Haeupler, Goran Zuzic, and Jason
  Li.
\newblock Undirected $(1+\varepsilon)$-shortest paths via minor-aggregates:
  Near-optimal deterministic parallel distributed algorithms, 2022.
\newblock \href {http://arxiv.org/abs/2204.05874} {\path{arXiv:2204.05874}}.

\bibitem{schudy2008finding}
Warren Schudy.
\newblock Finding strongly connected components in parallel using o (log2 n)
  reachability queries.
\newblock In {\em Proceedings of the twentieth annual symposium on Parallelism
  in algorithms and architectures}, pages 146--151, 2008.

\bibitem{BLN20}
Jan van~den Brand, Yin~Tat Lee, Danupon Nanongkai, Richard Peng, Thatchaphol
  Saranurak, Aaron Sidford, Zhao Song, and Di~Wang.
\newblock Bipartite matching in nearly-linear time on moderately dense graphs.
\newblock In Sandy Irani, editor, {\em 61st {IEEE} Annual Symposium on
  Foundations of Computer Science, {FOCS} 2020, Durham, NC, USA, November
  16-19, 2020}, pages 919--930. {IEEE}, 2020.
\newblock \href {https://doi.org/10.1109/FOCS46700.2020.00090}
  {\path{doi:10.1109/FOCS46700.2020.00090}}.

\end{thebibliography}




\end{document}